\documentclass[aos]{imsart}
\RequirePackage{amsthm,amsmath,amsfonts,amssymb}
\RequirePackage[authoryear]{natbib} 
\RequirePackage[colorlinks,citecolor=blue,urlcolor=blue]{hyperref}
\RequirePackage{graphicx}

\usepackage{bbm}
\usepackage{graphicx,psfrag}
\usepackage{enumerate}
\usepackage{natbib}
\usepackage{url} 
\usepackage{bm}
\usepackage[bottom]{footmisc}
\usepackage{booktabs}
\usepackage{caption} 
\usepackage{color}
\usepackage{mathtools}
\usepackage{bbm}
\usepackage{multirow} 
\usepackage{multibib}
\usepackage{enumitem}
\usepackage[capitalize]{cleveref}

\usepackage[utf8]{inputenc}
\usepackage{amsmath}
\usepackage{amsfonts}
\usepackage[english]{babel}
\usepackage{amsthm}
\usepackage{amssymb}
\usepackage[ruled,vlined]{algorithm2e}
\usepackage{graphicx}
\usepackage{hyperref}
\usepackage{float}
\usepackage{epsfig}

\newcommand{\BLUE}[1]{#1}

\newtheorem{lemma}{Lemma}[section]
\newtheorem{remark}{Remark}[section]
\newtheorem{corollary}{Corollary}[section]
\newtheorem{proposition}{Proposition}[section]
\newtheorem{theorem}{Theorem}[section]

\newtheorem{definition}{Definition}[section]
\newtheorem{example}{Example}[section]

\newcommand{\Reals}{\mathbb{R}}
	 
\DeclareMathOperator*{\SP}{SP}
\DeclareMathOperator*{\GSP}{GSP}
\newcommand{\dee}{\mathrm{d}}
\def\[#1\]{\begin{align}#1\end{align}}
\def\*[#1\]{\begin{align*}#1\end{align*}}
\def\EE{\mathbb{E}}
\newcommand{\bigO}{\mathcal{O}}
\DeclareMathOperator*{\diag}{Diag}	
\DeclareMathOperator{\ESJD}{ESJD}
\DeclareMathOperator*{\var}{var}
\def\bone{\mathbf{1}}
\def\Ind{\bone}
\renewcommand{\Pr}{\mathbb{P}}

\def\citeps{\citep}
\def\citets{\citets}
\def\cites{\cite}
\def\TITLE{Stereographic Markov Chain Monte Carlo}

\begin{document}

\begin{frontmatter} 
\title{\TITLE{}}
\runtitle{Stereographic Markov Chain Monte Carlo}

\begin{aug} 
\author[A]{\fnms{Jun} \snm{Yang}\ead[label=e1]{jy@math.ku.dk}}, 
\author[B]{\fnms{Krzysztof} \snm{{\L}atuszyński}\ead[label=e2]{k.g.latuszynski@warwick.ac.uk}}, 
\and \author[B]{\fnms{Gareth O.} \snm{Roberts}\ead[label=e3]{gareth.o.roberts@warwick.ac.uk}}
\address[A]{Department of Mathematical Sciences, University of Copenhagen, Denmark.}
\address[B]{Department of Statistics,	University of Warwick, United Kingdom.}
\end{aug}

\begin{abstract} 
		High-dimensional distributions, especially those with heavy tails, are notoriously difficult for off-the-shelf MCMC samplers: the combination of unbounded state spaces, diminishing gradient information, and local moves results in empirically observed ``stickiness'' and poor theoretical mixing properties -- lack of geometric ergodicity. In this paper, we introduce a new class of MCMC samplers that map the original high-dimensional problem in Euclidean space onto a sphere and remedy these notorious mixing problems. In particular, we develop random-walk Metropolis type algorithms as well as versions of the Bouncy Particle Sampler that are uniformly ergodic for a large class of light and heavy-tailed distributions and also empirically exhibit rapid convergence in high dimensions. 
	In the best scenario, the proposed samplers can enjoy the ``blessings of dimensionality'' that \BLUE{the convergence is faster in higher dimensions.}
\end{abstract}

\begin{keyword}[class=MSC2020]
\kwd[Primary ]{60J20}
\kwd{60J25}
\kwd{65C05} 
\end{keyword}

\begin{keyword}
\kwd{random walk Metropolis}
\kwd{piecewise deterministic Markov processes}
\kwd{stereographic projection}
\kwd{uniform ergodicity}
\kwd{heavy tailed distributions}
\kwd{blessings of dimensionality}
\end{keyword}
\end{frontmatter}
 

    \section{Introduction}\label{section_intro}
    
    Bayesian analysis relies heavily on Markov chain Monte Carlo (MCMC) methods to explore complex posterior distributions. In most typical settings, such distributions have support contained within a subset $S$ of $\Reals^{d}$ for some $d>0$, and then it is natural to construct appropriate MCMC algorithms directly on $S$. In practice, this is how the vast majority of algorithms are constructed, although there are intrinsic problems with this approach. For instance, it is now well-established \citep{Mengersen1996} that the popular vanilla MCMC workhorse, the random-walk Metropolis (RWM) algorithm fails to be uniformly ergodic for any target density $\pi $ when $S$ is unbounded. All existing generic MCMC methods are built upon local proposal mechanisms and are similarly afflicted. Lack of uniform ergodicity results in sensitivity of the algorithm's convergence to its starting value, and potentially long burn-in periods. 
    
    Moreover, these convergence problems are exacerbated for target distributions with heavy (heavier than exponential) tails. For instance, in such cases RWM and the Metropolis-adjusted Langevin Algorithm (MALA) fail to be even geometrically ergodic (i.e., they converge at a rate slower than any geometric rate) \citep{Roberts1996,jarner2000geometric,roberts1996exponential}. In practice, this manifests itself on the algorithm trajectories by the presence of infrequent excursions of heavy-tailed duration into the target distribution tails.
    This can lead to further theoretical and practical problems, e.g., the absence of a Central Limit Theorem (CLT) for all $L^2$ functions, for example \citep{jarner2007convergence}, and instability of Monte Carlo estimators with large and difficult-to-quantify mean square errors which are highly sensitive to initial values.
    
    Important modern innovations in MCMC algorithms have come from Piecewise Deterministic Markov Processes (PDMPs) \citep{bernard2009event,bouchard2018bouncy,bierkens2019zig} which offer for the first time generic recipes for the construction of non-reversible MCMC and often yield substantial gains in algorithm efficiency as a result. Although not completely necessary, it turns out to be convenient and natural to construct PDMPs as continuous time algorithms. While the practical use of PDMPs for posterior exploration is still in its infancy, these methods offer substantial promise. However PDMPs are still localised algorithms and as such also suffer from the lack of uniform and/or geometric ergodicity on unfounded state spaces and/or heavy-tailed targets \citep{vasdekis2021note,vasdekis2021speed,andrieu2021subgeometric,andrieu2021hypocoercivity}.
   
   In contrast, when the target density tails are exponential or lighter, RWM, MALA, and related algorithms
   are generally geometrically ergodic under weak regularity conditions \citep{Roberts1996,jarner2000geometric}.
   Some transformation strategies for achieving this are described in \citep{sherlock2010random,johnson2012variable} with closely related strategies being proposed in \citep{kamatani2018efficient}, although these approaches are unlikely to regain uniform ergodicity. When $S$ is bounded, then MCMC algorithms are generally easily shown to be uniformly ergodic (see for example \citep{Mengersen1996} for RWM). This naturally suggests that transformations designed to compactify $S$ might facilitate the construction of more robust families of MCMC algorithms. However, finding generic solutions to the construction of such transformations which lead to well-behaved densities on the transformed space is challenging.
   
   It is, however, far easier to construct general transformations to transform $\Reals^d$ to a {\em pre-}compact space. The most celebrated of such transformations is \emph{stereographic projection} which was known to the ancient Egyptians, see \citep{coxeter1961introduction} for a modern description.
  It maps $\Reals^d\to \mathbb{S}^d\setminus N$ (where $N$ denotes the {\em North Pole}), i.e., the $d$-sphere excluding its North Pole.
   This paper will explore the use of stereographically projected algorithms. One iteration of such an algorithm
   takes the current state ${\bf x} \in  {\Reals}^d$, transforms to $\mathbb{S}^d\setminus N$ through inverse stereographic projection, carrying out an appropriately constructed MCMC step on $\mathbb{S}^d\setminus N$ before returning to $\Reals^d$ by stereographic projection. The contributions of our work are as follows.
   \begin{enumerate}
       \item 
       We present the Stereographic Projection Sampler (SPS), an efficient and practical implementation of the above programme for RWM. It employs a simple reprojection step to ensure the Markov chain remains on $\mathbb{S}^d\setminus N$. We provide a dimension and scale dependent recipe for choosing the radius of $S$. See \cref{subsec_SPS}.
       \item
        We prove that for continuous positive densities, $\pi $ on $\Reals^d$ with tails no heavier than those of a $d$-dimensional multivariate student's $t$ distribution with $d$ degrees of freedom, SPS is uniformly ergodic. The tail conditions on $\pi$ are in fact necessary for geometric ergodicity. See \cref{thm-uniform-ergodic,remark_uniform_ergodic}.
       \item
       We give a high-dimensional analysis of SPS for the stylised family of product i.i.d.\ form targets, by maximizing the expected squared jumping distance (ESJD) as well as by showing for a variant of SPS that each component converges (as $d\to \infty $) to a Langevin diffusion. This affords a direct and uniformly favorable comparison with the Euclidean RWM algorithm. See \cref{section_scaling}.
             \item
       We also introduce the Stereographic Bouncy Particle Sampler (SBPS) which replaces the random walk move in SPS with a PDMP algorithm that follows great circle trajectories interspersed with abrupt direction changes. See \cref{subsec_SBPS}.
       \item
       We prove the uniform ergodicity of SBPS under weaker conditions than those for SPS, specifically requiring tails to be lighter than those of a $d$-dimensional multivariate student's $t$ distribution with $d-1/2$ degrees of freedom. See \cref{thm_PDMP,coro_PDMP}.
       \item 
       \BLUE{For isotropic targets (i.e., spherically symmetric targets), both light- and heavy-tailed, we (informally) demonstrate that the proposed SPS and SBPS can enjoy a \emph{blessings of dimensionality} effect which
       implies that 
       SPS and SBPS converge arbitrarily faster in higher dimensions than their Euclidean analogues. See \cref{section_isotropic} and \cref{section_simu}.} 
       \item  
       We introduce generalisations of stereographic projection which are suitable for elliptical targets. See \cref{section_ellipse}. The framework we established in this paper opens opportunities for developing other mappings from $\Reals^d$ to $\mathbb{S}^d$ for other classes of targets. See \cref{section_discuss}. 
   \end{enumerate}

In very recent work, \citep{lie2021dimensionindependent} presents a related methodology for exploring distributions defined directly on a manifold, providing supporting theory showing that under suitable conditions their method's spectral gap is dimension-independent. However, the focus of their work is very different. Their work uses a different reprojection scheme based on \citep{cotter2013mcmc} and does not consider distributions defined directly on $\Reals^d$. Moreover, they consider densities that are absolutely continuous with respect to a Gaussian measure in the infinite-dimensional limit. This is arguably a very restrictive class. So their dimension-independent results are not really comparable with our results. An alternative reprojection scheme is introduced in \citep{zappa2018monte}. Compared to that method, the reprojection scheme used in this paper has the advantage that reprojection from a fixed point on the tangent space is always possible and is a much more natural approach for the hypersphere. The \citep{zappa2018monte} method does have advantages for use in more general manifolds, but this is not of any use to us here.

      There are strong theoretical reasons for wanting to construct the algorithm dynamics directly on a manifold of positive curvature such as a hypersphere \citep{mangoubi2018rapid,ollivier2009ricci,mijatovic2018projections}. However, quite naturally, the existing literature has concentrated primarily on the case of Brownian motion which clearly has the uniform invariant distribution on $S$, or on the geodesic walk designed specifically to target uniform distributions on manifolds.
      The uniform distribution on $S$ maps via stereographic projection to the student's $t$ distribution with $d$ degrees of freedom on $\Reals^d$. Therefore, we can immediately lift theory from existing theory to \BLUE{(informally) show that SPS on such target densities has a dimension-free convergence time.} 
      In our paper, we go further. Rapid convergence extends to a large family of spherically symmetric target densities (essentially excluding only very heavy-tailed distributions).

{\BLUE{
However the proposed algorithms are not only applicable for stylized classes of target distributions. To showcase the practical value of stereographic samplers in more realistic statistical contexts, 
we shall demonstrate the utility of our methodology on a Bayesian analysis of a Cauchy regression model. Full specification of the model, its parameter values and other details will be given in \cref{subsec_example_Cauchy}. 

\begin{example}\label{example_Cauchy}
Consider $Y_i\sim \textrm{Cauchy}(\alpha+\beta^TX_i,\gamma)$, where $\{X_i,Y_i\}_{i=1}^n$ are the design matrices and responses, respectively, $\alpha\in\Reals$, $\beta\in\Reals^{d-2}$, $\gamma\in\Reals^+$ are the parameters. Assume a flat prior for $\alpha$ and $\beta$, a $\text{Gamma}(a,b)$ prior for $\gamma$ (i.e., $\pi_0(\gamma)\propto \gamma^{a-1}\exp(-b\gamma)$),
the posterior can be written by
\[
\pi(\alpha,\beta,\gamma)\propto \frac{\gamma^{a-1-n}\exp(-b\gamma)}{\prod_{i=1}^n\left\{1+\left[\frac{Y_i-(\alpha+\beta^TX_i)}{\gamma}\right]^2\right\}}.
\]
For sampling the posterior, we compare SPS and RWM in \cref{fig-Cauchy-Regression}, where we plot the traceplots for $\alpha$, $\beta_1$ (the first coordinate of $\beta$), and $\log(\gamma)$, as well as $(\alpha,\beta_1)$ for both algorithms. From the figure, one can see clearly that: (1) RWM which is not geometric ergodic completely failed; (2) the proposed SPS which is uniformly ergodic converged extremely fast.
\begin{figure}
		\centering
\includegraphics[width=\textwidth]{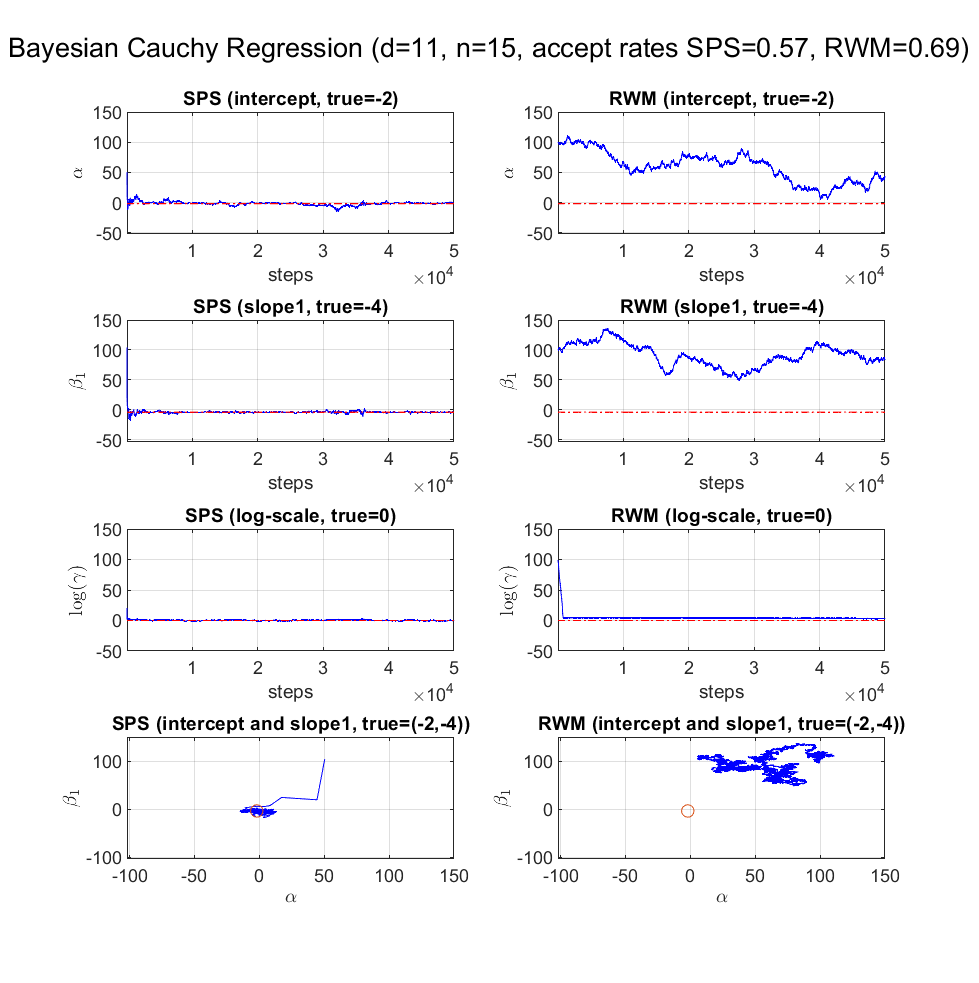}
		\caption{\BLUE{Bayesian Cauchy Regression. Red dotted lines and circles represent the parameters of the true sampling distribution.
 }}
		\label{fig-Cauchy-Regression}
	\end{figure}
\end{example}
} }

   Our paper is structured as follows. \cref{section_formulate} introduces both SPS and SBPS algorithms and presents formal ergodicity and uniform ergodicity results for both methods. \cref{section_isotropic,section_ellipse,section_scaling} provide additional results for SPS and its generalisations. In \cref{section_isotropic} a detailed analysis of SPS for isotropic densities is provided, while in \cref{section_ellipse}, a generalised version of the algorithm is introduced. \cref{section_ellipse} also gives robustness results to departures from isotropy. In \cref{section_scaling} we provide a high-dimensional analysis of  SPS for the stylised family of product i.i.d.\ target densities and give positive comparisons to the Euclidean RWM algorithm. Numerical studies on both SPS and SBPS are provided in \cref{section_simu} to illustrate our theory and we conclude with the discussions on SPS and SBPS in \cref{section_discuss}. 
   The supplementary materials contain all the technical proofs and some additional simulations. 

\section{Stereographic Markov Chain Monte Carlo}  \label{section_formulate}
	\subsection{Stereographic Projection Sampler (SPS)}	\label{subsec_SPS}
	
	We begin by giving a brief description of stereographic projection and then describe in detail two novel MCMC samplers that exploit the properties of stereographic projection.
	
	Let $\mathbb{S}^{d}$ denote the  \emph{unit sphere} in $\Reals^{d+1}$ centered at the origin.
A stereographic projection describes a bijection
from  $\mathbb{S}^{d}\setminus \{(0,\ldots 0,1)\}$ to $\Reals^d$. Within this paper, we shall restrict attention to projections indexed by a single parameter $R \in \Reals ^+$ and described by
the mapping
\*[
	x=\SP(z):=\left(R\frac{z_1}{1-z_{d+1}},\dots,R\frac{z_d}{1-z_{d+1}}\right)^T,
	\]
with Jacobian determinant at $x \in \Reals^d$ satisfying
\begin{equation}\label{eq_jacobian}
    J_{\SP}(x) \propto (R^2+ \|x\|^2)^{d},
\end{equation}
and inverse $\SP^{-1}: \Reals^d \to \mathbb{S}^{d}\setminus \{(0,\ldots 0,1)\}$ given by 
\[\label{eq_SP}
z_i = {2Rx_i \over \| x\|^2 + R^2}, \quad \forall 1\le i \le d,\qquad 
z_{d+1}={\| x\|^2 - R^2 \over \| x\|^2 + R^2}.
\]
See \cref{fig-stereo} for a geometric illustration of this stereographic projection
	(in the case $d=1$ for simplicity) and \cref{proof_jacobian} for the proof of the Jacobian determinant \cref{eq_jacobian}.


\begin{figure}
\centering
\begin{minipage}{.49\linewidth}
  \includegraphics[width=\linewidth]{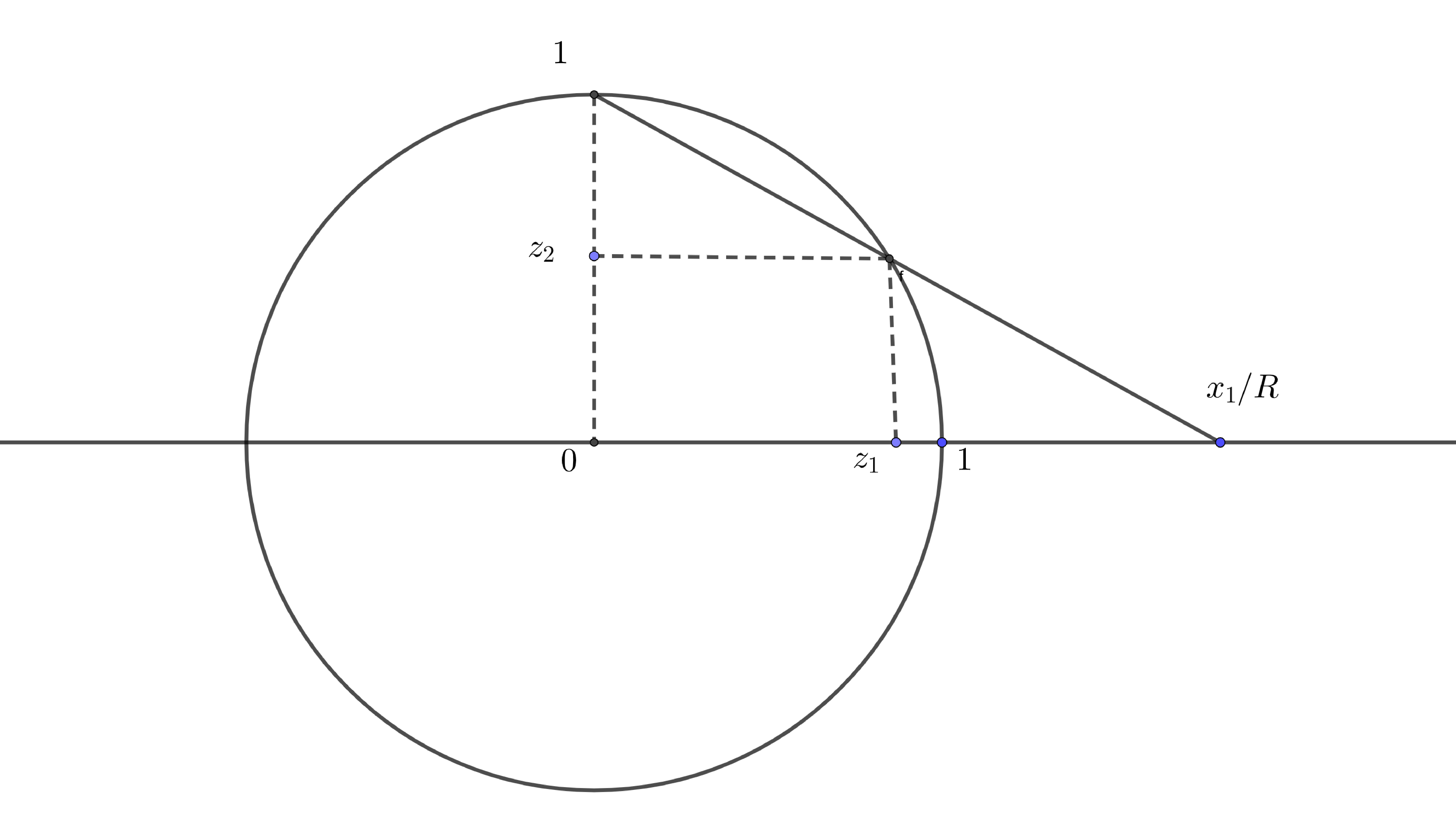}
 \caption{Illustration of Stereographic Projection $\SP: \mathbb{S} \to \Reals$.}
	\label{fig-stereo}
\end{minipage}
\begin{minipage}{.49\linewidth}
  \includegraphics[width=\linewidth]{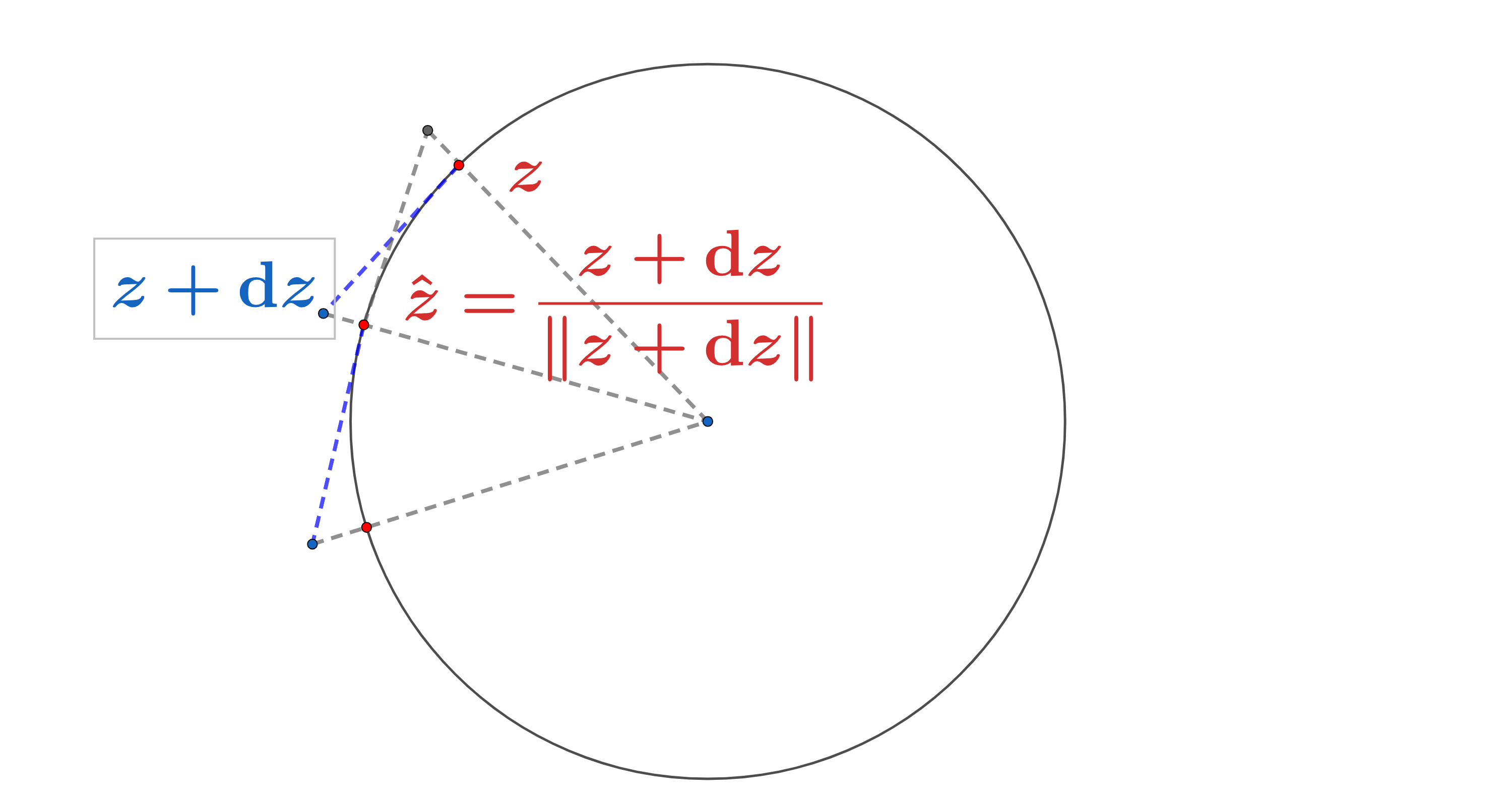}
	\caption{Illustration of SPS proposing $\hat{z}=\SP^{-1}(\hat{X})$ from $z=\SP^{-1}(x)$. By symmetricity, the proposal distribution is the same as proposing $z$ from $\hat{z}$.}
		\label{fig-algo}
\end{minipage}
\end{figure}


	
	Now suppose we wish to sample from a target density $\pi (x)$ where $x \in \Reals^d$.
	Our aim is to take advantage of our stereographic bijection to construct an MCMC sampler directly on $\mathbb{S}^{d}$, projecting the output back onto $\Reals ^d$. 
	We denote the transformed target as $\pi_S(z)$ then for $x=\SP(z)$ we have
	\[\label{eq_pi_S}
	\pi_S(z)&\propto \pi(x)(R^2+\|x\|^2)^d.
	\]
	
	First, we just consider a random-walk Metropolis algorithm 
	\BLUE{with a step size $h$ on the unit sphere. Note that we describe our algorithms on the unit sphere. The radius $R$ will only be taken into account when projecting to $\Reals^d$.} \cref{fig-algo} illustrates how the algorithm moves are constructed on the unit sphere. (Note that we could very easily have constructed more general Metropolis--Hastings algorithms.)

	\begin{algorithm}\label{our_algo_SPS}
		\caption{Stereographic Projection Sampler (SPS)}
	\begin{itemize}
		\item Let the current state be $X^d(t)=x$;
		\item Compute the proposal $\hat{X}$:
		\begin{itemize}
			\item Let $z:=\SP^{-1}(x)$;
			\item Sample independently $\dee\tilde{z}\sim \mathcal{N}(0,h^2I_{d+1})$;
			\item Let $\dee z:= \dee \tilde{z}-\frac{(z^T\cdot\dee\tilde{z})z}{\|z\|^2}$ and $\hat{z}:=\frac{z+\dee z}{\|z+\dee z\|}$;
			\item The proposal $\hat{X}:=\SP(\hat{z})$.
		\end{itemize}
		\item $X^d(t+1)=\hat{X}$ with probability $1\wedge\frac{\pi(\hat{X})(R^2+\|\hat{X}\|^2)^d}{\pi(x)(R^2+\|x\|^2)^d}$; otherwise $X^d(t+1)=x$.
	\end{itemize}
	\end{algorithm}
	
	The symmetry of the re-projection mechanism in the proposal of the algorithm means that SPS is indeed a valid MCMC algorithm for $\pi$. More formally, we have the following result.
	
	\begin{proposition}
	\label{prop-SPSergodic} If $\pi(x)$ is positive and continuous in $\Reals^d$, then
	SPS gives rise to an ergodic Markov chain on $\Reals^d$ with invariant distribution $\pi $.
		\end{proposition}
		\begin{proof}
		See \cref{proof-prop-SPSergodic}. 
	\end{proof}

    \def\state{E}
\BLUE{
We shall address the chain's uniform ergodicity properties.
Recall that a Markov chain $X$ on state space $\state$ with transition kernel $P$ is {\em uniformly ergodic} if $\forall \epsilon >0$, there exists, $N\in {\Bbb N}$ such that $\|
P^N(x, \cdot ) - \pi
\|_{\textrm{TV}}\le \epsilon, \forall x\in\state$, where $\pi$ denotes the chain's (unique) invariant distribution and $\| \cdot \|_{\textrm{TV}}$ represents total variation distance. Intuitively, for uniformly ergodic chains, we cannot have ``arbitrarily bad'' starting values. }

	It is long recognised that random-walk Metropolis (RWM) algorithms on bounded spaces are usually uniformly ergodic, while the same algorithms on unbounded spaces are never uniformly ergodic \citep{Mengersen1996}. Therefore, given the compactness of $\mathbb{S}^{d}$, it is reasonable to hope that SPS might be uniformly ergodic under mild regularity conditions on $\pi $. Since the actual state space of SPS is in fact $\mathbb{S}^{d}\setminus N$ 
	which is not compact, this question is more complicated than in the Euclidean state space case as we need to consider the properties of the transformed density near $N$. However, our first main result confirms that we do get uniform ergodicity if and only if the transformed density on the sphere is bounded at $N$.
	
	
	\begin{theorem}\label{thm-uniform-ergodic}
		If $\pi(x)$ is positive and continuous in $\Reals^d$, then SPS is uniformly ergodic if and only if
		\[\label{cond_uniform_ergodic}
		\sup_{x\in \Reals^d} \pi(x)(R^2+\|x\|^2)^d<\infty.
		\]
	\end{theorem}
	\begin{proof}
		See \cref{proof-thm-uniform-ergodic}.
	\end{proof}
	\begin{example}
	If the target $\pi(x)$ where $x\in \Reals^d$ is multivariate student's $t$ distribution with degrees of freedom no smaller than $d$, then the SPS algorithm is \emph{uniformly ergodic}.
\end{example}
	\begin{remark}\label{remark_uniform_ergodic}
		We make the following remarks:
	\begin{enumerate}
		\item The condition \cref{cond_uniform_ergodic} is necessary: if $\sup_{x\in \Reals^d} \pi(x)(R^2+\|x\|^2)^d=\infty$, then the chain is not even geometrically ergodic \citep[Proposition 5.1]{Roberts1996}; 
		\item The traditional RWM algorithm is \emph{not} uniformly ergodic if the support of $\pi$ is $\Reals^d$ \citep[Theorem 3.1]{Mengersen1996} and \emph{not} geometrically ergodic for any heavy-tailed target distribution \citep[Corollary 3.4]{jarner2000geometric};
		\item The condition that $\pi(x)$ is positive and continuous in $\Reals^d$ in both \cref{prop-SPSergodic} and \cref{thm-uniform-ergodic} can be relaxed. We used it here just for the simplicity of the proof. 
	\end{enumerate}
	\end{remark}

    \subsection{Stereographic Bouncy Particle Sampler (SBPS)}\label{subsec_SBPS}

Many recent innovations in MCMC algorithm construction have focused on non-reversible methods, most particularly those described by {\em piecewise deterministic Markov processes} (PDMPs)
	(see for example \citep{bouchard2018bouncy,bierkens2019zig} and \citep{MR790622} for theoretical background). PDMPs are continuous-time processes that have stochastic jumps at event times of a point process, but where the state evolves deterministically between the event times.
	In this subsection, we shall demonstrate that we can readily incorporate these methods within our projective framework.
	We shall concentrate on a version of the {\em Bouncy Particle Sampler} (BPS) \citep{bouchard2018bouncy} as this adapts naturally to our context. The algorithm is described as follows.
	
	PDMPs utilise an auxiliary random variable $v$, which in the case of the Stereographic Bouncy Particle Sampler (SBPS) has stationary distribution uniformly distributed on $\mathbb{S}^{d}$ and independently of $x$. One feature of PDMP algorithms such as SBPS is the option to include {\em refresh} moves that contribute to the intensity of their constituting point process (according to some possibly $x$-dependent hazard rate) and which independently refresh $v$ by sampling it from its invariant distribution (uniform on $\mathbb{S}^{d}$). In the description below, we restrict ourselves to the case where this refresh rate is constant.

\begin{algorithm}\label{our_algo_SBPS}
	\caption{Stereographic Bouncy Particle Sampler (SBPS)}
	\begin{itemize}
		\item Initialize $z^{(0)}\in\mathbb{S}^{d}$ and $v^{(0)}$ such that $v^{(0)}\cdot z^{(0)}=0$ and $\|v^{(0)}\|=1$.
		\item Simulate BPS on unit sphere: for $i=1,2,\dots$
		\begin{itemize}
			\item Simulate bounce time $\tau_{\textrm{bounce}}$ of a Poisson process of intensity
			\*[
			\chi(t)=\lambda(\sin(t)v^{(i-1)}+\cos(t)z^{(i-1)},\cos(t)v^{(i-1)}-\sin(t)z^{(i-1)}),
			\]
			where 
			\*[
			\lambda(z,v):=\max\{0,\left[-v\cdot\nabla_z\log\pi_S(z)\right]\}.
			\]
			\item Simulate refreshment time $\tau_{\textrm{refresh}}\sim \textrm{Exponential}(\lambda_{\textrm{refresh}})$.
			\item Let $\tau_i=\min\{\tau_{\textrm{bounce}},\tau_{\textrm{refresh}}\}$ and
			\*[
			z^{(i)}&=\sin(\tau_i)v^{(i-1)}+\cos(\tau_i)z^{(i-1)}.
			\]
			\item If $\tau_i=\tau_{\textrm{refresh}}$, sample new $v^{(i)}$ independently
			\*[
			v^{(i)}\sim \textrm{Uniform}\{v: z^{(i)}\cdot v=0, \|v\|=1\}
			\] 
			\item If $\tau_i=\tau_{\textrm{bounce}}$, compute
			\*[
			v^{(i)}=v_{\textrm{temp}}-2\left[\frac{v_{\textrm{temp}}\cdot\tilde{\nabla}_z\log\pi_S(z^{(i)})}{\tilde{\nabla}_z\log\pi_S(z^{(i)})\cdot \tilde{\nabla}_z\log\pi_S(z^{(i)})}\right] \tilde{\nabla}_z\log\pi_S(z^{(i)}),
			\]
			where 
			\*[
			v_{\textrm{temp}}&=\cos(\tau_i)v^{(i-1)}-\sin(\tau_i)z^{(i-1)}\\
			\tilde{\nabla}_z\log\pi_S(z^{(i)})&=\nabla_z\log\pi_S(z^{(i)})-\left[z^{(i)}\cdot\nabla_z\log\pi_S(z^{(i)})\right]z^{(i)}.
			\]
			\item If $\sum_{j=1}^i \tau_j\ge T$ (where $T$ is some constant time), exit.
		\end{itemize}
  \item Return $x=\SP(z)$ where $z$ denotes BPS on unit sphere.
	\end{itemize}
\end{algorithm}

Again, we demonstrate that SBPS is indeed a valid MCMC algorithm, at least under certain conditions on $\pi $ and the refresh rate in the algorithm. 	Note that the conditions on $\pi $ and $\lambda_{\textrm{refresh}}$ in \cref{prop-SBPSergodic} are not required for invariance, only for ensuring $\phi$-irreducibility of the algorithm.

\begin{proposition}
	\label{prop-SBPSergodic}
	Suppose that $\lambda_{\textrm{refresh}}>0$ and $\pi >0$ for all $x \in \Reals^d$.
	Then SBPS gives rise to an ergodic Markov chain on $\Reals^d$ with invariant distribution for $(x,v)$ with joint density on $\mathbb{R}^d\times \mathbb{S}^{d}$. The marginal density on $\mathbb{R}^d$ is proportional to $\pi (x)$. 
		\end{proposition}
		\begin{proof}
		The proof of this result is routine (but tedious), following the lines of existing results for PDMPs in the literature.
		We give a sketch proof in \cref{proof-prop-SBPSergodic}. 
	\end{proof}

\BLUE{
 We shall prove a uniform ergodicity result analogous to that of 
	Theorem \ref{thm-uniform-ergodic}. For a Markov process $X$ with state space $\state$, transition semi-group $P$, and invariant distribution $\pi$, it is uniformly ergodic
 if $\forall \epsilon >0$ there exists $T$ such that $\| P^T(x, \cdot )-\pi \|_{\textrm{TV}}\le \epsilon, \forall x\in \state$.
}
  
    
	\begin{theorem}\label{thm_PDMP}
		If $\pi(x)$ is positive in $\Reals^d$ with continuous first derivative in all components,
	then SBPS is \emph{uniformly ergodic} if
	{
		\*[
		\limsup_{\{x: \|x\|\to \infty\}} \sum_{i=1}^d \left(\frac{\partial \log \pi(x)}{\partial x_i} x_i\right)+2d<\frac{1}{2}.
		\]
	}
\end{theorem}
\begin{proof}
	See \cref{proof_thm_PDMP}.
\end{proof}

\begin{corollary}\label{coro_PDMP}
	If the target $\pi(x)$ where $x\in \Reals^d$ is multivariate student's $t$ distribution with degrees of freedom larger than $d-\frac{1}{2}$, then the SBPS algorithm is \emph{uniformly ergodic}.
\end{corollary}
\begin{proof}
  See \cref{proof_coro_PDMP}.
\end{proof}

\begin{remark}\label{conjecture}
We make the following remarks:
\begin{enumerate}
\item This is the first known PDMP algorithm that is uniformly ergodic for a large family of target distributions including heavy-tailed targets. For comparison, the Euclidean BPS is only known to be geometrically ergodic under certain restrictive conditions on the target \citep{deligiannidis2019exponential,durmus2020geometric} and is known not to be geometrically ergodic for any heavy-tailed target distribution \citep{vasdekis2021speed}.
\item We {conjecture} that the best possible condition for \cref{thm_PDMP} is
		\*[
		\limsup_{\{x: \|x\|\to \infty\}} \sum_{i=1}^d \left(\frac{\partial \log \pi(x)}{\partial x_i} x_i\right)+2d<{1}.
		\]
We explain the reason for our {conjecture} in \cref{proof_conjecture}.
Therefore,  we {conjecture} that \cref{coro_PDMP} is only loose by $\frac{1}{2}$ degree. That is, if the target $\pi(x)$ where $x\in \Reals^d$ is multivariate student's $t$ distribution with degree of freedom larger than $d-1$, then the SBPS algorithm is {conjectured} to be \emph{uniformly ergodic}. 
\end{enumerate}
\end{remark}

	To finish this section, we present some typical sample paths obtained by implementing SBPS.


\begin{example}
In \cref{fig-toy-sphere-BPS}, we show the proposed SBPS without and with refreshment for standard Gaussian target in two dimensions.  Note that the SBPS without refreshment is not irreducible in this case. The same issue exists for the Euclidean BPS for standard Gaussian targets. This issue can be fixed by adding the refreshment.

	\begin{figure}[h]
	\includegraphics[width=0.45\textwidth]{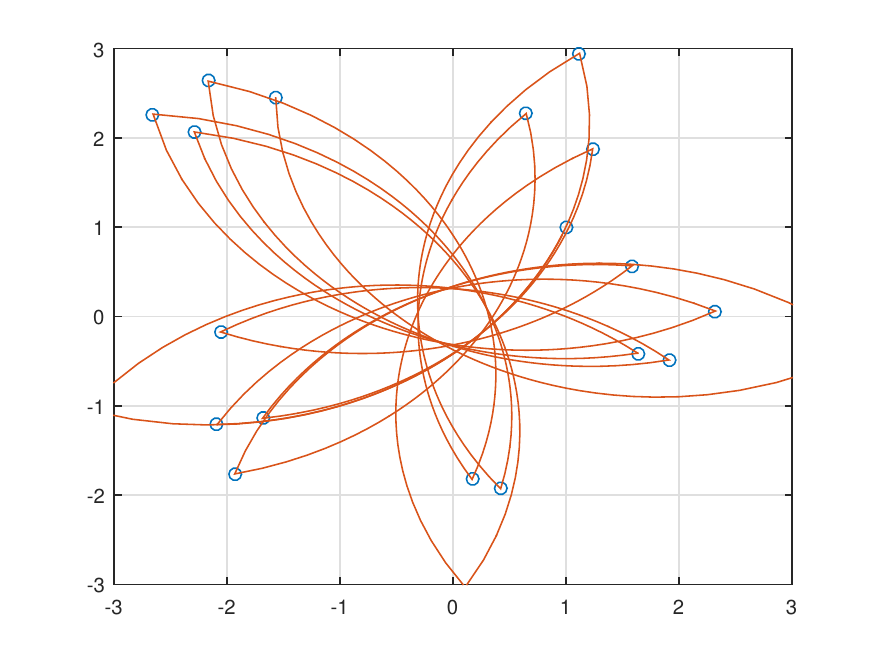}
	\includegraphics[width=0.45\textwidth]{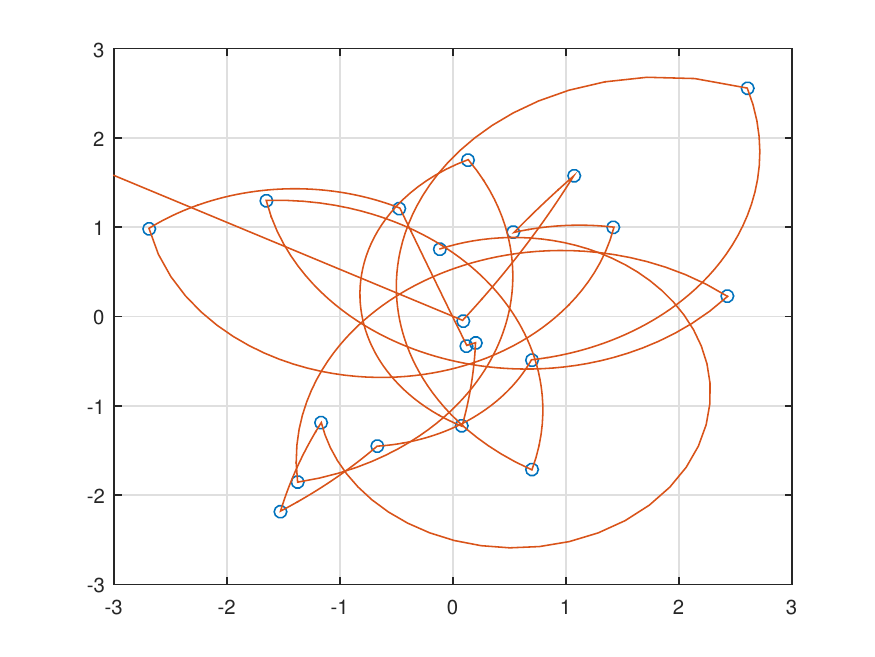}
	\caption{SBPS without (left) and with (right) refreshment for target distribution $\mathcal{N}(0,I_2)$. Note that the SBPS chain is not irreducible without refreshment.}
	\label{fig-toy-sphere-BPS}
    \end{figure}
\end{example}

	\section{SPS: Isotropic Targets }\label{section_isotropic}
	In this section, we consider \emph{isotropic targets}, which is also called spherical symmetric targets\footnote{Note that in some literature, isotropic distribution could be used to denote a distribution with zero mean and identity covariance matrix, which is different from our definition in this paper.}. That is, $\pi(x)$ is only a function of $\|x\|$. Recall the mapping and \cref{eq_SP} and the transformed target $\pi_S(z)$ in \cref{eq_pi_S}. We can see that, for any isotropic target, the transformed target
	$\pi_S(z)$ is only a function of the ``latitude'' $z_{d+1}$. \BLUE{In this sense, isotropic targets are the ``best'' targets for stereographically projected algorithms.}
	
	We shall assume that all second moments exist, and without loss of further generality, we suppose $\pi(x)$ satisfies that
	\*[
	\EE_{X\sim \pi}[\|X\|^2]=d.
	\]
	\BLUE{Under this assumption, the ``optimal'' radius should be chosen as
	$R=\sqrt{d}$, since it maps the concentration region of $\pi$ to the neighborhood of the ``equator'' of the sphere. Throughout this section, we study this best scenario. We will study the robustness to $R$ and the optimal scaling of SPS for any given $R$ in \cref{section_scaling} for another family of targets.}
	
	{Under the above assumptions, it suffices to study the path of the absolute value of the ``latitude'' $z_{d+1}$ of SPS. Informally, the ``stationary phase'' of SPS is the period in which $z_{d+1}=\bigO(d^{-1/2})$ and the ``transient phase'' is the period in which $|z_{d+1}|$ is larger than $\bigO(d^{-1/2})$. \BLUE{Somehow surprisingly, by analyzing the proposed ``latitude'' $\hat{z}_{d+1}$, we can (informally) show that the SPS enjoys the \emph{blessings of dimensionality}: the number of iterations for the ``latitude'' of SPS to decrease to $\bigO(d^{-1/2})$ \emph{decreases} with the dimension $d$. This implies (informally) that the ``transient phase'' of SPS is $\bigO(1)$ and it decreases with dimension.}

 
	\subsection{Analysis of the proposal distribution}
	We assume the chain starts from either the North Pole or the South Pole.
	By the assumptions, the chain is in the  ``stationary phase'' once it is around the ``equator''.
	
	In the following, we give useful approximations for the proposed ``latitude'' $\hat{z}_{d+1}$ in both the transient phase and stationary phase. This can be used to analyze the behavior of isotropic targets. See \cref{sec-fig-latitute-approximation-simulation} for some simulations using the result from \cref{lemma-lati}.
	
	\begin{lemma}\label{lemma-lati}
		Let $z_i$ be the current $i$-th coordinate and $\hat{z}_i$ be the $i$-th coordinate of the proposal, where $i=1,\dots,d+1$. Then we have the following expression
		\[\label{eq_approx_large_h}
		\hat{z}_i&=\frac{1}{\sqrt{1+h^2(U^2+U_{\perp}^2)}}\left(z_i+\sqrt{1-z_i^2}hU\right),
		\]
		where $U\sim\mathcal{N}(0,1)$ and $U_{\perp}^2\sim \chi^2_{d-1}$ which is independent with $U$.
		Furthermore, if $h=\bigO(d^{-1/2})$, then we have the following coordinate-wise approximation
		\[\label{eq_approx_whole}
		\hat{z}_{i}= \frac{1}{\sqrt{1+h^2 U_{\perp}^2}}\left[\left(1-\frac{1}{2} h^2 U^2\right)z_{i}-\sqrt{1-z_{i}^2}h U\right]+\bigO_{\Pr}(h^3+dh^4z_i).
		\]

		As special cases of \cref{eq_approx_whole}, $z_{d+1}$ is the current ``latitude'' and $\hat{z}_{d+1}$ be the proposed ``latitude''. Then, in the transient phase, if $z_{d+1}^2=1-o(h^2)$, we have
		\[\label{eq_approx_transient}
		\hat{z}_{d+1}= \frac{1}{\sqrt{1+h^2(d-1)}}\left(1-\frac{1}{2}h^2U^2\right)z_{d+1}+\bigO_{\Pr}(d^{-1/2}),
		\]
		and in the stationary phase, if $z_{d+1}=\bigO(d^{-1/2})$, we have
		\[\label{eq_approx_stationary}
		\hat{z}_{d+1}=\frac{1}{\sqrt{1+h^2(d-1)}}\left(z_{d+1}-hU\right)+\bigO_{\Pr}(d^{-1}).
		\]
	\end{lemma}
	\begin{proof}
		See \cref{proof-lemma-lati}. 
	\end{proof}
	\BLUE{The above lemma informally suggests that in the ``transient phase'', the proposed ``latitude'' is almost deterministic, whereas in the ``stationary phase'', the proposed ``latitude'' approximately follows an autoregressive process. For example, if we take $h$ to be constant and $i=d+1$, then \cref{eq_approx_large_h} shows that the proposal ``latitude'' $\hat{z}_{d+1}$ decreases to $\bigO(d^{-1/2})$ faster when dimension is larger since the concentration of $U_{\perp}^2/d$. If the proposals will be accepted by a probability bounded away from $0$, the ``transient phase'' of SPS only involve $\bigO(1)$ iterations and the number even decreases with the dimension. As an example, in the next subsection, we show a class of isotropic targets such that the SPS enjoys the \emph{blessings of dimensionality}.}
 

	\subsection{Examples of Isotropic Targets}\label{subsec-iso}
	We denote the multivariate student's $t$ distributions by $\pi_{\nu}(x)$ where $\nu$ is the degree of freedom. We denote the standard multivariate Gaussian as the limit $\pi_{\infty}(x)$. That is, 
	\*[
\pi_{\nu}(x)\propto \left(1+\frac{1}{\nu}\|x\|^2\right)^{-(\nu+d)/2},
\quad
\pi_{\infty}(x)\propto \exp\left(-\frac{1}{2}\|x\|^2\right)
\]
Then the logarithm of the likelihood ratio can be written as a function of $z_{d+1}$ and $\hat{z}_{d+1}$:
\*[
\log\frac{\pi_{\nu}(\hat{X})}{\pi_{\nu}(x)}+d\log\frac{R^2+\|\hat{X}\|^2}{R^2+\|x\|^2}= d(g_{\nu/d}(z_{d+1})-g_{\nu/d}(\hat{z}_{d+1})),
\]
where $g_{k}(z):=\frac{k+1}{2}\log(k+\frac{1+z}{1-z})+\log(1-z)$. If $k=\nu/d\to \infty$, we have $g_{k}(z)$ converges to $g_{\infty}(z):=\frac{1}{1-z}-\frac{1}{2}+\log(1-z)$ up to a constant. See \cref{sec-fig-gk} for the plots of the function of $g_k(z)$ for different values of $k$.

\begin{example}
	(Multivariate student's $t$ distribution with DoF $\nu=d$)
	It can be easily verified that
	$g_1(z)=\log(2)$.
	Therefore, any proposal will be accepted since the acceptance rate is always $1$ whatever $h$ is. 
\end{example}

\begin{example}
	(Multivariate student's $t$ distribution with DoF $\nu>d$, including Gaussian)
	\BLUE{From \cref{lemma-lati},  one can see informally: (i) in the ``transient phase'', as the proposed ``latitude'' is almost deterministic which has higher target density, the acceptance probability is almost $1$ starting from either the North Pole or the South Pole; (ii) In the ``stationary phase'', as the proposed ``latitude'' approximately follows an autoregressive process, the acceptance rate is well-approximated by a positive constant.}  This suggests that as long as the target has a ``lighter tail'' than multivariate student's $t$ with DoF $d$, the ``transient phase'' of SPS takes at most $\bigO(1)$ steps. {For comparison, for the standard multivariate Gaussian target, the ``transient phase'' of the Euclidean RWM takes $\bigO(d)$ steps \citep{Christensen2005}.}
\end{example}

\begin{remark}\label{rmk_heavy_tail}
	Our theoretical results don't cover the cases of SPS for targets with heavier tails, such as multivariate student's $t$ with DoF $\nu<d$. In this case, the SPS cannot start from the North Pole since the first proposal of the SPS will be rejected with probability $1$. One might consider to start from the South Pole. However, the SPS could get stuck at the South Pole if the proposal variance is large (even if the origin is the mode of the target density). See \cref{comment_rmk_heavy_tail} for comments. In practice, we suggest choosing the initial state of the SPS as a random state uniformly sampled on the sphere. 
\end{remark}

	\section{SPS: Extension to Elliptical Targets}\label{section_ellipse}
	
	\subsection{Extensions of Stereographic Projection}
	Same as the previous section, we denote the state of the Markov chain by $x=(x_1,\dots,x_d)^T\in\Reals^{d}$.
	Suppose $z=(z_1,\dots,z_{d+1})^T$ is the coordinates of a \emph{unit sphere in $\Reals^{d+1}$} (that is, $\|z\|=1$). 
	
	Now, we map $z\in \mathbb{S}^{d}$ to $x\in \Reals^{d}$ by the generalized stereographic projection (GSP)
	\*[
	x=\GSP(z):=Q\left(R\sqrt{\lambda_1}\frac{z_1}{1-z_{d+1}},\dots,R\sqrt{\lambda_d}\frac{z_d}{1-z_{d+1}}\right)^T,
	\]
	where $R$ is the radius parameter, $Q^T=Q^{-1}\in \Reals^{d\times d}$ is a rotation matrix, $\{\lambda_1,\dots,\lambda_d\}$ are non-negative constants. That is, the GSP is obtained by ``stretching'' via $\{\lambda_i\}$ and ``rotating'' via $Q$ from the SP.
	
	Defining the norm
	\*[
	\|x\|^2_{\Lambda,Q}:=x^TQ\Lambda^{-1}Q^Tx,
	\]
	where $\Lambda=\diag\{\lambda_1,\dots,\lambda_d\}$,
	we have the following Generalized Stereographic Projection Sampler (GSPS) with parameters $R,h,\Lambda,Q$. 
 
		\begin{algorithm}\label{our_algo}
		\caption{Generalized Stereographic Projection Sampler (GSPS)}
		\begin{itemize}
			\item Let the current state be $X^d(t)=x$;
			\item Compute the proposal $\hat{X}$:
			\begin{itemize}
				\item Let $z:=\GSP^{-1}(x)$;
				\item Sample independently $\dee\tilde{z}\sim \mathcal{N}(0,h^2I_{d+1})$;
				\item Let $\dee z:= \dee \tilde{z}-\frac{(z^T\cdot\dee\tilde{z})z}{\|z\|^2}$ and $\hat{z}:=\frac{z+\dee z}{\|z+\dee z\|}$;
				\item The proposal $\hat{X}:=\GSP(\hat{z})$.
			\end{itemize}
			\item $X^d(t+1)=\hat{X}$ with probability {$1\wedge\frac{\pi(\hat{X})(R^2+\|\hat{X}\|^2_{\Lambda,Q})^d}{\pi(x)(R^2+\|x\|^2_{\Lambda,Q})^d}$}; otherwise $X^d(t+1)=x$.
		\end{itemize}
	\end{algorithm}

\BLUE{Note that even though GSPS can be formulated alternatively as using ``ellipsoid'' instead of ``sphere'', this is not our formulation. We still use the unit sphere to compute $\hat{z}$ for GSPS and only replace the mapping $\SP$ by $\GSP$. That is, we only ``stretch'' and ``rotate'' the proposal $\hat{z}$ via $\GSP$ when projecting to $\Reals^d$. }
    
	The following example shows that we can extend isotropic targets to elliptical targets using the GSPS.
	\begin{example}
		Suppose $\pi(x)$ is multivariate student's $t$ with covariance matrix $\Sigma=Q\Lambda Q^T$ where $\Lambda=\diag\{\lambda_1,\dots,\lambda_d\}$:
		\*[
		\pi(x)\propto \left(1+\frac{1}{\nu}x^T\Sigma^{-1}x\right)^{-(\nu+d)/2}
		\]
		Then, for the GSPS with the corresponding $Q$ and $\Lambda$, and $R^2=\sum_i\lambda_i$, the acceptance rate is always $1$ for any $h$. 
	\end{example}

	The GSPS naturally suggests that one can estimate the covariance matrix $\Sigma=Q\Lambda Q^T$ under the adaptive MCMC framework, which results in adaptive GSPS. Indeed, if both $\Lambda$ and $Q$ are known, then one can normalize GSPS and reduce it to SPS, which is GSPS with $\Sigma=I_d$. 
	
	As the robustness to estimations of the mean and the covariance matrix of the target is the key to the success of adaptive GSPS, we study the robustness of Gaussian targets in the next section.
		
	\subsection{Robustness for Gaussian Targets}

	We consider multivariate Gaussian targets with  mean vector $\mu$ and covariance matrix $\Sigma$:
	\*[
	\pi_{\mu,\Sigma}(x)\propto \exp\left(-\frac{1}{2}(x-\mu)^T\Sigma^{-1}(x-\mu)\right),
	\]
	where
		\*[
	\Sigma=\diag(\lambda_1,\dots,\lambda_d),\quad \mu=(\mu_1,\dots,\mu_d)^T.
	\]
	We are interested in the robustness when $\mu\neq 0$ and $\Sigma\neq I_d$.
	
	 {For comparison with the traditional RWM in $\Reals^d$ on the orders of stepsize, we recall that the optimal scaling theory gives an optimal stepsize of $\bigO(d^{-1/2})$ for RWM \citep{Roberts1997}. Our stepsize $h$ is defined on the unit sphere. When projecting to $\Reals^d$, we multiply by $R=\bigO(d^{1/2})$. Therefore, the optimal stepsize of RWM roughly corresponds to $h=\bigO(d^{-1})$ in our setting. In this subsection, we study how large the stepsize $h$ can be so the expected acceptance probability of SPS goes to $1$. This roughly corresponds to a stepsize of $\bigO(hd^{1/2})$ in $\Reals^d$. Our results show that the order of the stepsize of SPS when projected back to $\Reals^d$ is always no smaller than the order of the optimal stepsize of RWM.}

	\begin{theorem}\label{thm-robust}
		Assume there exist constants $C<\infty$ and $c>0$ such that
		$c\le \lambda_i\le C$ and $|\mu_i|\le C$ for all $i=1,\dots,d$. Furthermore, assume $R=d^{1/2}$ and
		\[\label{eq_robust_cond}
		\left|\sum_{i=1}^d\mu_i^2-\sum_{i=1}^d(1-\lambda_i)\right|=\bigO(d^{\alpha}),
		\]
		where $\alpha\le 1$.
		Then, under stationarity that $X\sim \pi$, the expected acceptance probability converges to $1$ as $d\to\infty$ 
		\*[
		\EE_{X\sim \pi_{\mu,\Sigma}}\EE_{\hat{X}\,|\,X}\left[1\wedge\frac{\pi_{\mu,\Sigma}(\hat{X})(R^2+\|\hat{X}\|^2)^d}{\pi_{\mu,\Sigma}(X)(R^2+\|X\|^2)^d}\right]\to 1,
		\]
		for all $h$ such that
		\[\label{eq-result-robust}
		h=o\left(\frac{d^{-1}}{\sqrt{\max\left\{\frac{1}{d}\sum_i |1-\lambda_i|,\frac{1}{d}\sum_i\mu_i^2\right\}}}\wedge {d^{-(\frac{1}{2} \vee \alpha)}} \right).
		\]
	\end{theorem}
	\begin{proof}
		See \cref{proof-thm-robust}.
	\end{proof}

Note that for the standard Gaussian target, we know that the acceptance rate goes to $1$ as $d\to \infty$ for any $h$. Recall that for traditional RWM in dimension $k$, the optimal stepsize is $\bigO(k^{-1/2})$ in $\Reals^k$ under the optimal scaling framework \citep{Roberts1997}, which corresponds to $h=\bigO(k^{-1/2}d^{-1/2})$ in the SPS setting. Therefore, \cref{thm-robust} suggests that the ``effective dimension'' is determined by $\{\mu_i\}$ and $\{\lambda_i\}$: as long as $\{\lambda_i\}$ and $\{\mu_i\}$ are uniformly bounded, the ``effective dimension'' of SPS is never larger than $d$, which is the ``effective dimension'' of the traditional RWM. Furthermore, we can compare the ``effective dimension'' for SPS with $d$ using \cref{thm-robust} under different settings.

\begin{example}	
	Suppose $\mu_i=0$ for all $i=1,\dots,d$. Moreover, $\lambda_1=\lambda_2=\dots=\lambda_k=2$ and $\lambda_{k+1}=\dots=\lambda_d=1$, then as long as $k$ is a fixed number, $\frac{1}{d}\sum_i|1-\lambda_i|=\bigO(kd^{-1})$. By \cref{thm-robust}, when $h=o(k^{-1/2}d^{-1/2})$ then the acceptance rate goes to $1$ as $d\to\infty$. This suggests that the ``effective dimension'' for SPS is no more than $k$.
	\end{example}
\begin{example}
	Suppose {$\lambda_i=0$ for all $i=1,\dots,d$.} Furthermore, $\mu_1=\mu_2=\dots=1$ and $\mu_{k+1}=\dots=\mu_d=0$ where $k$ is a fixed number. Then we have $\frac{1}{d}\sum_i\mu_i^2=\bigO(kd^{-1})$.
	By \cref{thm-robust}, the acceptance rate goes to $1$ when
	$h=o(k^{-1/2}d^{-1/2})$, which implies the ``effective dimension'' for SPS is no more than $k$.
\end{example}

	One can consider \cref{thm-robust} in the context of adaptive MCMC \BLUE{for Gaussian targets}, in which $\{\mu_i\}$ and $\{1-\lambda_i\}$ represent the ``estimation errors'' of the coordinate means and eigenvalues of the covariance matrix of the target. One can choose the ``radius'' $R$ of the stereographic projection properly to satisfy \cref{eq_robust_cond} with $\alpha=\frac{1}{2}$. This is to properly scale $R$ so that the ``latitude'' on the unit sphere of the SPS is $\bigO_{\Pr}(d^{-1/2})$. Then, according to \cref{thm-robust}, the ``effective dimension'' of SPS is smaller than $d$ if $\frac{1}{d}\sum_i|1-\lambda_i|=o(1)$ and $\frac{1}{d}\sum_i\mu_i^2=o(1)$.  \BLUE{However, the above results on ``effective dimension'' are based on \cref{thm-robust}, which only hold for Gaussian targets. For more results on robustness and optimal acceptance rate for adaptive MCMC, we will study another family of targets in \cref{section_scaling}.}

	Finally, we show two examples that the result of \cref{thm-robust} is tight. {First, consider the special case that $\mu_i=\mu>0$ and $\lambda_i=\sigma^2$ for all $i$. We assume $\mu^2=1-\sigma^2$ so that \cref{eq_robust_cond} holds for any $\alpha\le 1/2$.} Then \cref{eq-result-robust} suggests that the acceptance probability goes to zero if $h=o(d^{-1}/\mu)$. On the other hand, the target in this special case is a product i.i.d.\ target with marginal distribution $\mathcal{N}(\mu,1-\mu^2)$. By our optimal scaling results in \cref{section_scaling}, we know that the acceptance probability does not go to zero if $h=\bigO(d^{-1}/\mu)$ (see \cref{lemma_CLT} and \cref{coro_CLT}). Therefore, \cref{eq-result-robust} is tight. {Second, consider the special case that $\mu_i=\mu=0$ and $\lambda_i=\sigma^2\neq 1$ for all $i$. Then \cref{eq_robust_cond} holds for $\alpha=1$. In this case, \cref{eq-result-robust} suggests that the acceptance probability goes to zero if $h=o(d^{-\alpha})=o(d^{-1})$. On the other hand, the target in this special case is a product i.i.d.~target with marginal $\mathcal{N}(0,\sigma^2)$ where $\sigma^2\neq 1$. If we properly rescale $R$ and $\sigma^2$, it reduces to the case of \cref{lemma_CLT} where the target is the standard Gaussian but $R\neq d^{1/2}$. Therefore, by \cref{lemma_CLT}, the acceptance probability doesn't converge to zero if $h=\bigO(d^{-1})$. Therefore, \cref{eq-result-robust} is again tight.}
	\section{SPS in High-dimensional Problems}
	\label{section_scaling}
	In this section, we shall give a case study of the behavior of SPS for high-dimensional target densities. We shall consider various limits of SPS as $d\to \infty $ though to have tractability of this limit we need to consider a very specialised class of target densities.
		To this end, analogous to \citep{Roberts1997}  we assume the target $\pi(x)$ has a product i.i.d.~form.
		
	\subsection{Assumptions on $\pi$}\label{subsec_assumption}
	 We assume the target $\pi(x)$ has a product i.i.d.~form: 
		\[\label{assump_iid}
		\pi(x)=\prod_{i=1}^d f(x_i).
		\]
		Without loss of generality, we assume $f$ is normalized such that
		\[\label{assump0}
		\EE_f(X^2)=\int x^2f(x)\dee x=1,\quad \EE_f(X^6)<\infty.
		\]
		We further assume $f'/f$ is Lipschitz continuous, $\lim_{x\to\pm\infty}xf'(x)=0$, and
		\[\label{assump1}
		&\EE_f\left[\left(\frac{f'(X)}{f(X)}\right)^8\right]<\infty,\quad \EE_f\left[\left(\frac{f''(X)}{f(X)}\right)^4\right]<\infty,\quad \EE_f\left[\left(\frac{Xf'(X)}{f(X)}\right)^4\right]<\infty.
		\]
			\begin{remark}\label{remark_nu}
		Under the assumption $\EE_f(X^2)=1$, by Cauchy--Schwarz inequality
		\*[
		\EE_f\left[\left((\log f)'\right)^2\right]\ge 1
		\]
		where the equality is achieved by standard Gaussian (or truncated standard Gaussian). 
		{For univariate student's $t$ distribution with any DoF$=\nu>2$, rescaling by a factor $\sqrt{\frac{\nu-2}{\nu}}$, we can obtain a target density $f_{\nu}$ with $\EE_{f_{\nu}}(X^2)=1$ and
		\*[
			C_{\nu}:=\EE_{f_{\nu}}\left[\left((\log f_{\nu})'\right)^2\right]=\left(\frac{\nu}{\nu-2}\right)\left(\frac{\nu+1}{\nu}\right)\left(\frac{\nu+4}{\nu+3}\right)\sqrt{\frac{\nu+4}{\nu}}> 1.
		\]
		where $\nu\to\infty$ recovers the case for the standard Gaussian target. See \cref{table-nu} for values of $C_{\nu}$  for different $\nu$. One can see that $C_{\nu}$  is very close to $1$ for even a medium size of $\nu$.
		}
\begin{table}
\caption{Examples of $C_{\nu}$ and $C_{\nu}/(C_{\nu}-1)$ for different $\nu$ in \cref{remark_nu} and \cref{remark-ESJD}}
\label{table-nu}
\begin{tabular}{@{}lcrcrrr@{}}
\hline
 & $\nu=3$ &
$\nu=5$ &
$\nu=10$ &
$\nu=20$ &
$\nu=50$ &
$\nu=100$ \\
\hline\hline
$C_{\nu}$ & $7.1285$ &  $3.0187$  & $1.7521$ & $1.3336$ & $1.1250$ & $1.0612$  \\
\hline
$C_{\nu}/(C_{\nu}-1)$ & $1.1632$ & $1.4954$ & $2.3297$ & $3.9977$ & $8.9990$ & $17.3328$\\
\hline
\end{tabular}
\end{table}
	\end{remark}
	
	In this section, we consider $f$ has full support in $\Reals$. Then the product i.i.d.~target $\pi$ is not isotropic unless it is the standard Gaussian. 
	
	\subsection{Acceptance Probability}
	To make progress, we need a detailed understanding of the high-dimensional behavior of the acceptance probability of SPS. It turns out that to optimally apply SPS, we need to scale $R$ to be $\bigO(d^{1/2})$. Not doing this will concentrate mass at either the North or South polls in an ultimately degenerate way. We shall thus assume $R$ to be scaled in this way in what follows.
    \BLUE{Therefore, we consider $R=\sqrt{\lambda d}$ for a fixed constant $\lambda$. To simplify the final result, we replace the step size $h$ by another parameter $\ell$ via
		\begin{equation}\label{new_par_ell}
		   h=\frac{1}{\sqrt{d-1}}\left[\frac{1}{\left(1-\frac{\ell^2}{2d}{\frac{4\lambda}{(1+\lambda)^2}}\right)^2}-1\right]^{1/2}
		\end{equation}
  which implies $\frac{1}{\sqrt{1+h^2(d-1)}}=1-\frac{\ell^2}{2d}{\frac{4\lambda}{(1+\lambda)^2}}$.
  Note that $\ell$ is simply a re-parameterization of $h$. When $\ell$ is a fixed constant, $h$ is scaled as $\bigO(d^{-1})$ since $h\approx \frac{\ell}{d}\sqrt{4\lambda /(1+\lambda)^2}$.
}

	\begin{lemma}\label{lemma_CLT}
			Under the assumptions on $\pi$ in \cref{subsec_assumption}, suppose the current state $X\sim \pi$, and the parameter of the algorithm {$R=\sqrt{\lambda d}$, where $\lambda>0$ is a fixed constant}. \BLUE{We re-parameterize $h$ by $\ell$ according to \cref{new_par_ell}}.
		Then, if either $\lambda\neq 1$ or $f$ is not the standard Gaussian density, 
	 there exists a sequence of sets $\{F_d\}$ such that $\pi(F_d)\to 1$ and
		\*[
	\sup_{X\in F_d}\EE_{\hat{X}\mid X}\left[\left|1\wedge \frac{\pi(\hat{X})(R^2+\|\hat{X}\|^2)^d}{\pi(X)(R^2+\|X\|^2)^d}-1\wedge \exp(W_{\hat{X}\mid X})\right|\right]=o(d^{-1/4}\log(d)),
		\]
		where $W_{\hat{X}\mid X}\sim\mathcal{N}(\mu,\sigma^2)$ and
			\*[
		\mu&= \frac{\ell^2}{2}\left\{\frac{4\lambda}{(1+\lambda)^2}-\EE_f\left[\left((\log f)'\right)^2\right]\right\},\quad
		\sigma^2= \ell^2\left\{\EE_f\left[\left((\log f)'\right)^2\right]-\frac{4\lambda}{(1+\lambda)^2}\right\}.
		\]
	\end{lemma}
	\begin{proof}
		See \cref{proof-lemma_CLT}.
	\end{proof}
	\begin{remark}
	\cref{lemma_CLT} requires either $\lambda\neq 1$ or $\pi$ is not the standard multivariate Gaussian.
	    If $\lambda=1$ and $\pi$ is the standard multivariate Gaussian, the case reduces to isotropic targets discussed in \cref{subsec-iso}, so the Gaussian approximation in \cref{lemma_CLT} doesn't hold.
	\end{remark}

\BLUE{
\subsection{Optimisation and robustness of SPS}
Note that
		\cref{lemma_CLT} suggests that the expected acceptance probability in the stationary phase
		\*[
		\EE_{X\sim\pi}\EE_{\hat{X}\,|\,X}\left[1\wedge 	\frac{\pi(\hat{X})(R^2+\|\hat{X}\|^2)^d}{\pi(X)(R^2+\|X\|^2)^d}\right]\to 2\Phi\left(-\frac{\sigma}{2}\right).
		\]
		Furthermore, as $\EE[\|\hat{X}-X\|^2]\to \ell^2$, we obtain the commonly used approximation of Expected Squared Jumping Distance (ESJD):
\begin{equation}\label{approx_ESJD}
		 \EE[\|\hat{X}-X\|^2]\cdot\EE\left[1\wedge 	\frac{\pi(\hat{X})(R^2+\|\hat{X}\|^2)^d}{\pi(X)(R^2+\|X\|^2)^d}\right]\to 2\ell^2\cdot\Phi\left(-\frac{\ell}{2}\sqrt{\EE_f\left[\left((\log f)'\right)^2\right]-\frac{4\lambda}{(1+\lambda)^2}}\right).    
		\end{equation}

In this subsection we shall explicitly consider the optimisation of SPS, and take a close look at its relative performance in comparison to standard Euclidean RWM. All the results in this section to date have used the convenient restriction that
$\EE_f(X^2)=1$. However, at this point we shall need to generalise this notion. Therefore consider density $f$ to have mean
$m$ and variance $s^2$.

We consider SPS with $R=\sqrt{\lambda (s^2+m^2) d}$. 
and stepsize $
h=\frac{1}{\sqrt{d-1}}\left[\frac{1}{(1-\frac{\ell^2}{2d}{\frac{4\lambda}{(1+\lambda)^2}})^2}-1\right]^{1/2}$.
We shall denote this algorithm's approximate limiting ESJD by
$E(\ell, m, s, \lambda)$ and its corresponding acceptance probability:
$A(\ell, m, s, \lambda)$.
The following result is a direct consequence of \cref{approx_ESJD} applied to a scaled density (dividing by $\sqrt{s^2+m^2}$):
$$
E(\ell, m, s, \lambda) = 
\sqrt{s^2+m^2} 
\cdot 2 \ell^2 \Phi\left(
{-\ell\over 2} \sqrt{
\tilde{I}(s,m, \lambda )
}  \right),\quad
A(\ell, m, s, \lambda) =  
2  \Phi\left(
{-\ell\over 2} \sqrt{
\tilde{I}(s,m,\lambda )
 } \right)\ ,
$$
where 
$\tilde{I}(s, m,\lambda ) = (s^2+m^2)I-{4\lambda \over (1+\lambda )^2} $
and
$I={\EE} _f((\log f)'(X)^2)$.

We note that we can readily recover the the corresponding quantities for RWM: $E(\ell, m, s, \infty )$
and $A(\ell, m, s, \infty )$. In particular, to consider robustness, we shall consider the relative
performance ratio
$$
\tilde{R}(m,s,\lambda ) := {\sup_\ell E(\ell, m, s, \lambda ) \over
\sup_\ell E(\ell, m, s, \infty ) }
$$
To get an expression for $\tilde{R}$, we shall follow the standard approach of \cite{Roberts1997} of expressing the efficiency in terms of acceptance rate. To that end, we define
$$
{\tilde E}(A, m, s, \lambda) :=
{E}(\ell(A), m, s, \lambda)
$$
where $\ell $ is chosen to be the unique solution to
${A}(\ell(A), m, s, \lambda)=A$.
A simple calculation yields the following.
\begin{corollary}
\label{corr:rob}
We have ${\tilde E}(A, m, s, \lambda)
=A\cdot \Phi^{-1}\left({A\over 2} \right) \cdot {4 {(s^2+m^2)}\over \tilde{I}(s,m, \lambda )}$
and
$\tilde{R}(m,s,\lambda) = {(s^2+m^2) I \over \tilde{I}(s,m,\lambda )}
=
{1\over 1 - \alpha \cdot \beta \cdot \gamma}$,
where
$
\alpha = {4 \lambda \over (1+\lambda )^2},
\beta = {s^2 \over s^2 + m^2}, 
\gamma = {1\over s^2 I}$.
Since all of $\alpha, \beta, \gamma $ are in $[0,1]$, In this situation, SPS is never less efficient than Euclidean RWM.
\end{corollary}
\begin{remark}\label{remark-robustness}
Primarily, Corollary \ref{corr:rob} is a robustness result. However it also highlights the (only) ways in which SPS can fail to outperform Euclidean RWM. The three constants $\alpha, \beta, \gamma $ characterise different sensitivities of the algorithm. 
\begin{description}
    \item{$\alpha$} measures the penalty due to mis-specification of the sphere radius. $4\lambda /(1+\lambda )^2$ is optimised at $\lambda =1$ when there is no penalty for misspecification of the target distribution dispersion.
    \item{$\beta$} describes the penalty due to mis-location of the hyper-sphere. It is seen that the optimal choice is to locate the sphere at the mean of the target distribution. 
    \item{$\gamma$} is a distribution-specific penalty.  It is straightforward by to check by functional calculus that $\gamma \le 1$ with equality achieved only in the case where the target density is Gaussian. Thus $\gamma $ is characterising proximity to Gaussianity. 
\end{description}
From this result, it can be seen that we can only achieve super-efficiency (convergence complexity more rapid than $\bigO(d)$) when all three parameters, $\alpha$, $\beta$ and $\gamma$ are close to (and converging to) unity.
\end{remark}

}

	\subsection{Maximizing ESJD}
	In this subsection, \BLUE{we consider $\lambda=1$ only (i.e., $R=\sqrt{d}$) and prove that the approximate limiting ESJD in \cref{approx_ESJD} is indeed the limiting maximum ESJD.} We will demonstrate (again analogously to \citep{Roberts1997}) that its limit is optimised by targeting an acceptance probability of $0.234$.

	\begin{definition}
		Expected Squared Jumping Distance (ESJD):
		\*[
		\ESJD:=	\EE_{X\sim\pi}\EE_{\hat{X}\,|\,X}\left[\|\hat{X}-X\|^2\left(1\wedge 	\frac{\pi(\hat{X})(R^2+\|\hat{X}\|^2)^d}{\pi(X)(R^2+\|X\|^2)^d}\right)\right].
		\]
	\end{definition}

	\begin{theorem}\label{thm_ESJD}
		Under the assumptions on the target in \cref{subsec_assumption}, suppose $f$ is not the standard Gaussian density, the SPS chain is in the stationary phase, and the radius parameter is chosen as $R=\sqrt{d}$ \BLUE{and re-parameterize $h$ by $\ell$ according to \cref{new_par_ell} with $\lambda=1$}.
		Then, as $d\to\infty$, we have
		$\ESJD\to 2\ell^2\cdot\Phi\left(-\frac{\ell}{2}\sqrt{\EE_f\left[\left((\log f)'\right)^2\right]-1}\right)$.
	\end{theorem}
	\begin{proof}
		See \cref{proof-thm_ESJD}.
	\end{proof}
\begin{corollary}\label{coro_CLT}
The maximum ESJD is approximately $\frac{1.3}{\EE_f\left[\left((\log f)'\right)^2\right]-1}$, which is achieved when the acceptance rate is about $0.234$. The optimal 
	$\hat{\ell}\approx \frac{2.38}{\sqrt{\EE_f\left[\left((\log f)'\right)^2\right]-1}}$ and the optimal 
	$\hat{h}=\frac{1}{\sqrt{d-1}}\left[\left(1-\frac{\hat{\ell}^2}{2d}\right)^{-2}-1\right]^{1/2}\approx \frac{\hat{\ell}}{\sqrt{d(d-1)}} \approx \frac{2.38}{\sqrt{d(d-1)}}\frac{1}{\sqrt{\EE_f\left[\left((\log f)'\right)^2\right]-1}}$.
\end{corollary}

	\begin{remark}\label{remark-ESJD}
		Compared with the maximum ESJD of RWM, the maximum ESJD of SPS is $\frac{\EE_f\left[\left((\log f)'\right)^2\right]}{\EE_f\left[\left((\log f)'\right)^2\right]-1}$ times larger. {For example, for the class of distributions defined in \cref{remark_nu} indexed by $\nu$, we can compute $C_{\nu}/(C_{\nu}-1)$. See \cref{table-nu} for examples of $C_{\nu}/(C_{\nu}-1)$ for different $\nu$. One can see that the maximum ESJD of SPS can be much larger than the maximum ESJD of RWM even for a medium size of $\nu$.}
	\end{remark}

	\subsection{Diffusion Limit}
	Continuing our analogy to the Euclidean RWM case, we shall provide a diffusion limit result, giving a more explicit description of SPS for high-dimensional situations. However, it is difficult to obtain a diffusion limit directly for SPS.
	Instead, we shall slightly change the original SPS algorithm in a way that is asymptotically negligible (as $d$ gets large) but which greatly facilitates our limiting diffusion approach. The revised algorithm is called RSPS.

		\begin{algorithm}\label{our_algo}
		\caption{Revised Stereographic Projection Sampler (RSPS)}
		\begin{itemize}
			\item Let the current state be $X^d(t)=x$;
			\item Compute the proposal $\hat{X}$:
			\begin{itemize}
				\item Let $z:=\SP^{-1}(x)$;
				\item Sample independently $\dee\tilde{z}', \dee\tilde{z}''\sim \mathcal{N}(0,h^2I_{d+1})$;
				\item Let $\dee z':= \dee \tilde{z}-\frac{(z^T\cdot\dee\tilde{z}')z}{\|z\|^2}$ and $\dee z'':= \dee \tilde{z}-\frac{(z^T\cdot\dee\tilde{z}'')z}{\|z\|^2}$;
				\item Let $\hat{z}':=\frac{z+\dee z'}{\|z+\dee z'\|}$ and $\hat{z}'':=\frac{z+\dee z''}{\|z+\dee z''\|}$;
				\item Two independent proposals $\hat{X}':=\SP(\hat{z}')$ and $\hat{X}'':=\SP(\hat{z}'')$;
				\item The proposal $\hat{X}:=(\hat{X}_1',\hat{X}_{2:d}'')$.
			\end{itemize}
			\item $X^d(t+1)=\hat{X}$ with probability $1\wedge\frac{\pi(\hat{X})(R^2+\|\hat{X}\|^2)^d}{\pi(x)(R^2+\|x\|^2)^d}$; otherwise $X^d(t+1)=x$.
		\end{itemize}
	\end{algorithm}

   \BLUE{Note that the difference between SPS and RSPS is that in RSPS two independent proposals $\hat{X}'$ and $\hat{X}''$ are first computed. Then the final proposal is composed by the first coordinate of $\hat{X}'$ and other coordinates of $\hat{X}''$, i.e., $\hat{X}:=(\hat{X}_1',\hat{X}_{2:d}'')$. This design guarantees $\hat{X}_1$ is independent with $\hat{X}_{2:d}$ conditional on the current state, which is the technical condition needed to prove the following result on a diffusion limit.}

\begin{theorem}\label{thm_diffusion}
	Under the assumptions on $\pi$ in \cref{subsec_assumption}, suppose $f$ is not the standard Gaussian density and the RSPS chain $\{X^d(t)\}$ starts from the stationarity, i.e.~$X^d(0)\sim \pi$, and the radius parameter is chosen as $R=\sqrt{d}$. 
	Writing $X^d(t)=(X^d_1(t),\dots,X^d_d(t))$, we let $U^d(t):=X^d_1(\lfloor dt\rfloor)$ be the sequence of the first coordinates of $\{X^d(t)\}$ sped-up by a factor of $d$. Then, as $d\to\infty$, we have $U^d\Rightarrow U$, where $\Rightarrow$ denotes weak convergence in Skorokhod topology, and $U$ satisfies the following Langevin SDE
	\*[
	\dee U(t)=(s(\ell))^{1/2}\dee B(t)+ s(\ell) \frac{f'(U(t))}{2f(U(t))}\dee t,
	\]
	where	$s(\ell):=2\ell^2\Phi\left(-\ell\frac{\sqrt{\EE_f\left[\left((\log f)'\right)^2\right]-1}}{2}\right)$
	is the speed measure for the diffusion process, and $\Phi(\cdot)$ being the standard Gaussian cumulative density function.
\end{theorem}
\begin{proof}
	See \cref{proof-thm_diffusion}.
\end{proof}
\begin{corollary}
	The optimal acceptance rate for RSPS is $0.234$ and the maximum speed of the diffusion limit is
	$s(\hat{\ell})\approx \frac{1.3}{\EE_f\left[\left((\log f)'\right)^2\right]-1}$
	where
	$\hat{\ell}\approx \frac{2.38}{\sqrt{\EE_f\left[\left((\log f)'\right)^2\right]-1}}$. The optimal 
	$\hat{h}=\frac{1}{\sqrt{d-1}}\left[\left(1-\frac{\hat{\ell}^2}{2d}\right)^{-2}-1\right]^{1/2}\approx \frac{\hat{\ell}}{\sqrt{d(d-1)}} \approx \frac{2.38}{\sqrt{d(d-1)}}\frac{1}{\sqrt{\EE_f\left[\left((\log f)'\right)^2\right]-1}}$.
\end{corollary}

	\begin{remark}
		Compared with the maximum speed of the diffusion limit of RWM \citep{Roberts1997}, the maximum speed of the diffusion limit of RSPS is $\frac{\EE_f\left[\left((\log f)'\right)^2\right]}{\EE_f\left[\left((\log f)'\right)^2\right]-1}$ times larger.
	\end{remark}
	\begin{remark}
		A reasonable conjecture is that the same diffusion limit holds for the original SPS algorithm. In order to establish the same diffusion limit, if we follow the same arguments as in the proof of \cref{thm_ESJD}, it is required to show the acceptance rate term becomes ``asymptotically independent'' with the first coordinate as a rate of $\bigO(d^{-1/2})$. However, our current technical arguments in \cref{thm_ESJD} can achieve a rate of $\bigO(d^{-1/8})$ which is not enough for establishing the weak convergence to a diffusion limit. Therefore, we only prove the diffusion limit for RSPS in this paper and leave the proof for SPS as an open problem.
	\end{remark}

\section{Simulations}\label{section_simu}

In this section, we study the proposed SPS and SBPS through numerical examples. In most of the examples, we consider $d=100$ dimensions and two choices of target distributions, the heavy-tailed multivariate student's $t$ target with $d$ degree of freedom and the standard Gaussian target. By default, we choose $R=\sqrt{d}$ for SPS and SBPS.
We refer to \cref{sec_more_simu} for additional simulations such as different choices of $R$.

\BLUE{
\subsection{SPS: Bayesian Cauchy regression}\label{subsec_example_Cauchy}
Here we return to \cref{example_Cauchy}. In \cref{fig-Cauchy-Regression}, we choose $a=b=0.1$, $d=11$, $n=15$, and $R=\sqrt{d}$. Random design $X_i\sim \mathcal{N}(0,I)\in \Reals^d$ and responses $\{Y_i\}_{i=1}^{n}$ are generated by $Y_i=\alpha_0+\beta_0^TX_i+\epsilon_i$ where $\alpha_{0}=-2$, $\beta_{0}=(-4,-3,-2,-1,0,1,2,3,4)^T$, $\{\epsilon_i\}$ are i.i.d.~zero-mean Cauchy distribution with scale $\gamma_0=1$. 
We take a logarithm transformation for $\gamma$ and implement both SPS and RWM in $\Reals^{d}$. RWM starts from $(100,100,\dots,100)$ and SPS starts from the North Pole. From the figure, one can see clearly that: (1) RWM which is not geometric ergodic completely failed; (2) the proposed SPS which is uniformly ergodic converged extremely fast.
}

\subsection{SPS: ESJD per dimension}
In this example, we study the ESJD for SPS and its robustness to the choice of the radius $R$. We first tune the proposal stepsize $h$ of SPS to get different acceptance rates. Then we plot ESJD per dimension for varying acceptance rates as the efficiency curve of SPS. \cref{fig-efficiency-curve} shows eight efficiency curves for different choices of $R$. The target distribution is multivariate student's $t$ distribution with $d=100$ degrees of freedom. In this setting, $R=\sqrt{d}$ is optimal. The four subplots in the first column are for $R<\sqrt{d}$ and the second column contains four cases of $R>\sqrt{d}$.

Although we do not plot the ESJD for RWM, the maximum ESJD per dimension for RWM is known to have an order of $\bigO(d^{-1})$. In all cases in \cref{fig-efficiency-curve}, SPS has a much larger ESJD than RWM. For the two subplots in the first row of \cref{fig-efficiency-curve}, since $R$ is closer to $\sqrt{d}$, the acceptance rate cannot be lower than $0.5$ whatever the proposal variance is. For all the other efficiency curves, an interesting observation is that the maximum ESJDs are always achieved when the acceptance rate is around $0.234$. This suggests the optimal acceptance rate $0.234$ is quite robust and not limited to product i.i.d.~targets, which is similar to the case of optimal acceptance rate $0.234$ for RWM \citep{Yang2020optimal}.

		\begin{figure}[h]
		\centering
		\includegraphics[width=0.95\textwidth]{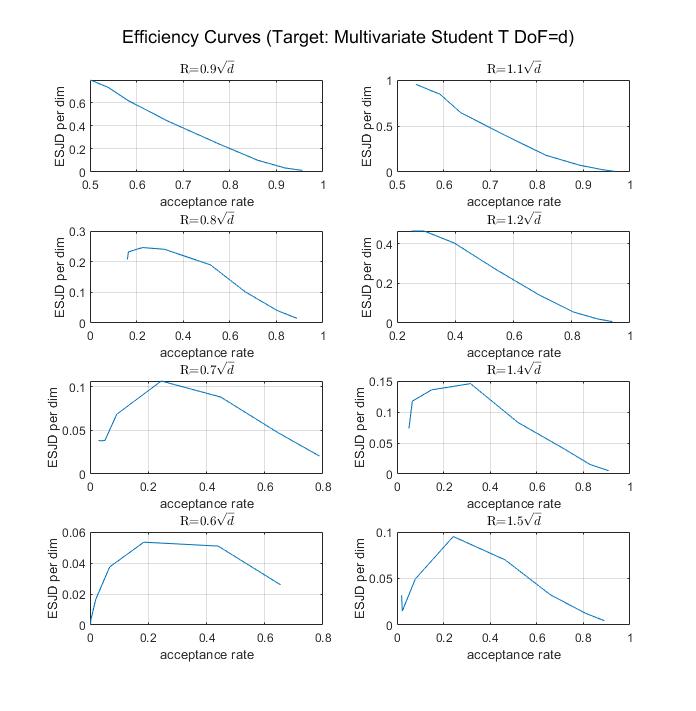}
		\caption{Efficiency Curves (ESJD per dimension) as functions of the acceptance rate of SPS for different choices of $R$. Target distribution is multivariate student's $t$ distribution with DoF$=d=100$. When $R$ is close to $\sqrt{d}$ (such as $R=0.9\sqrt{d}$ and $R=1.1\sqrt{d}$), the acceptance rate is always larger than $0.234$.
		}
		\label{fig-efficiency-curve}
	\end{figure}

\subsection{SBPS: ESS per Switch}
In this example, we study the efficiency curves of SBPS and BPS in terms of ESS per Switch versus the refresh rate. The first subplot of \cref{fig-pdmp-ESSpS} contains the proportion of refreshments in all the $N$ events for varying refresh rates. It is clear that as the refresh rate increases, the proportion of refreshments increases. In the other three subplots of \cref{fig-pdmp-ESSpS}, we plot the logarithm of ESS per Switch as a function of the refresh rate for three cases, the $1$-st coordinate, the negative log-density, and the squared $1$-st coordinate, respectively. For each efficiency curve, $N=1000$ events are simulated. We use random initial states for our SBPS. For comparison, BPS starts from stationarity to avoid slow mixing. As SBPS and BPS are continuous-time processes, each unit time period is discretized into $5$ samples. The target is standard Gaussian with $d=100$.  We refer to \cref{sec_more_simu} for additional efficiency curves for heavy-tailed targets.

According to \cref{fig-pdmp-ESSpS}, the ESS per Switch of SBPS is much larger than the ESS per Switch of BPS for all cases (actually the gap becomes larger in higher dimensions). For both the $1$-st coordinate and the negative log-target density, the ESS per Switch of SBPS can be larger than $1$ if the refresh rate is relatively low. For BPS, however, even starting from stationarity, the ESS per Switch is always much smaller than $1$.

	\begin{figure}[h]
		\centering
		\includegraphics[width=0.9\textwidth]{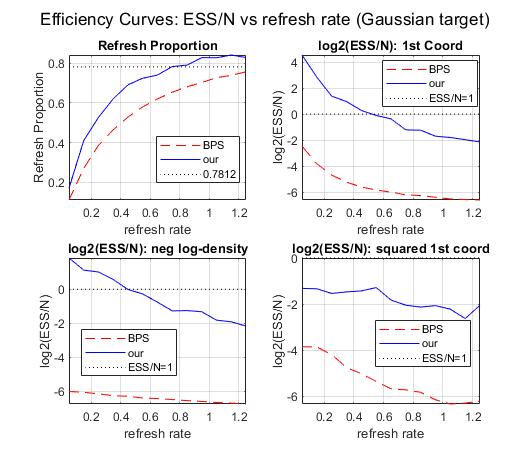}
		\caption{Efficiency curves (ESS per Switch) for SBPS and BPS for varying refresh rate.  $N=1000$ events are simulated, Random initial values for SBPS and BPS starting from stationarity. Each unit time is discretized to $5$ samples. Target distribution:  standard Gaussian with $d=100$. The first subplot is the proportion of refreshment events in all $N$ events. The other three subplots are ESS for the $1$-st coordinate, the negative log density, and the squared $1$-st coordinate. 
		}
		\label{fig-pdmp-ESSpS}
	\end{figure}

\section{Discussion}\label{section_discuss}

We have explored the use of stereographically projected algorithms and developed SPS and SBPS that are uniformly ergodic for a large class of light and heavy-tailed targets and also exhibit fast convergence in high-dimensions. The framework established opens new opportunities for developing MCMC algorithms for sampling high-dimensional and/or heavy-tailed distributions. We finish the paper with some future directions.

\begin{itemize}
    \item \emph{Adaptive MCMC}: the proposed GSPS algorithm fits the adaptive MCMC framework very well \citep{roberts2009examples}. The tuning parameters of GSPS such as the covariance matrix, the location as well as the radius of the sphere can be tuned adaptively. Empirical study of the sensitivity of GSPS to the tuning parameters and extending GSPS for sampling multimodal distributions \citep{pompe2020framework} are important future directions. 
    \item \emph{Quantitative bounds}: we know SPS is uniform ergodicity for a large class of targets and dimension-free for isotropic targets. However, we do not establish quantitative bounds for the mixing time. The quantitative bound and its dimension dependence (e.g., \citep{Yang2017}) would be an interesting direction for future work.
    \item \emph{Scaling limit for SBPS}: we know SBPS is dimension-free for isotropic targets. For product i.i.d.~targets, we establish optimal scaling results for SPS but not for SBPS. Recently, scaling limits for the traditional BPS have been developed in different settings \citep{bierkens2018high,deligiannidis2021randomized}. It would be interesting to obtain such scaling limits for SBPS for non-isotropic targets to compare with BPS directly.
    \item \emph{Stereographic MALA, HMC, and others}: another future direction is to develop new MCMC samplers based on other popular MCMC algorithms, such as (Riemann) MALA and Hamiltonian Monte Carlo (HMC) \citep{girolami2011riemann} and other PDMPs, or based on other mappings from $\Reals^d$ to $\mathbb{S}^d$ than the (generalized) stereographic projection.
\end{itemize}

\section*{Acknowledgement}
JY was supported by Florence Nightingale Fellowship, Lockey Fund, and St Peter’s College Research Fund (O'Connor Fund) from University of Oxford. K{\L} was supported by a Royal Society University Research Fellowship. GOR was supported by EPSRC grants Bayes for Health (R018561) CoSInES (R034710) PINCODE (EP/X028119/1), EP/V009478/1 and by the UKRI grant, OCEAN, EP/Y014650/1.






\clearpage

\renewcommand{\thelemma}{S\arabic{lemma}}
\renewcommand{\thetheorem}{S\arabic{theorem}}
\renewcommand{\thesection}{S\arabic{section}}

\begin{center}
\Large{Supplement to ``\TITLE{}''} \\ \medskip
\normalsize{
Jun Yang, Krzysztof Łatuszyński, and Gareth O.\ Roberts
}
\end{center}


\section{Proofs of Main Results}\label{section_proof}

\subsection{Proof of \cref{prop-SPSergodic}}\label{proof-prop-SPSergodic}
\begin{proof}
If $\pi(x)>0, \forall x\in \mathbb{R}^d$ and $\pi$ being continuous, then $\pi_S(z)$ is continuous and $\pi_S(z)>0$ for all $z\in \mathbb{S}^d$ except the North Pole. According to \citeps[Theorem 13.0.1]{Meyn2012}, it suffices to prove the existence of a stationary distribution, aperiodicity, $\pi$-irreducibility, and positive Harris recurrence.

The existence of a stationary distribution comes from the detailed balance of Metropolis--Hastings chain. SPS is aperiodic by \citeps[Theorem 3(i)]{roberts1994simple} using the facts that the acceptance probability does not equal to zero and the proposal distribution as a Markov kernel is aperiodic. Furthermore, $\pi$-irreducibility comes from \citeps[Theorem 3(ii)]{roberts1994simple} using the fact that the proposal distribution as a Markov kernel is $\pi$-irreducible. Since $\pi$ is finite, the chain is positive recurrent \cites[pp.1712]{tierney1994markov}. Harris recurrence comes directly from \cites[Corollary 2]{tierney1994markov}.

 \end{proof}
 
\subsection{Proof of \cref{thm-uniform-ergodic}}\label{proof-thm-uniform-ergodic}
\begin{proof}
	If there exists $x$ such that $\sup_{x\in \Reals^d} \pi(x)(R^2+\|x\|^2)^d=\infty$, then the chain is not even geometric ergodic by \citeps[Proposition 5.1]{Roberts1996}. This implies that the condition of \cref{cond_uniform_ergodic} is necessary. Therefore, it suffices to prove  \cref{cond_uniform_ergodic} is also a sufficient condition.
	
	We prove uniform ergodicity by showing the whole $\mathbb{S}^d$ is a small set. More precisely, we will show the $3$-step transitional kernel satisfies the following minorization condition:
	\*[
	P^3(z,A)\ge \epsilon Q(A)
	\]
	for all measurable set $A$, where $Q(\cdot)$ is a fixed probability measure on $\mathbb{S}^d$ and $\epsilon>0$. \BLUE{The rationale behind establishing a $3$-step minorization condition is that this facilitates 
 the SPS chain being able to reach any point on the sphere, 
  and there exists enough probability mass around ``$2$-stop'' paths to ensure that the chain can avoid stopping too close to the North Pole, the potentially problematic region.}
 
	We use the following notations:
	\begin{enumerate}
		\item We {use $\textrm{AC}(\epsilon):=\{z\in\mathbb{S}^d: z_{d+1}\ge 1-\epsilon\}$ to denote the ``Arctic Circle'' with ``size'' $\epsilon\ge 0$}. 
		\item We {use $\textrm{HS}(z,0):=\{z'\in\mathbb{S}^d: z^Tz'\ge 0\}$ to denote the ``hemisphere'' from $z$. and $\textrm{HS}(z,\epsilon'):=\{z'\in\mathbb{S}^d: z^Tz'\ge \epsilon'\}$ to denote the area left after ``cutting'' the hemisphere's boundary of ``size'' $\epsilon'\ge 0$}. 
	\end{enumerate}

	Now consider the $3$-step path $z\to z_1\to z_2\to z'\in A$ where $z\in\mathbb{S}^d$ is the starting point, $z_1$ and $z_2$ are two intermediate points, and $z'\in A$ is final point. 
		Denoting $q(z,\cdot)$ as the proposal density, we have
	\*[
	p(z,z')\ge q(z,z')\left(1\wedge \frac{\pi_S(z')}{\pi_S(z)}\right).
	\]
	Then the  $3$-step transition kernel can be bounded below by
\*[
	P^3(z,A)
	&\ge\int_{z'\in A}\int_{z_2\in\mathbb{S}^d}\int_{z_1\in\mathbb{S}^d} q(z,z_1)q(z_1,z_2)q(z_2,z')\\
	&\qquad \cdot\left(1\wedge\frac{\pi_S(z_1)}{\pi_S(z)}\right)\left(1\wedge\frac{\pi_S(z_2)}{\pi_S(z_1)}\right)\left(1\wedge\frac{\pi_S(z')}{\pi_S(z_2)}\right)\nu(\dee z_1)\nu(\dee z_2)\nu(\dee z') 
\]
	where $\nu(\cdot)$ denotes the Lebesgue measure on $\mathbb{S}^d$. 

	By the conditions of \cref{thm-uniform-ergodic}, $\pi(x)>0$ is continuous in $\Reals^d$, which implies $\pi_S(z)$ is positive and continuous in any compact subset of $\mathbb{S}^d\setminus N$ (recall that $N$ denotes the North Pole). Then using the fact that a continuous function on a compact set is bounded, if we rule out $\textrm{AC}(\epsilon)$, then $\pi_S(z)$ is bounded away from $0$. That is,
	\*[
	\inf_{z\notin \textrm{AC}(\epsilon)}\pi_S(z)\ge \delta_{\epsilon}>0, \forall \epsilon\in (0,1),
	\]
	where $\delta_{\epsilon}\to 0$ as $\epsilon\to 0$.
	
	Furthermore, since the surface area of $\mathbb{S}^d$ is finite and $\pi_S(z)\propto \pi(x)(R^2+\|x\|^2)^d$ where $x=\SP(z)$, we know
	$\sup_{x\in \Reals^d}\pi(x)(R^2+\|x\|^2)^d<\infty$ implies that $M:=\sup \pi_S(z)<\infty$. Therefore, we have
\*[
	P^3(z,A)
	&\ge\int_{z'\in A}\int_{z_2\in\mathbb{S}^d\setminus \textrm{AC}(\epsilon)}\int_{z_1\in\mathbb{S}^d\setminus\textrm{AC}(\epsilon)} q(z,z_1)q(z_1,z_2)q(z_2,z')\\
	&\qquad \cdot\left(1\wedge\frac{\delta_{\epsilon}}{M}\right)^2\left(1\wedge\frac{\pi_S(z')}{M}\right)\nu(\dee z_1)\nu(\dee z_2)\nu(\dee z').
\]
	Next, we consider the term 
	$q(z,z_1)q(z_1,z_2)q(z_2,z')$.
	For our algorithm, the proposal can cover the whole hemisphere except the boundary, by ``cutting'' the boundary a little bit, the proposal density will be bounded below. That is, we have
	\*[
	q(z,z')\ge \delta'_{\epsilon'}>0, \quad \forall z'\in \textrm{HS}(z,\epsilon'), \epsilon'\in (0,1),
	\]
	where $\delta'_{\epsilon'}\to 0$ as $\epsilon'\to 0$.
	
	Therefore, if 
	$z_1\in \textrm{HS}(z,\epsilon')$ and $z_2\in \textrm{HS}(z',\epsilon')$
	then 
	$q(z,z_1)\ge \delta'_{\epsilon'}$ and $q(z_2,z')\ge \delta'_{\epsilon'}$.
	Then we have
\*[
	&\int_{z_2\in\textrm{HS}(z',\epsilon')\setminus \textrm{AC}(\epsilon)}\int_{z_1\in\textrm{HS}(z,\epsilon')\setminus\textrm{AC}(\epsilon)} q(z,z_1)q(z_1,z_2)q(z_2,z')\nu(\dee z_1)\nu(\dee z_2)\\
	&\ge (\delta'_{\epsilon'})^3\int_{z_2\in\textrm{HS}(z',\epsilon')\setminus \textrm{AC}(\epsilon)}\int_{z_1\in\textrm{HS}(z,\epsilon')\setminus\textrm{AC}(\epsilon)}\Ind_{z_1\in \textrm{HS}(z_2,\epsilon')}\nu(\dee z_1)\nu(\dee z_2).
\]

	Note that $3$ steps are enough to reach any point $z'$ from any $z$ on the sphere. It is clear that under the Lebesgue measure $\nu$ on $\mathbb{S}^d$, we can define the following positive constant $C$:
	\*[
	C:=\inf_{z,z'\in \mathbb{S}^d}\int_{z_2\in\textrm{HS}(z',0)\setminus \textrm{AC}(0)}\int_{z_1\in\textrm{HS}(z,0)\setminus\textrm{AC}(0)}\Ind_{z_1\in \textrm{HS}(z_2,0)}\nu(\dee z_1)\nu(\dee z_2)>0.
	\]
	Now consider the following function of $\epsilon\ge 0$ and $\epsilon'\ge 0$
	\*[
	C_{\epsilon,\epsilon'}:=\inf_{z,z'\in\mathbb{S}^d}\int_{z_2\in\textrm{HS}(z',\epsilon')\setminus \textrm{AC}(\epsilon)}\int_{z_1\in\textrm{HS}(z,\epsilon')\setminus\textrm{AC}(\epsilon)}\Ind_{z_1\in \textrm{HS}(z_2,\epsilon')}\nu(\dee z_1)\nu(\dee z_2).
	\]
    It is clear that $C\ge 	C_{\epsilon,\epsilon'}$. The difference between them can be bounded by
\*[
    C-	C_{\epsilon,\epsilon'}\le &\sup_{z\in \mathbb{S}^d} \nu(\{z_1: z_1\in (\textrm{HS}(z,0)\setminus\textrm{HS}(z,\epsilon'))\cap (\textrm{AC}(\epsilon)\setminus\textrm{AC}(0))\})\\
    &+\sup_{z'\in \mathbb{S}^d}\nu(\{z_2: z_2\in (\textrm{HS}(z',0)\setminus\textrm{HS}(z',\epsilon'))\cap (\textrm{AC}(\epsilon)\setminus\textrm{AC}(0))\})\\
    &+\sup_{z_2\in \mathbb{S}^d}\nu(\{z_1: z_1\in \textrm{HS}(z_2,0)\setminus\textrm{HS}(z_2,\epsilon')\}).
\]
	Then clearly the upper bound is continuous w.r.t.~both $\epsilon$ and $\epsilon'$ and goes to zero when $\epsilon\to 0$ and $\epsilon'\to 0$. Therefore, there exists $\epsilon>0$ and $\epsilon'>0$ such that $C-C_{\epsilon,\epsilon'}\le \frac{1}{2} C$, that is
	\*[
	C_{\epsilon,\epsilon'}\ge\frac{1}{2}C>0.
	\]
	Using such $\epsilon$ and $\epsilon'$ we have
	\*[
	&\int_{z_2\in\textrm{HS}(z',\epsilon')\setminus \textrm{AC}(\epsilon)}\int_{z_1\in\textrm{HS}(z,\epsilon')\setminus\textrm{AC}(\epsilon)} q(z,z_1)q(z_1,z_2)q(z_2,z')\nu(\dee z_1)\nu(\dee z_2)\ge \frac{1}{2}C(\delta'_{\epsilon'})^3>0.
	\]
	Then we have
	\*[
	P^3(z,A)&\ge \left(1\wedge \frac{\delta_{\epsilon}}{M}\right)^2\frac{1}{2}C(\delta'_{\epsilon'})^3\int_{z'\in A} \left(1\wedge \frac{\pi_S(z')}{M}\right)\nu(\dee z')
	\]
	Note that $\pi_S(z')>0$ almost surely w.r.t.~$\nu$.
	Denoting
	\*[
	g(z'):=\left(1\wedge \frac{\delta_{\epsilon}}{M}\right)^2\frac{1}{2}C(\delta'_{\epsilon'})^3\left(1\wedge \frac{\pi_S(z')}{M}\right),
	\]
	we have established
	\*[
	P^3(z,A)\ge \int_{z'\in A}g(z')\nu(\dee z'),
	\]
	where $g(z)>0$ almost surely w.r.t.~$\nu$. Then we can define	
	\*[
	\int_A g(z')\nu(\dee z')=\int_{\mathbb{S}^d} g(z')\nu(\dee z')\frac{\int_{A} g(z')\nu(\dee z')}{\int_{\mathbb{S}^d} g(z')\nu(\dee z')}=:\epsilon Q(A),
	\]
	where $\epsilon=\int_{\mathbb{S}^d}g(z')\nu(\dee z')>0$ and $Q(A):=\frac{\int_{A} g(z')\nu(\dee z')}{\int_{\mathbb{S}^d} g(z')\nu(\dee z')}$ is a probability measure. This established the desired minorization condition and completed the proof.
\end{proof}

\subsection{Sketch proof of \cref{prop-SBPSergodic}}\label{proof-prop-SBPSergodic}
\begin{proof}
In this proof, for two vectors $a$ and $b$, we use $a\cdot b$ to denote the inner product.
We denote the generator of our PDMP and its adjoint (see \citeps[Section 2.2]{Fearnhead2018piecewise} or \citeps{liggett2010continuous}) as $\mathcal{A}$ and $\mathcal{A}^*$, respectively.
Then, we can divide the generator $\mathcal{A}$ into two parts
\*[
\mathcal{A}=\mathcal{A}_D + \mathcal{A}_J,
\]
where the first term $\mathcal{A}_D$ is the generator of the deterministic dynamics (moving along the greatest circle) and the second term $\mathcal{A}_J$ relates the jumping process (that is, the change in expectation at bounce or refreshment events).

Then it can be easily derived that, for suitable functions $f(z,v)$ and $g(z,v)$, where $v\perp z$ and $\|v\|=\|z\|=1$, the generator $\mathcal{A}_D$ and its adjoint $\mathcal{A}_D^*$ of the deterministic process satisfy
\*[
\mathcal{A}_D f&= \nabla_z f \cdot v - \nabla_v f\cdot z,\\
\mathcal{A}_D^* g &= \nabla_v g\cdot z - \nabla_z g \cdot v.
\]
		 Denoting $|\mathbb{S}^{d-1}|$ as the volume of the $\mathbb{S}^{d-1}$, we substitute the following joint distribution on $\mathbb{S}^d\times \mathbb{S}^d$:
		\*[
		\pi(z,v)=\pi_S(z)p(v\mid z)=\pi_S(z)\frac{\Ind_{v\cdot z=0}}{|\mathbb{S}^{d-1}|},\quad \forall (z,v) \in \mathbb{S}^d\times \mathbb{S}^d.
		\]
		For any $ (z,v) \in \mathbb{S}^d\times \mathbb{S}^d$ such that $z\perp v$, we have
\*[
		\mathcal{A}_D^*[\pi(z,v)]
		&=0 \cdot z-\frac{1}{|\mathbb{S}^{d-1}|} \nabla_z \pi_S(z)\cdot v\\
		&=-\pi_S(z)\frac{\Ind_{v\cdot z=0}}{|\mathbb{S}^{d-1}|}[v\cdot \nabla_z\log \pi_S(z)]\\
		&=-\pi(z,v)[v\cdot \nabla_z\log \pi_S(z)].
\]
		
		For $\mathcal{A}_J$, since the refreshment clearly preserves $\pi$ as the stationary distribution, we only need to consider the bounce events. When there is no refreshment, we can follow similar arguments as in \citeps[Section 2.2 and 3.2]{Fearnhead2018piecewise} to get the adjoint of the second part of the generator
		\*[
		\mathcal{A}^*_J[\pi(z,v)]=-\pi(z,v)\lambda(z,v)+\int \lambda(z,v')q(v\mid z,v')\pi(z,v')\dee v',
		\]
			where
		\*[
		q(\cdot \mid z,v):=\delta_{\tilde{v}}(\cdot),\quad 	\tilde{v}:=v-2\left[\frac{v\cdot\tilde{\nabla}_z\log\pi_S(z)}{\tilde{\nabla}_z\log\pi_S(z)\cdot \tilde{\nabla}_z\log\pi_S(z)}\right] \tilde{\nabla}_z\log\pi_S(z).
		\]
	    Therefore, it suffices to verify that $\pi(z,v)$ satisfies $\mathcal{A}^*[\pi(z,v)]=\mathcal{A}_D^*[\pi(z,v)]+\mathcal{A}_J^*[\pi(z,v)]=0$. That is,
		\*[
		\mathcal{A}^*[\pi(z,v)]=-\pi(z,v)[v\cdot \nabla_z\log\pi_S(z)+\lambda(z,v)]+\int \lambda(z,v')q(v\mid z,v')\pi(z,v')\dee v'.
		\]
		Then it is enough to verify
		\*[
		p(v\mid z)[\lambda(z,v)+v\cdot \nabla_z\log\pi_S(z)]-\lambda(z,\tilde{v})p(\tilde{v}\mid z)=0.
		\]
		Note that $\tilde{v}\perp z$ which implies that $p(v\mid z)=p(\tilde{v}\mid z)$. Therefore, it suffices to verify
		\*[
		\lambda(z,v)-\lambda(z,\tilde{v})=-v\cdot \nabla_z\log\pi_S(z).
		\]
		Recall that in our PDMP, we have $\lambda(z,v)=\max\{0,\left[-v\cdot\nabla_z\log\pi_S(z)\right]\}$. Therefore, we only need to verify 
		\*[
		-v\cdot \nabla_z\log\pi_S(z)=\tilde{v}\cdot \nabla_z\log\pi_S(z).
		\]
    Using the definition of $\tilde{v}$ and
	$\tilde{\nabla}_z\log\pi_S(z)=\nabla_z\log\pi_S(z)-\left[z\cdot\nabla_z\log\pi_S(z)\right]z$,
	together with the fact that $v\perp z$ and $\tilde{v}\perp z$, we have
	\*[
	-v\cdot \nabla_z\log\pi_S(z)=-v\cdot \tilde{\nabla}_z\log\pi_S(z)=\tilde{v}\cdot \tilde{\nabla}_z\log\pi_S(z)=\tilde{v}\cdot \nabla_z\log\pi_S(z),
	\]
	which completes the proof.
	\end{proof}
	
\subsection{Proof of \cref{thm_PDMP}}\label{proof_thm_PDMP}

\begin{proof}
Recall the definition of the transformation between $x\in \Reals^d$ and $z\in \mathbb{S}^d$:
	\*[
	x_i=R\frac{z_i}{1-z_{d+1}},\quad
	z_i=\frac{2Rx_i}{R^2+\|x\|^2}, \quad \forall i=1,\dots,d,\quad
	z_{d+1}=\frac{\|x\|^2-R^2}{\|x\|^2+R^2},
	\]
	and the transformed target on $\mathbb{S}^d$ is
	$\pi_S(z)\propto \pi(x)(R^2+\|x\|^2)^d$.
	One version of $\nabla_z\log \pi_S(z)$ (there is no unique representation since $z\in \Reals^{d+1}$ but $x\in \Reals^d$) can be derived as
	\*[
	\frac{\partial \log \pi_S(z)}{\partial z_i}&=\frac{\partial \log \pi(x)}{\partial x_i}\frac{R^2+\|x\|^2}{2R},\quad i=1,\dots,d.\\
	\frac{\partial \log \pi_S(z)}{\partial z_{d+1}}
	&=\left[\frac{1}{d}\sum_{j=1}^{d}\frac{\partial \log \pi(x)}{\partial x_j}x_j+1\right]\frac{d}{2R^2}(R^2+\|x\|^2).
	\]
	Recall that we define $\tilde{\nabla}_z\log\pi_S(z)$ as the projection of (any version of) $\nabla_z\log \pi_S(z)$ onto the (tangent plane of the) sphere:
	\*[
	\tilde{\nabla}_z\log\pi_S(z)&=\nabla_z\log\pi_S(z)-\left[z\cdot\nabla_z\log\pi_S(z)\right]z.
	\]
	Then the representation of $\tilde{\nabla}_z\log\pi_S(z)$ is unique.
	
	{	Denote $\left(\cdot\right)_{d+1}$ as the $(d+1)$-th coordinate. One can verify that
	\*[
	\limsup_{\{z: z_{d+1}\to 1\}}\left(\tilde{\nabla}_z\log\pi_S(z)\right)_{d+1}=\limsup_{\{x: \|x\|\to \infty\}} \sum_{i=1}^d \left(\frac{\partial \log \pi(x)}{\partial x_i} x_i\right)+2d.
	\]
	Therefore, the condition is equivalent to
		\*[
	\limsup_{\{z: z_{d+1}\to 1\}}\left(\tilde{\nabla}_z\log\pi_S(z)\right)_{d+1}< \frac{1}{2}.
	\]
	}
	
We have assumed $\|v\|=1$. Define the Lyapunov function
\*[
f(z,v)=2+v_{d+1},
\]
which is bounded since $|v_{d+1}|\le 1$. 
Note that moving along the greatest circle on the sphere with constant speed $\|v\|=1$ implies the equations
\*[
z(t)&=z(0)\cos(t)+v(0)\sin(t),\\
v(t)&=v(0)\cos(t)-z(0)\sin(t).
\]
For simplicity, we will write $z(0)$ and $v(0)$ as $z$ and $v$ in the rest of the proof. For two vectors $a$ and $b$, we will use $a \cdot b$ to denote the inner product $a^Tb$.
The key property we will use is that $\frac{\dee f}{\dee t}=\frac{\dee v_{d+1}}{\dee t}=-z_{d+1}$. 

Now we can derive the (extended) generator of SBPS
\*[
\tilde{\mathcal{L}}f(z,v):=&\lim_{t\to 0^+} \frac{1}{t}\left[f(z(t),v(t))-f(z,v)\right]\\
&+\int \lambda \psi(\dee w)\left[f(z,w)-f(z,v)\right]\\
&+ \left(-\nabla_z \log\pi_S(z)\cdot v\right)^+\left[f(z,\tilde{v})-f(z,v)\right]\\
=& -z_{d+1} - \lambda v_{d+1} + \left(-\nabla_z \log\pi_S(z)\cdot v\right)^+\left(\tilde{v}_{d+1}-v_{d+1}\right),
\]
where we have used
	\*[
\int \lambda \psi(\dee w)\left[f(z,w)-f(z,v)\right] = \int \lambda \psi(\dee w)\left[w_{d+1}-v_{d+1}\right]=-\lambda v_{d+1}.
\]
Our goal is to establish a drift condition for the generator.
Note that
\*[
&\tilde{\mathcal{L}}f/f=\tilde{\mathcal{L}}f/(2+v_{d+1})\le \max\{\tilde{\mathcal{L}}f, \tilde{\mathcal{L}}f/3 \}
\]
Therefore, we only need to focus on $\tilde{\mathcal{L}}f$ and prove it is negative outside a small set. It suffices to prove $\tilde{\mathcal{L}}f$ is negative in a small neighborhood of the North Pole.

Recall that
	\*[
	\tilde{v}:=v-2\left[\frac{v\cdot\tilde{\nabla}_z\log\pi_S(z)}{\tilde{\nabla}_z\log\pi_S(z)\cdot \tilde{\nabla}_z\log\pi_S(z)}\right] \tilde{\nabla}_z\log\pi_S(z).
	\]
Also, it can be verified that
	\*[
	v\cdot\tilde{\nabla}_z\log\pi_S(z)=v\cdot\nabla_z\log\pi_S(z).
	\]
Therefore, we have
\*[
\tilde{\mathcal{L}}f(z,v)= -z_{d+1} - \lambda v_{d+1} + \left(-\tilde{\nabla}_z \log\pi_S(z)\cdot v\right)^+\left[-2 \frac{v\cdot\tilde{\nabla}_z\log\pi_S(z)}{\|\tilde{\nabla}_z\log\pi_S(z)\|^2} \left(\tilde{\nabla}_z\log\pi_S(z)\right)_{d+1}\right].
\]
Note that if $\left(\tilde{\nabla}_z\log\pi_S(z)\right)_{d+1}\le 0$ then $\tilde{\mathcal{L}}f(z,v)\le -z_{d+1} - \lambda v_{d+1}$. Otherwise, $\left(\tilde{\nabla}_z\log\pi_S(z)\right)_{d+1}>0$, then we maximize the last term over $v$ to get
\*[
&\left(-\tilde{\nabla}_z \log\pi_S(z)\cdot v\right)_+\left[-2 \frac{v\cdot\tilde{\nabla}_z\log\pi_S(z)}{\|\tilde{\nabla}_z\log\pi_S(z)\|^2} \left(\tilde{\nabla}_z\log\pi_S(z)\right)_{d+1}\right]\\
&\le \|\tilde{\nabla}_z \log\pi_S(z)\| \left[2\frac{\|\tilde{\nabla}_z \log\pi_S(z)\|}{\|\tilde{\nabla}_z \log\pi_S(z)\|^2}\left(\tilde{\nabla}_z\log\pi_S(z)\right)_{d+1}\right]\\
&=2\left(\tilde{\nabla}_z\log\pi_S(z)\right)_{d+1}
\]
where the upper bound is achieved if
\*[
v=\frac{-\tilde{\nabla}_z \log\pi_S(z)}{\|\tilde{\nabla}_z \log\pi_S(z)\|}.
\]

Overall, using $v_{d+1}\le \sqrt{1-z_{d+1}^2}$, we have already shown that
\*[
\sup_{\{z: z_{d+1}=z_*\}}\tilde{\mathcal{L}}f(z,v)\le -z_* + \lambda \sqrt{1-z_*^2} +2\sup_{\{z: z_{d+1}=z_*\}}\left(\tilde{\nabla}_z\log\pi_S(z)\right)_{d+1}^+
\]
Finally, let $z_*\to 1$,  the right hand side goes to
\*[
-1 +2\limsup_{\{z: z_{d+1}\to 1\}}\left(\tilde{\nabla}_z\log\pi_S(z)\right)_{d+1}^+
\]
which is negative because of the condition in \cref{thm_ESJD} that
\*[
\limsup_{\{z: z_{d+1}\to 1\}}\left(\tilde{\nabla}_z\log\pi_S(z)\right)_{d+1}< \frac{1}{2}.
\]
Therefore, we have shown the (extended) generator is negative in a neighborhood of the North Pole, which implies a drift condition for the (extended) generator. Finally, by \citeps[Theorem 5.2(c)]{down1995exponential}, as the Lyapunov function is bounded, the Markov process is uniformly ergodic, which completes the proof. 
\end{proof}

\subsection{Proof of \cref{coro_PDMP}}\label{proof_coro_PDMP}
\begin{proof}
    We present two versions of the proof.
  In this first version of the proof, we simply check the condition in \cref{thm_PDMP}. Let $\nu$ to be the DoF of the multivariate student's $t$ target $\pi(x)$. We have
	\*[
	\log\pi(x) = -\frac{\nu+d}{2}\log\left(1+\frac{1}{\nu}\|x\|^2\right)+C
	\]
	for some constant $C$.
	Then, one can get
	\*[
	\sum_{i=1}^d\frac{\partial \log\pi (x)}{\partial x_i}x_i=-\frac{\nu+d}{\nu}\frac{\|x\|^2}{1+\|x\|^2/\nu}.
	\]
	Taking $\|x\|\to \infty$ yields
	\*[
	-\frac{\nu+d}{\nu}\nu < \frac{1}{2}-2d
	\]
	which completes the proof since
	\*[
	\nu>d-\frac{1}{2}.
	\]

	The alternative proof is by following the proof of \cref{thm_PDMP}.
Let the degree of freedom for the multivariate student's $t$ target $\pi(x)$ to be $kd$ where $k>1-\frac{1}{2d}$.
	For simplicity of notations, we consider the case $R=\sqrt{d}$ without loss of generality.
	Since $\pi(x)$ is isotropic, we can verify that {if $R=\sqrt{d}$}, one version of $\nabla\log \pi_S(z)$ (as the representation is not unique) is
	\*[
	\nabla\log \pi_S(z)=\left(0,\dots, 0, \frac{-dz_{d+1}(k-1)}{(1-z_{d+1})[(k-1)(1-z_{d+1})+2]}\right)
	\]
	Then it is easy to verify using the basic geometry of the sphere that if $k\ge 1$ and $z_{d+1}>0$ then
	\*[
	\left(\tilde{\nabla}_z\log\pi_S(z)\right)_{d+1}\le 0.
	\]
	If $k<1$ and $z_{d+1}>0$ then
	\*[
	\frac{\left(\tilde{\nabla}_z\log\pi_S(z)\right)_{d+1}}{\|\tilde{\nabla}_z\log\pi_S(z)\|}=\sqrt{1-z_{d+1}^2}
	\]
	and
	\*[
	\|\tilde{\nabla}_z\log\pi_S(z)\|=\sqrt{1-z_{1+d}^2}\frac{-dz_{d+1}(k-1)}{(1-z_{d+1})[(k-1)(1-z_{d+1})+2]}
	\]
	Denoting $k=1-\epsilon$ and multiplying the two equations above yields
	\*[
	\left(\tilde{\nabla}_z\log\pi_S(z)\right)_{d+1}=\frac{d\epsilon z_{d+1} (1-z_{d+1}^2)}{(1-z_{d+1})[2-\epsilon(1-z_{d+1})]}=\frac{d\epsilon z_{d+1} (1+z_{d+1})}{2-\epsilon(1-z_{d+1})}
	\]
	Letting $z_{d+1}\to 1$ yields
	\*[
	d\epsilon < d \frac{1}{2d} = \frac{1}{2}.
	\]
	Therefore, the uniform ergodicity holds by \cref{thm_PDMP}.
	\end{proof}

	
	\subsection{On the conjecture in \cref{conjecture}}\label{proof_conjecture}
	Our arguments include two parts. We first study a one-dimensional target distribution with density $\pi(t)\propto |t|^{-\alpha}$, $t\in [-1,1]$ and $\alpha>0$. We show that, for this target, the expected hitting time to the boundary (that is $1$ or $-1$) of the BPS starting from $t=0$ is finite (i.e., the BPS ``escapes'' the origin) if and only if $\alpha<1$. Next, we show that, for target with the multivariate student's $t$ distribution with DoF $\nu$, we can approximate the density on the greatest circle passing the North Pole by a form of $\pi(t)\propto |t|^{-\alpha}$ if the chain is close to the North Pole. Then $\alpha<1$ corresponds to the condition in our conjecture
	in \cref{conjecture}, which is $\nu>d-1$ in this case.

	\subsubsection{One-dimensional target} 
	
	We consider the following one-dimensional target with density
	\*[
    f_0(t)\propto |t|^{-\alpha},\quad \forall |t|<1,
	\]
	where $\alpha>0$. Note that the density goes to infinity as $t\to 0$. To study the behavior of BPS, we consider the following ``truncated'' targets:
	\*[
	f_{\epsilon}(t)\propto (\epsilon+|t|)^{-\alpha}, \quad \forall t\in [-(1-\epsilon), 1-\epsilon]
	\]
	and we will later let $\epsilon\to 0$ to recover $f_0(t)$.
	
	Next, we consider $T_{\epsilon}$ to be the first bouncing time when the starting state is from $0$ with unit velocity for target $f_{\epsilon}$. Then we have
	\*[
	\Pr(T_{\epsilon}< 1-\epsilon)=1-\exp\left(-\int_0^{1-\epsilon}(\log f_{\epsilon})'\dee t\right)
	\]
	Note that $f_{\epsilon}(1-\epsilon)\propto 1$ and $f_{\epsilon}(0)\propto \epsilon^{-\alpha}$, both are up to the same constant. Therefore, we have
	\*[
	\int_0^{1-\epsilon}(\log f_{\epsilon})'\dee t=\log(1)-\log(\epsilon^{-\alpha})
	\]
	which implies
	\*[
		\Pr(T_{\epsilon}< 1-\epsilon)=1-\epsilon^{-\alpha}.
	\]
	Similarly, we can get the density of $T_{\epsilon}$ as
	\*[
	f_{T_{\epsilon}}(t)= -\epsilon^{\alpha}\frac{\dee}{\dee t}  (\epsilon+|t|)^{-\alpha}=\alpha \epsilon^{\alpha}(\epsilon+|t|)^{-\alpha-1}=\frac{\alpha}{\epsilon}\left(\frac{|t|}{\epsilon}+1\right)^{-1-\alpha},\quad \forall 0<t<1-\epsilon.
	\]
	Then, we define $\tilde{T}_{\epsilon}$ as the conditional bouncing time which is $T_{\epsilon}$ conditioned on $T_{\epsilon}<1-\epsilon$. Then we have the density of $\tilde{T}_{\epsilon}$
	\*[
	f_{\tilde{T}_{\epsilon}}(t)=\frac{\alpha}{\epsilon(1-\epsilon^{\alpha})}\left(\frac{t}{\epsilon}+1\right)^{-1-\alpha},\quad t\in [0,1-\epsilon].
	\]
	Now we consider the ``hitting time'' to the boundary (that is, to $\pm(1-\epsilon)$). We can denote the ``hitting time'' as
	\*[
	S^{\epsilon}=\sum_{i=1}^{G^{\epsilon}} (2\tilde{T}_{\epsilon,i}) + (1-\epsilon),
	\]
	where $\tilde{T}_{\epsilon,i}$ are i.i.d.\ from density $f_{\tilde{T}_{\epsilon}}$ and $G^{\epsilon}$ is an independent geometric random variable with parameter $\epsilon^{\alpha}$. Note that
	\*[
	\EE[G^{\epsilon}]=\epsilon^{-\alpha}.
	\]
	Therefore, to study $\lim_{\epsilon\to 0}\EE[S^{\epsilon}]<\infty$, it suffices to study
	\*[
	\lim_{\epsilon\to 0}\EE\left[\sum_{i=1}^{G^{\epsilon}} \tilde{T}_{\epsilon,i}\right]=\lim_{\epsilon\to 0}\epsilon^{-\alpha}\EE[\tilde{T}_{\epsilon}].
	\]
	Using the density of $\tilde{T}_{\epsilon}$, we have
\*[
	\EE[\tilde{T}_{\epsilon}]&=\frac{\alpha}{\epsilon(1-\epsilon^{\alpha})}\int_0^{1-\epsilon} t \left(\frac{t}{\epsilon}+1\right)^{-1-\alpha}\dee t\\
	&=\frac{\alpha}{\epsilon(1-\epsilon^{\alpha})}\int_0^{1-\epsilon} \left[\epsilon\left(\frac{t}{\epsilon}+1\right)-\epsilon\right] \left(\frac{t}{\epsilon}+1\right)^{-1-\alpha}\dee t\\
	&=\frac{\alpha}{\epsilon(1-\epsilon^{\alpha})}\left[\int_0^{1-\epsilon} \epsilon \left(\frac{t}{\epsilon}+1\right)^{-\alpha}\dee t-\int_0^{1-\epsilon} \epsilon \left(\frac{t}{\epsilon}+1\right)^{-\alpha-1}\dee t\right]\\
	&=\frac{\alpha}{\epsilon(1-\epsilon^{\alpha})}\int_0^{1-\epsilon} \epsilon \left(\frac{t}{\epsilon}+1\right)^{-\alpha}\dee t-\epsilon\\
	&=\frac{\alpha\epsilon}{1-\epsilon^{\alpha}}\int_0^{\frac{1-\epsilon}{\epsilon}}  \left(u+1\right)^{-\alpha}\dee u-\epsilon\\
	&=\begin{cases}
	\frac{\epsilon}{1-\epsilon}\log(\frac{1-\epsilon}{\epsilon}+1)-\epsilon,\quad &\textrm{if } \alpha=1,\\
	\frac{\alpha\epsilon}{1-\epsilon^{\alpha}}\frac{\left(1+\frac{1-\epsilon}{\epsilon}\right)^{1-\alpha}-1}{1-\alpha}-\epsilon,\quad &\textrm{otherwise.}\\
	\end{cases}\\
	&=\begin{cases}
	\frac{\epsilon}{1-\epsilon}\log(\frac{1}{\epsilon})-\epsilon,\quad &\textrm{if } \alpha=1,\\
	\frac{\alpha}{1-\alpha}\frac{\epsilon^{\alpha}-\epsilon}{1-\epsilon^{\alpha}}-\epsilon,\quad &\textrm{otherwise.}
	\end{cases}
\]
	Therefore, letting $\epsilon\to 0$, if $\alpha<1$, then
	\*[
	\epsilon^{-\alpha}\EE[\tilde{T}_{\epsilon}]\to \epsilon^{-\alpha} \frac{\alpha}{1-\alpha} \epsilon^{\alpha}= \frac{\alpha}{1-\alpha}<\infty.
	\]
	On the other hand, if $\alpha=1$, then
	\*[
		\epsilon^{-\alpha}\EE[\tilde{T}_{\epsilon}]=\epsilon^{-1} \left[	\frac{\epsilon}{1-\epsilon}\log(\frac{1}{\epsilon})-\epsilon\right]= \frac{1}{1-\epsilon}\log(\frac{1}{\epsilon})-1 \to \infty.
	\]
	Or if $\alpha>1$, then
	\*[
	\epsilon^{-\alpha}\EE[\tilde{T}_{\epsilon}]=\epsilon^{-\alpha}\left[\frac{\alpha}{\alpha-1}\frac{\epsilon-\epsilon^{\alpha}}{1-\epsilon^{\alpha}}-\epsilon\right]\to \frac{1}{\alpha-1}\epsilon^{1-\alpha}\to \infty.
	\]
	Therefore, $\alpha<1$ is required for the expected ``escaping time'' of the BPS to be finite.
	
	\subsubsection{Approximation round the North Pole}
	
	We pick the greatest circle such that $z_2=z_3=\dots=z_d=0$ and $z_1^2+z_{d+1}^2=1$. Suppose the target distribution
	\*[
	\pi_{\nu}(x)\propto \left(1+\frac{\|x\|^2}{\nu}\right)^{-\frac{\nu+d}{2}},\quad x\in \Reals^d.
	\]
	
	Then, for any $z\in \mathbb{S}^d$, we have
	\*[
	\log \pi_S(z)=\log\pi_{\nu}(x)+ d\log(d+\|x\|^2)=C-d g_{\nu/d}(z_{d+1}),
	\]
	where $C$ is a constant and 
	\*[
	g_k(z_{d+1}):=\frac{k+1}{2}\log\left(k-1+\frac{2}{1-z_{d+1}}\right)+\log(1-z_{d+1}).
	\]
	Therefore, let $k:=\nu/d$, we have
	\*[
	\pi_S(z)\propto \exp(-dg_k(z_{d+1}))=\left[\left(k-1+\frac{2}{1-z_{d+1}}\right)^{\frac{k+1}{2}}\right]^{-d}(1-z_{d+1})^{-d}.
	\]
	Now consider the picked greatest circle, for any $z$ such that $z_1^2+z_{d+1}^2=1$, we have
\*[
	\pi_S(z_1,z_{d+1}\mid z_{2:d}=0)&\propto \left[\left(k-1+\frac{2}{1-z_{d+1}}\right)^{\frac{k+1}{2}}\right]^{-d}(1-z_{d+1})^{-d}\\
	&=\left[\left(k-1+\frac{2}{1-\sqrt{1-z_1^2}}\right)^{\frac{k+1}{2}}\right]^{-d}\left(1-\sqrt{1-z_1^2}\right)^{-d}.
\]
	If $k$ is fixed, letting $z_1\to 0$ and using the approximation $\sqrt{1-z_1^2}\to 1-z_1^2/2$, we have
	\*[
	 \left[\left(k-1+\frac{4}{z_1^2}\right)^{\frac{k+1}{2}}\right]^d \left(\frac{z_1^2}{2}\right)^{-d}\to  \left[\left(\frac{4}{z_1^2}\right)^{\frac{k+1}{2}}\right]^{-d} \left(\frac{z_1^2}{2}\right)^{-d}\propto |z_1|^{(k+1)d-2d}
	\]
	Therefore, if we write it as $|z_1|^{-\alpha}$ then we have
	\*[
	\alpha=(1-k)d=(1-\nu/d)d=d-\nu.
	\]
	Therefore, $\alpha<1$ corresponds to $\nu>d-1$.

\subsection{Proof of \cref{lemma-lati}}\label{proof-lemma-lati}
\begin{proof}
All the randomness for generating the proposal comes from the standard multivariate Gaussian random variable with distribution $\mathcal{N}(0,I_d)$ in the tangent space. We can denote it using the proposal variance $h^2$ and $d$ independent standard Gaussian random variables $V_1,\dots,V_d\sim \mathcal{N}(0,1)$. Then we can use $\sqrt{\sum_{j=1}^dV_j^2}$ to denote the Euclidean norm of the multivariate Gaussian random variable.

Suppose the current location of the chain is at $z$, and we are interested on the $i$-th coordinate of the proposal $\hat{z}$. If $z_i\not\in\{-1,1\}$, in the ``tangent space'' at $z$, there must be $d-1$ directions orthogonal to the direction of the $i$-th coordinate of the sphere. This leaves only one direction in the ``tangent space'' at $z$ which is not orthogonal to the direction of the $i$-th coordinate of the sphere. We denote the marginal of the multivariate Gaussian to this direction as $U_i$. Clearly, $U_i\sim\mathcal{N}(0,1)$.

Then, using some facts from basic geometry, the random walk in the tangent space without projection back to sphere changes the ``latitude'' from $z_{d+1}$ to $z_{d+1}+\sqrt{1-z_{d+1}^2}hU_{d+1}$ with distance to the origin equals to $\sqrt{1+h^2\sum_{j=1}^d V_j^2}$. Therefore, basic geometry tells us after projection back to the sphere, the ``latitude'' of the proposal satisfies
\*[
\frac{\hat{z}_{d+1}}{1}=\frac{z_{d+1}+\sqrt{1-z_{d+1}^2}hU_{d+1}}{\sqrt{1+h^2\sum_{j=1}^d V_j^2}}.
\]
Similarly for other coordinates,
we can get closed forms for any $\hat{z}_i$ in terms of $\sum_{j=1}^d V_j^2$ and $U_i$ for $i=1,\dots,d+1$. Writing $\sum_{j=1}^dV_j^2=U_i^2+U_{i,\perp}^2$, we have
\*[
\hat{z}_i&=
\frac{1}{\sqrt{1+h^2\sum_{j=1}^d V_j^2}}\left(z_i+\sqrt{1-z_i^2}hU_i\right)
=\frac{1}{\sqrt{1+h^2(U_i^2+U_{i,\perp}^2)}}\left(z_i+\sqrt{1-z_i^2}hU_i\right).
\]
Furthermore, we have
\*[
\hat{z}_i&=
\frac{1}{\sqrt{1+h^2\sum_{j=1}^d V_j^2}}\left(z_i+\sqrt{1-z_i^2}hU_i\right)\\
&=\frac{\sqrt{1+h^2U_i^2}}{\sqrt{1+h^2\sum_{j=1}^d V_j^2}}\left(\frac{z_i}{\sqrt{1+h^2U_i^2}}+\frac{\sqrt{1-z_i^2}hU_i}{\sqrt{1+h^2U_i^2}}\right)
\]
 Intuitively, the term
$\frac{z_i}{\sqrt{1+h^2U_i^2}}+\frac{\sqrt{1-z_i^2}hU_i}{\sqrt{1+h^2U_i^2}}$ 
comes from the change of $i$-th coordinate caused by $U_i$. The other term
$\frac{\sqrt{1+h^2U_i^2}}{\sqrt{1+h^2\sum_{j=1}^d V_j^2}}$
comes from the change of the $i$-th coordinate by the other $d-1$ orthogonal directions through the ``curvature'' of the sphere when projecting back to the sphere.

Note that $\{U_i,i=1,\dots,d+1\}$ are dependent, all of which can be written as linear combinations of $\{V_i\}$. If $h=\bigO(d^{-1/2})$, observing that $\sum_{j=1}^d V_j^2-U_i^2$ is chi-squared distributed with degree of freedom $d-1$, we can write
\*[
	\hat{z}_{i}&=\frac{\sqrt{1+h^2U_{i}^2}}{\sqrt{1+h^2\left(\sum_{j=1}^d V_j^2-U_{i}^2\right)+h^2U_{i}^2}}\left(\frac{z_{i}}{\sqrt{1+h^2U_{i}^2}}+\frac{\sqrt{1-z_{i}^2}hU_{i}}{\sqrt{1+h^2U_{i}^2}}\right)\\
	&=\frac{\sqrt{1+h^2U_{i}^2}(1+\bigO_{\Pr}(h^4d))}{\sqrt{\left(1+h^2\left(\sum_{j=1}^d V_j^2-U_{i}^2\right)\right)\left(1+h^2U_{i}^2\right)}}\left(\frac{z_{i}}{\sqrt{1+h^2U_{i}^2}}+\frac{\sqrt{1-z_{i}^2}hU_{i}}{\sqrt{1+h^2U_{i}^2}}\right)\\
	&=\left[\frac{1}{\sqrt{1+h^2\left(\sum_{j=1}^d V_j^2-U_{i}^2\right)}}+\bigO_{\Pr}(h^4d)\right]\left[\left(1-\frac{1}{2} h^2 U_{i}^2\right)\left(z_{i}+\sqrt{1-z_{i}^2}h U_{i}\right)+\bigO_{\Pr}(h^3)\right]\\
	&=\frac{1}{\sqrt{1+h^2U_{i,\perp}^2}}\left[\left(1-\frac{1}{2} h^2 U_i^2\right)z_{i}-\sqrt{1-z_{i}^2}h U_i\right]+\bigO_{\Pr}(h^3+h^4dz_i),
\]
where $U_{i,\perp}^2\sim\chi^2_{d-1}$ which is independent with $U_i$.
Therefore, in the stationary phase, if $z_{d+1}=\bigO(d^{-1/2})$, then we have
\*[
	\hat{z}_{d+1}&=\frac{1}{\sqrt{1+h^2(d-1)}}\left[1+\bigO_{\Pr}(d^{-1/2}+h^2d^{1/2})\right]\left(z_{d+1}-hU_{d+1}\right)+o_{\Pr}(d^{-1})\\
	&=\left[\frac{1}{\sqrt{1+h^2(d-1)}}+\bigO_{\Pr}(d^{-1/2})\right]\left(z_{d+1}-hU_{d+1}\right)+o_{\Pr}(d^{-1})\\
	&=\frac{1}{\sqrt{1+h^2(d-1)}}\left(z_{d+1}-hU_{d+1}\right)+\bigO_{\Pr}(d^{-1}).
\]
Similarly, in the transient phase, if $z_{d+1}^2=1-o(h^2)$, we have
\*[
	\hat{z}_{d+1}&= \left[\frac{1}{\sqrt{1+h^2(d-1)}}+\bigO_{\Pr}(d^{-1/2})\right]\left(1-\frac{1}{2}h^2U^2\right)z_{d+1}+\bigO_{\Pr}(h^4d)\\
&= \frac{1}{\sqrt{1+h^2(d-1)}}\left(1-\frac{1}{2}h^2U_{d+1}^2\right)z_{d+1}+\bigO_{\Pr}(d^{-1/2}).
\]
\end{proof}
\subsection{Proof of \cref{thm-robust}}\label{proof-thm-robust}

\begin{proof}
	Throughout the proof, we assume $h=o(d^{-1/2})$ so that $dh^2\to 0$.
	{For bounded random variables, convergence in probability implies convergence in $L_1$. Therefore, it suffices to show $\left(1\wedge\frac{\pi_{\mu,\Sigma}(\hat{X})(R^2+\|\hat{X}\|^2)^d}{\pi_{\mu,\Sigma}(X)(R^2+\|X\|^2)^d}\right)\overset{\Pr}{\to} 1$. Furthermore,
     for any $\epsilon>0$
	\*[
	\Pr\left(\left|1\wedge\frac{\pi_{\mu,\Sigma}(\hat{X})(R^2+\|\hat{X}\|^2)^d}{\pi_{\mu,\Sigma}(X)(R^2+\|X\|^2)^d}-1\right|>\epsilon\right)\le \Pr\left(\left|\frac{\pi_{\mu,\Sigma}(\hat{X})(R^2+\|\hat{X}\|^2)^d}{\pi_{\mu,\Sigma}(X)(R^2+\|X\|^2)^d}-1\right|>\epsilon\right)
	\]
	Therefore,
	it suffices to show $\frac{\pi_{\mu,\Sigma}(\hat{X})(R^2+\|\hat{X}\|^2)^d}{\pi_{\mu,\Sigma}(X)(R^2+\|X\|^2)^d}\overset{\Pr}{\to} 1$, or equivalently 
	\*[
	\log\frac{\pi_{\mu,\Sigma}(\hat{X})(R^2+\|\hat{X}\|^2)^d}{\pi_{\mu,\Sigma}(X)(R^2+\|X\|^2)^d}\overset{\Pr}{\to} 0.
	\]
	}
	Defining
	$\delta_i:=\lambda_i^{-1}-1$,
	we can write $\pi_{\mu,\Sigma}(x)$ as
	\*[
	\log\pi_{\mu,\Sigma}(x)=\log\pi_{0,I_d}(x)-\frac{1}{2}\sum_i\delta_ix_i^2+\sum_i\mu_i(1+\delta_i)x_i -\frac{1}{2}\sum_i \log(\lambda_i),
	\]
	where $\pi_{0,I_d}$ is just the standard multivariate Gaussian density. Furthermore, using $R=d^{1/2}$ and
	\*[
	\log\frac{(d+\|\hat{X}\|^2)^d}{(d+\|X\|^2)^d}=d\log\frac{1-z_{d+1}}{1-\hat{z}_{d+1}},
	\]
	we have
	\*[
	      	&\log\frac{\pi_{\mu,\Sigma}(\hat{X})(d+\|\hat{X}\|^2)^d}{\pi_{\mu,\Sigma}(X)(d+\|X\|^2)^d}=\log\frac{\pi_{\mu,\Sigma}(\hat{X})}{\pi_{\mu,\Sigma}(X)}+d\log\frac{1-z_{d+1}}{1-\hat{z}_{d+1}}\\
	&=\log\frac{\pi_{0,I_d}(\hat{X})}{\pi_{0,I_d}(X)}+d\log\frac{1-z_{d+1}}{1-\hat{z}_{d+1}}-\frac{1}{2}\sum_i\delta_i(\hat{X}_i^2-X_i^2)+\sum_i\mu_i(1+\delta_i)(\hat{X}_i-X_i).
	  \]
	  
	Therefore, it suffices to show that if $h$ satisfies the condition of \cref{thm-robust} then we have all of the three results:
	\[\label{temp987}
	\log\frac{\pi_{0,I_d}(\hat{X})}{\pi_{0,I_d}(X)}+d\log\frac{1-z_{d+1}}{1-\hat{z}_{d+1}}\overset{\Pr}{\to} 0,
	\]
	\[\label{temp3}
        \sum_i\delta_i(\hat{X}_i^2-X_i^2)\overset{\Pr}{\to} 0,
	\]
	\[\label{temp4}
	\sum_i\mu_i(1+\delta_i)(\hat{X}_i-X_i)\overset{\Pr}{\to} 0.
	\]
    The rest of the proof is just to show \cref{temp987,temp3,temp4} using Taylor expansions and \cref{lemma-lati}.

	\subsubsection{Proof of \cref{temp987}}
    Writing $X=\mu+\Sigma^{1/2}\tilde{N}$ where $\tilde{N}$ is a standard multivariate Gaussian vector, because of the conditions in \cref{thm-robust} that there exist constants $C<\infty$ and $c>0$ such that
		$c\le \lambda_i\le C$ and $|\mu_i|\le C$, we have $\mu^T\Sigma^{1/2}\tilde{N}=\bigO_{\Pr}(d^{1/2})$. Then using the fact $\tilde{N}^T\Sigma \tilde{N}=\sum_i\lambda_i + \bigO_{\Pr}(d^{1/2})$, we have
\*[
	\|X\|^2&=\|X-\mu\|^2+\sum_i\mu_i^2+ 2\mu^T\Sigma^{1/2}\tilde{N}\\
	&=\tilde{N}^T\Sigma \tilde{N}+\sum_i \mu_i^2 +2\mu^T\Sigma^{1/2}\tilde{N}\\
	&=\left(d-\sum_i(1-\lambda_i)+\bigO_{\Pr}(d^{1/2})\right)+\sum_i\mu_i^2 +\bigO_{\Pr}(d^{1/2}).
	    \]
	Then we have 
	$\|X\|^2=d+\sum_i\mu_i^2-\sum_i(1-\lambda_i)+\bigO_{\Pr}(d^{1/2})$ and
	\[\label{eq_temp33}
	z_{d+1}=\frac{\|X\|^2-d}{\|X\|^2+d}=\frac{\sum_i\mu_i^2-\sum_i(1-\lambda_i)}{2d+\sum_i\mu_i^2-\sum_i(1-\lambda_i)}+\bigO_{\Pr}(d^{-1/2}).
	\]
	By \cref{eq_robust_cond}, $\left|\sum_i\mu_i^2-\sum_i(1-\lambda_i)\right|=\bigO(d^{\alpha})$, we have $\|X\|^2=d+\bigO(d^{\alpha})+\bigO_{\Pr}(d^{1/2})$ and $z_{d+1}=\bigO(d^{\alpha-1})+\bigO_{\Pr}(d^{-1/2})$, where the $\bigO(d^{\alpha-1})$ term represents $\frac{\sum_i\mu_i^2-\sum_i(1-\lambda_i)}{2d+\sum_i\mu_i^2-\sum_i(1-\lambda_i)}$ which is bounded away from $1$.
	From \cref{eq_approx_large_h} we have
	\[\label{eq_approx_zz}
	\hat{z}_i=(1+\bigO_{\Pr}(h^2d))(z_i+\bigO_{\Pr}(h))=z_i+\bigO_{\Pr}(h)+\bigO_{\Pr}(h^2d)z_i.
	\]
	Using $h=o(d^{-1/2})$, $h=o(d^{-\alpha})$, and $z_{d+1}=\bigO(d^{\alpha-1})+\bigO_{\Pr}(d^{-1/2})$ we have $h^2dz_{d+1}=o_{\Pr}(d^{-1})$ then from \cref{eq_approx_zz} we have
	\[\label{eq_approx_zz_dplus1}
	\hat{z}_{d+1}=z_{d+1}+\bigO_{\Pr}(h)+o_{\Pr}(d^{-1}).
	\]
	Then we have
	\*[
	&\log\frac{\pi_{0,I_d}(\hat{X})}{\pi_{0,I_d}(X)}+d\log\frac{1-z_{d+1}}{1-\hat{z}_{d+1}}\\
	&=d\left[\left(\frac{1}{1-z_{d+1}}-\frac{1}{1-\hat{z}_{d+1}}\right)+\log\frac{1-z_{d+1}}{1-\hat{z}_{d+1}}\right]\\
	&=d\left[\frac{z_{d+1}-\hat{z}_{d+1}}{(1-z_{d+1})(1-\hat{z}_{d+1})}-\frac{z_{d+1}-\hat{z}_{d+1}}{1-\hat{z}_{d+1}}\right]+\bigO_{\Pr}(dh^2)+\bigO_{\Pr}(d^3h^4)z_{d+1}^2\\
	&=d\frac{(z_{d+1}-\hat{z}_{d+1})z_{d+1}}{(1-z_{d+1})(1-\hat{z}_{d+1})}+\bigO_{\Pr}(dh^2)+\bigO_{\Pr}(d^3h^4)z_{d+1}^2\\
	&{=d \frac{\left[\bigO_{\Pr}(h)+\bigO_{\Pr}(h^2d)z_{d+1}\right][\bigO_{\Pr}(d^{-1/2})+\bigO(d^{\alpha-1})]}{1-\bigO_{\Pr}(d^{-1/2})-\bigO(d^{\alpha-1})-\bigO_{\Pr}(h)}+\bigO_{\Pr}(dh^2)+\bigO_{\Pr}(d^3h^4)z_{d+1}^2}=o_{\Pr}(1),
	    \]
	where the last equality holds because (i) $dh^2=o(1)$ as $h=o(d^{-1/2})$; (ii) since $h=o(d^{-\alpha})$, we can get $d^3h^4z_{d+1}^2$ equals to $o_{\Pr}(d^{2\alpha-1})$ if $\alpha\le 1/2$ and equals to $o_{\Pr}(d^{1-2\alpha})$ if $1/2<\alpha\le 1$. In both cases, $\bigO_{\Pr}(d^4h^4)z_{d+1}^2=o_{\Pr}(1)$; (iii)  similarly, we can verify in both cases $\alpha\le 1/2$ and $1/2<\alpha\le 1$ we have $d(h+h^2dz_{d+1})(d^{-1/2}+d^{\alpha-1})=o_{\Pr}(1)$; (iv) the term $\bigO(d^{\alpha-1})$ is $\frac{\sum_i\mu_i^2-\sum_i(1-\lambda_i)}{2d+\sum_i\mu_i^2-\sum_i(1-\lambda_i)}$ which is bounded away from $1$. 

	\subsubsection{Proof of \cref{temp3}}
First of all, we have
	\*[
\hat{X}_i=\frac{d^{1/2}}{1-\hat{z}_{d+1}}\hat{z}_i,\quad  X_i=\frac{d^{1/2}}{1-z_{d+1}}z_i.
\]
Then we have
\*[
	\sum_i\delta_i\left(\hat{X}_i^2-X_i^2\right)
	&=d\sum_i\delta_i \frac{\hat{z}_i^2(1-z_{d+1})^2-z_i^2(1-\hat{z}_{d+1})^2}{(1-z_{d+1})^2(1-\hat{z}_{d+1})^2} \\
	&=d\sum_i\delta_i \left[\frac{\hat{z}_i^2-z_i^2}{(1-\hat{z}_{d+1})^2}+z_i^2\frac{2(\hat{z}_{d+1}-z_{d+1})-(\hat{z}_{d+1}^2-z_{d+1}^2)}{(1-z_{d+1})^2(1-\hat{z}_{d+1})^2}\right].
	    \]
		Using Taylor expansion, we approximate $z_i$ and $\hat{z}_i$ using \cref{eq_approx_whole} in \cref{lemma-lati}. First, we can replace $\hat{z}_i$ by $z_i$ :
	\[\label{temp5}
	\hat{z}_i=a_d^{(i)}z_i+a_d^{(i)}\sqrt{1-z_i^2}hU_i-\frac{1}{2}a_d^{(i)}z_ih^2U_i^2+\bigO_{\Pr}(h^3+dh^4z_i),
	\]
	where $a_d^{(i)}=\frac{1}{\sqrt{1+h^2U_{i,\perp}^2}}=\frac{1}{\sqrt{1+h^2(d-1)}}(1+\bigO_{\Pr}(h^2d^{1/2})+ \bigO_{\Pr}(d^{-1/2}))$ and $U_i\sim \mathcal{N}(0,1)$. Then for given $z$, we can get 
	\*[
	\hat{z}_i^2-z_i^2&=((a_d^{(i)})^2-1)z_i^2 +2(a_d^{(i)})^2z_i\sqrt{1-z_i^2}hU_i\\
	&\qquad+(a_d^{(i)})^2(1-2z_i^2)h^2U_i^2+\bigO_{\Pr}(z_ih^3+z_i^2h^4+h^6+d^2h^8z_i^2).
	\]
	Then, we can replace $z_i$ using $z_{d+1}$ and $X_i$ since
	\[\label{temp51}
	z_i=(1-z_{d+1})d^{-1/2}X_i.
	\]
	By \cref{eq_temp33}, for $z_{d+1}$ we have
	\[\label{temp52}
	z_{d+1}=\frac{\sum_i\mu_i^2-\sum_i(1-\lambda_i)}{2d+\sum_i\mu_i^2-\sum_i(1-\lambda_i)}+\bigO_{\Pr}(d^{-1/2}).
	\]
    Then, by \cref{eq_approx_zz_dplus1} we have
    \begin{equation}\label{temp556}
        \begin{split}
            \sum_i\delta_i\left(\hat{X}_i^2-X_i^2\right)
	&=\bigO_{\Pr}(dh)\left(\frac{1}{d}\sum_i\delta_iX_i^2\right)+d\sum_i\delta_i\frac{\hat{z}_i^2-z_i^2}{(1-\hat{z}_{d+1})^2}.
        \end{split}
    \end{equation}
    We first focus on the term $d\sum_i\delta_i(\hat{z}_i^2-z_i^2)$. For given $X$, using \cref{temp5}, replacing $z_i$ using \cref{temp51}, simplifying $z_{d+1}$ with \cref{temp52}, we have
    \*[
    d\sum_i\delta_i(\hat{z}_i^2-z_i^2)&=d\sum_i\delta_i[((a_d^{(i)})^2-1)X_i^2/d +d^{-1/2}hX_iU_i\\
    &\qquad+(a_d^{(i)})^2(1-2z_i^2)h^2U_i^2+\bigO_{\Pr}(z_ih^3+h^4+d^2h^8)]
    \]
    Recall that $U_i\sim\mathcal{N}(0,1)$ but for $i\neq j$, $U_i$ and $U_j$ are in general not independent since they are obtained by marginalizing the same multivariate Gaussian to two non-orthogonal directions. However, by basic geometry on the angle between the two marginalization directions corresponding to $U_i$ and $U_j$, respectively, we know $\EE[U_iU_j\mid X]=\bigO(z_iz_j)=\bigO(d^{-1}X_iX_j)+\bigO(d^{-3/2})$. That is, although $U_i$ and $U_j$ are independent if and only if either $z_i=0$ or $z_j=0$, their correlation is quite small. Note that we only consider $h=o(d^{-1/2})$ then  $a_d^{(i)}=a_d(1+\bigO_{\Pr}(h^2d^{1/2}))=a_d(1+o_{\Pr}(d^{-1/2}))$ where $a_d:=\frac{1}{\sqrt{1+h^2(d-1)}}$, we keep the leading terms of $d\sum_i\delta_i(\hat{z}_i^2-z_i^2)$:
    \*[
     &d\sum_i\delta_i(\hat{z}_i^2-z_i^2)=d\sum_i\delta_i(a_d^2-1)X_i^2/d+ d\sum_i\delta_i (d^{-1/2}h)X_iU_i +\bigO(dh^2)\sum_i\delta_i U_i^2 + o_{\Pr}(1)\\
     &=d(a_d^2-1)\left(\frac{1}{d}\sum_i\delta_iX_i^2\right)+\bigO(dh)\left(\frac{1}{d^{1/2}}\sum_i \delta_iX_iU_i\right)+\bigO(d^2h^2)\left(\frac{1}{d}\sum_i\delta_iU_i^2\right)+o_{\Pr}(1)\\
     &=d(a_d^2-1)\left(\frac{1}{d}\sum_i\delta_iX_i^2\right)+\bigO(dh)\bigO_{\Pr}\left(\sqrt{\frac{2}{d^2}\sum_{i\neq j} \delta_i\delta_j X_i^2X_j^2+\frac{1}{d}\sum_i \delta_i^2X_i^2}\right)\\
     &\qquad\qquad+\bigO(d^2h^2)\bigO_{\Pr}\left(\frac{1}{d}\sum_i\delta_i\right)+o_{\Pr}(1)\\
     &=d(a_d^2-1)\left(\frac{1}{d}\sum_i\delta_iX_i^2\right)+\bigO(dh)\left[\bigO_{\Pr}\left(\sqrt{\frac{1}{d}\sum_i \delta_i^2X_i^2}\right)+\bigO_{\Pr}\left(\frac{1}{d}\sum_i\delta_iX_i^2\right)\right]\\
     &\qquad\qquad+\bigO(d^2h^2)\bigO_{\Pr}\left(\frac{1}{d}\sum_i\delta_i\right)+o_{\Pr}(1).
    \]
    Substituting to \cref{temp556}, using $1-a_d^2=\bigO(h^2d)$ and $h=o(d^{-1/2})$, we have
    \[\label{temp557}
    	\sum_i\delta_i\left(\hat{X}_i^2-X_i^2\right)=\bigO_{\Pr}(d^2h^2+dh)\left(\frac{1}{d}\sum_i\delta_iX_i^2+\frac{1}{d}\sum_i\delta_i\right)+\bigO_{\Pr}(dh)\sqrt{\frac{1}{d}\sum_i\delta_i^2X_i^2}+o_{\Pr}(1).
    \]
    Now note that both $\frac{1}{d}\sum_i\delta_iX_i^2$ and $\frac{1}{d}\sum_i\delta_i^2X_i^2$ concentrate to their means. Using $X_i\sim\mathcal{N}(\mu_i,\lambda_i)$, we have
    \*[
    \frac{1}{d}\sum_i\delta_iX_i^2=\frac{1}{d}\sum_i\delta_i (\mu_i^2+\lambda_i) +\bigO_{\Pr}(d^{-1/2}),\quad \frac{1}{d}\sum_i\delta_i^2X_i^2=\frac{1}{d}\sum_i\delta_i^2(\mu_i^2+\lambda_i) +\bigO_{\Pr}(d^{-1/2})
    \]
    Therefore, substituting to \cref{temp557}, in order to guarantee $\sum_i\delta_i(\hat{X}_i^2-X_i^2)\overset{\Pr}{\to} 0$, a sufficient condition for $h$ in terms of $\{\mu_i\}$ and $\{\lambda_i\}$ is
	\*[
	&(d^2h^2+dh)\cdot\left(\frac{1}{d}\sum_i\delta_i(\lambda_i+\mu_i^2)+\frac{1}{d}\sum_i\delta_i\right)+(dh)\cdot \sqrt{\frac{1}{d}\sum_i\delta_i^2(\lambda_i+\mu_i^2)}=o(1).
	\]
	Therefore, it is sufficient to assume
	\[\label{temp6}
	h=o\left(d^{-1}/\sqrt{\max\left\{\left|\frac{1}{d}\sum_i\frac{\delta_i}{1+\delta_i}\right|,\left|\frac{1}{d}\sum_i\delta_i\mu_i^2\right|,\left|\frac{1}{d}\sum_i\frac{\delta_i^2}{1+\delta_i}\right|,\left|\frac{1}{d}\sum_i\delta_i^2\mu_i^2\right|, \left|\frac{1}{d}\sum_i\delta_i\right|\right\}}\wedge d^{-1/2}\right).
	\]
	\subsubsection{Proof of \cref{temp4}}	
We use a similar approach as the previous part. Using Taylor expansion and \cref{temp5} to replace $\hat{z}_i$, we have
\*[
&\hat{X}_i-X_i=\frac{d^{1/2}}{1-\hat{z}_{d+1}}\hat{z}_i-\frac{d^{1/2}}{1-z_{d+1}}z_i\\
&=\frac{d^{1/2}}{1-z_{d+1}}\left(a_d^{(i)}z_i+a_d^{(i)}\sqrt{1-z_i^2}hU_i+\bigO_{\Pr}(h^2)\right)\left(1+\bigO(\hat{z}_{d+1}-z_{d+1})+\bigO_{\Pr}(h^2)\right)-\frac{d^{1/2}}{1-z_{d+1}}z_i.
\]
Substituting to \cref{temp4}, and using \cref{temp51} to replace $z_i$ by $X_i$, $a_d$ to replace $a_d^{(i)}$, keeping the leading terms, we can get
\[\label{temp558}
\sum_i\mu_i(1+\delta_i)(\hat{X}_i-X_i)=(a_d-1)\sum_i\left[\mu_i(1+\delta_i)X_i\right]+a_d h\sum_i \left[\sqrt{d-X_i^2}\mu_i(1+\delta_i)U_i\right] +o_{\Pr}(1).
\]
Now we first focus on the term $(a_d-1)\sum_i\left[\mu_i(1+\delta_i)X_i\right]$ in \cref{temp558}. Using the fact that $a_d-1=\bigO(h^2d)$ and $X_i\sim\mathcal{N}(\mu_i,\lambda_i)$ are independent random variables, we have
\*[
(a_d-1)\sum_i\left[\mu_i(1+\delta_i)X_i\right]=\bigO(h^2d^2)\left[\frac{1}{d}\sum_i\mu_i^2(1+\delta_i) +\frac{1}{d}\bigO_{\Pr}\left(\sqrt{\sum_i\mu_i^2(1+\delta_i)^2\lambda_i}\right)\right].
\]
Therefore, $(a_d-1)\sum_i\left[\mu_i(1+\delta_i)X_i\right]\overset{\Pr}{\to} 0$ if 
$(h^2d^2)\cdot\frac{1}{d}\sum_i\mu_i^2(1+\delta_i)=o(1)$.
Next, we focus on the term $a_d h\sum_i \left[\sqrt{d-X_i^2}\mu_i(1+\delta_i)U_i\right]$ in \cref{temp558}. Using $U_i\sim\mathcal{N}(0,1)$ conditional on $X$ and variance decomposition formula, we have
	\*[
&\var \left[a_d h\sum_i \left(\sqrt{d-X_i^2}\mu_i(1+\delta_i)U_i\right)\right]\\
&=\EE\left[a_d^2h^2\sum_i(d-X_i^2)\mu_i^2(1+\delta_i)^2\right]+\bigO(d^{-1})\left(\sum_ia_dh\sqrt{d}\mu_i(1+\delta_i)\right)^2\\
&=\bigO(d^2h^2)\left[\frac{1}{d}\sum_i\left(1-\frac{\mu_i^2+\lambda_i}{d}\right)\mu_i^2(1+\delta_i)^2+\left(\frac{1}{d}\sum_i \mu_i(1+\delta_i)\right)^2\right].
	\]
	Therefore, $a_d h\sum_i \left(\sqrt{d-X_i^2}\mu_i(1+\delta_i)U_i\right)\overset{\Pr}{\to} 0$ if both $(d^2h^2)\cdot \frac{1}{d}\sum_i\mu_i^2(1+\delta_i)^2 =o(1)$ and $(dh) \left|\frac{1}{d}\sum_i\mu_i(1+\delta_i)\right|=o(1)$.
    Overall, in order to make $\sum_i\mu_i(1+\delta_i)(\hat{X}_i-X_i)$ to converge to $0$ in probability, we need \cref{temp558} to be $o_{\Pr}(1)$. Then, a sufficient condition for $h$ in terms of $\{\mu_i\}$ and $\{\lambda_i\}$ is
	\*[
	d^2h^2\max\left\{\frac{1}{d}\sum_i\mu_i^2(1+\delta_i),\frac{1}{d}\sum_i\mu_i^2(1+\delta_i)^2\right\}\to 0,\quad dh\left|\frac{1}{d}\sum_i\mu_i(1+\delta_i)\right|\to 0.
	\]
	Therefore, it is sufficient if
	\[\label{temp7}
h=o\left(d^{-1}/\sqrt{\max\left\{\left|\frac{1}{d}\sum_i\frac{\delta_i}{1+\delta_i}\right|,\left|\frac{1}{d}\sum_i\delta_i\mu_i^2\right|,\left|\frac{1}{d}\sum_i\delta_i^2\mu_i^2\right|,\left|\frac{1}{d}\sum_i\mu_i^2\right|\right\}}\wedge d^{-1/2}\right).
\]
Overall, combing \cref{temp6} and \cref{temp7}, since we only need the order of $h$, \cref{temp6} can be relaxed to
\*[
h=o\left(d^{-1}/\sqrt{\frac{1}{d}\sum_i |1-\lambda_i|}\wedge d^{-1/2}\right)
\]
and \cref{temp7} can be relaxed to
\*[
h=o\left(d^{-1}/\sqrt{\max\left\{\frac{1}{d}\sum_i |1-\lambda_i|,\frac{1}{d}\sum_i\mu_i^2\right\}}\wedge d^{-1/2}\right).
\]
Combing the above two with $h=o(d^{-\alpha})$ completes the proof.
\end{proof}

\subsection{Proof of \cref{lemma_CLT}}\label{proof-lemma_CLT}
\begin{proof} 
Note that we assumed $h=\bigO(d^{-1})$.
We use $z$ and $\hat{z}$ to denote the corresponding coordinates of $X$ and $\hat{X}$ on the unit sphere. 
We first ``upgrade'' \cref{eq_approx_whole} in \cref{lemma-lati} to a convergence in $L_1$ statement. Define
\*[
\bar{z}_i:=\frac{1}{\sqrt{1+h^2U_{\perp}^2}}\left[\left(1-\frac{1}{2} h^2 U^2\right)z_{i}-\sqrt{1-z_{i}^2}h U\right].
\]  and the facts that $\hat{z}_i$ is bounded and $\bar{z}_i$ has a finite exponential moment, applying Cauchy--Schwarz inequality and Markov inequality, we have
\*[
\EE[|\hat{z}_i-\bar{z}_i|]&\le \EE[\left|\hat{z}_i-\bar{z}_i\Ind_{|\bar{z}_i|\le C\log(d)}\right|]+\EE[\left|\bar{z}_i\Ind_{|\bar{z}_i|> C\log(d)}\right|]\\
&=C\log(d)\bigO(h^3+dh^4z_i)+\sqrt{\EE[\bar{z}_i^2]}\sqrt{\Pr(\exp(|\bar{z}_i|)>d^{C})}\\
&=\bigO(\log(d)h^3)+\bigO(d^{-C/2})=\bigO(\log(d)h^3),
	    \]
where $C$ is chosen as a large enough constant so that $d^{-C/2}=o(\log(d)h^3)$.
Therefore, \cref{eq_approx_whole} in \cref{lemma-lati} can be rewritten as a convergence in $L_1$ statement:
\[\label{eq_approx_whole_L1}
	\EE_{\hat{z}_{i}\mid z_i}\left[\left|\hat{z}_{i}-\left( \frac{1}{\sqrt{1+h^2U_{\perp}^2}}\left[\left(1-\frac{1}{2} h^2 U^2\right)z_{i}-\sqrt{1-z_{i}^2}h U\right]\right)\right|\right]=\bigO(\log(d)h^3).
\]

{Next, we want to rule out some subsets of ``bad states'', including those states with the ``latitude'' $z_{d+1}$ too close to $1$.} That is, we define a sequence of ``typical sets'' $\{F_d\}$ with $\pi(F_d)\to 1$ such that when $x\in F_d$, the corresponding $z_{d+1}$ is bounded away from $1$.
	We define
\*[
	F_d:=&\left\{x\in \Reals^d: \left|\frac{1}{d}\sum_{i=1}^d\left[(\log f(x_i))'\right]^2-\EE_f\left[((\log f)')^2\right]\right|<d^{-1/2}\log(d) \right\}\\
& \cap \left\{x\in \Reals^d: \left|\frac{1}{d}\sum_{i=1}^d\left[(\log f(x_i))''\right]-\EE_f\left[((\log f)'')\right]\right|<d^{-1/2}\log(d)\right\} \\
&\cap \left\{x\in \Reals^d: \left|\frac{1}{d}\sum_{i=1}^d x_i(\log f(x_i))' -\EE_f\left[X(\log f)' \right]\right|<d^{-1/2}\log(d)\right\}\\
&\cap 
\left\{x\in \Reals^d: \left|\frac{1}{d}\sum_{i=1}^d x_i^2-\EE_f(X^2)\right|<d^{-1/2}\log(d) \right\}
	    \]
	Note that by the assumption $\lim_{x\to\pm\infty}xf'(x)=0$, we have
 \*[
 \EE_f\left[X(\log f)'\right]=\int \left[\frac{x}{f(x)}f'(x)\right]f(x)\dee x=\int xf'(x) \dee x= 0-\int f(x)\dee x=-1.
 \] 
Therefore our assumptions on $\pi$ imply 
\*[
\pi(F_d^c)=&\Pr_{\pi}(X\notin F_d)=\bigO\left(
\left(\frac{d^{1/2}}{\log(d)d^{1/2}}\right)^4+\left(\frac{d^{1/2}}{\log(d)d^{1/2}}\right)^3\right)=o(1).
\]
Using the definition of $\{F_d\}$, we can assume $X\in F_d$ and only consider the expectation over $\hat{X}$ given $X$. By the definition of $\{F_d\}$, we know $z_{d+1}$ is bounded away from $1$. 

{Next, we want to rule out the subset of ``bad proposals'' where the proposal ``latitude'' $\hat{z}_{d+1}$ is too close to $1$.}
Under stationarity, we know
	$
	R^2\frac{1+z_{d+1}}{1-z_{d+1}}=\lambda d\frac{1+z_{d+1}}{1-z_{d+1}}=d+\bigO_{\Pr}(d^{1/2})$,
which implies
	$z_{d+1}= \frac{1-\lambda}{1+\lambda}+\bigO_{\Pr}(d^{-1/2})$ and $\frac{1}{1-z_{d+1}}=\frac{1+\lambda}{2\lambda}+\bigO_{\Pr}(d^{-1/2})$. Furthermore, for given $X\in F_d$, since we know
	\*[
	\hat{z}_{d+1}=z_{d+1}+\bigO_{\Pr}(h),
	\]
	uniformly on $X\in F_d$, for all large enough $d$, we have 
	\*[
	z_{d+1}+h\log(d)=\frac{1-\lambda}{1+\lambda}+\bigO(\log(d)d^{-1/2})+h\log(d)<1-\epsilon,
	\] 
	where $\frac{\lambda}{1+\lambda}<\epsilon<1$ is a fixed constant. Furthermore, we have
\*[
\sup_{X\in F_d}\Pr(\hat{z}_{d+1}> 1-\epsilon)\le \sup_{X\in F_d}\Pr(\hat{z}_{d+1}> 	z_{d+1}+h\log(d))=o(d^{-1/2}).
\]
Therefore, we only need consider the cases such that $\hat{z}_{d+1}$ is no smaller than $1-\epsilon$ since
\*[
		&\sup_{X\in F_d}\EE_{\hat{X}\mid X}\left[\left|1\wedge \frac{\pi(\hat{X})(R^2+\|\hat{X}\|^2)^d}{\pi(X)(R^2+\|X\|^2)^d}-1\wedge \exp(W_{\hat{X}\mid X})\right|\right]\\
		&\le\sup_{X\in F_d}\EE_{\hat{X}\mid X}\left[\left|1\wedge\left( \frac{\pi(\hat{X})(R^2+\|\hat{X}\|^2)^d}{\pi(X)(R^2+\|X\|^2)^d}\Ind_{\{\hat{z}_{d+1}\le 1-\epsilon \}}\right)-1\wedge \exp(W_{\hat{X}\mid X})\right|\right]\\
		&\qquad+\sup_{X\in F_d}\Pr(\hat{z}_{d+1}>1-\epsilon)\\
		&\le\sup_{X\in F_d}\EE_{\hat{X}\mid X}\left[\left|\log \frac{\pi(\hat{X})(R^2+\|\hat{X}\|^2)^d}{\pi(X)(R^2+\|X\|^2)^d}\Ind_{\{\hat{z}_{d+1}\le 1-\epsilon\}}- W_{\hat{X}\mid X}\right| \right]+o(d^{-1/2}),
	    \]
	where the last step is because $1\wedge e^x$ is a $1$-Lipschitz function.
	
	In the rest of the proof, all we will use for approximating $\hat{z}_{d+1}$ is \cref{eq_approx_whole_L1}. Since  $\frac{\lambda}{1+\lambda}<\epsilon<1$, \cref{eq_approx_whole_L1} is also true if we replace $\hat{z}_{d+1}$ by $\hat{z}_{d+1}\Ind_{\{\hat{z}_{d+1}\le 1-\epsilon\}}$. That is, for truncated $\hat{z}_{d+1}$:
	\[
	\label{eq_approx_whole_L1_latitude}
	\sup_{X\in F_d}\EE_{\hat{z}_{d+1}\mid z_{d+1}}\left[\left|\hat{z}_{d+1}\Ind_{\{\hat{z}_{d+1}\le 1-\epsilon\}}-\bar{z}_{d+1}\right|\right]=\bigO(\log(d)h^3).
\]
	Therefore, in the rest of the proof, we simply work on the cases such that $\hat{z}_{d+1}\le 1-\epsilon$ and sometimes omit the term $\Ind_{\{\hat{z}_{d+1}\le 1-\epsilon)\}}$ for simplicity. The goal is then to consider $X\in F_d$ and establish a Gaussian approximation result for
\*[
\left(\sum_{i=1}^d [\log f(\hat{X}_i)-\log f(X_i)] +d [\log(1-z_{d+1})-\log(1-\hat{z}_{d+1})]\right)\Ind_{\{\hat{z}_{d+1}<1-\epsilon\}}.
\]
Note that by \cref{eq_approx_whole_L1}, we know $\EE[|\hat{z}_{i}-z_{i}|^3]=\bigO(h^3)$, $\EE[|\hat{z}_{i}-z_{i}|(\hat{z}_{d+1}-z_{d+1})^2]=\bigO(h^3)$, $\EE[|\hat{z}_{d+1}-z_{d+1}|(\hat{z}_{i}-z_{i})^2]=\bigO(h^3)$. Furthermore, we only need to consider $\frac{1}{1-\hat{z}_{d+1}}$ is bounded as $\hat{z}_{d+1}$ is bounded away from $1$. Then, using 
\begin{equation}\label{eq_approx_hat_X_i}
    \begin{split}
        &\hat{X}_i-X_i=\frac{R}{1-\hat{z}_{d+1}}\hat{z}_i-\frac{R}{1-z_{d+1}}z_i\\
&= \frac{R}{(1-z_{d+1})(1-\hat{z}_{d+1})}(\hat{z}_{d+1}-z_{d+1})z_i+\left(\frac{R(\hat{z}_{d+1}-z_{d+1})}{(1-z_{d+1})(1-\hat{z}_{d+1})}+\frac{R}{1-z_{d+1}}\right)(\hat{z}_i-z_i),
    \end{split}
\end{equation}
together with Taylor expansion and mean value theorem, uniformly on $\forall X\in F_d$, we have
	\*[
	&\EE_{\hat{X}\mid X}\left[\left|\sum_{i=1}^d\log \frac{f(\hat{X}_i)}{f(X_i)}-d \log\frac{1-\hat{z}_{d+1}}{1-z_{d+1}}-(\textrm{Term I}) - (\textrm{Term II}) \right|\cdot\Ind_{\{\hat{z}_{d+1}\le 1-\epsilon\}}\right]=\bigO(dh^3),
	\]
where	
    \*[
	\textrm{Term I} &:= \sum_i\left[(\log f(X_i))' (\hat{X}_i-X_i)\right]+\frac{d}{1-z_{d+1}}(\hat{z}_{d+1}-z_{d+1}),\\ 
	\textrm{Term II}&:=\frac{1}{2}\left\{\sum_i\left[(\log f(X_i))''(\hat{X}_i-X_i)^2\right]+\frac{d}{(1-z_{d+1})^2}(\hat{z}_{d+1}-z_{d+1})^2\right\}.
	\]
We will see that, similar to the CLT arguments used for analyzing RWM \citeps{Roberts1997,Yang2020optimal} for general targets, the mean of the Gaussian approximation result is determined by the sum of the mean of Term I and the mean of Term II whereas the variance of the Gaussian approximation result is determined by the variance of Term I only. {In the rest of the proof, we first approximate Terms I and II separately, then combine them in the end.}
\subsubsection{Term I}
Using \cref{eq_approx_hat_X_i} and \cref{eq_approx_whole_L1}, with the definition of $F_d$, we can get the following approximation of $\hat{X}_i-X_i$:
	\*[
	\sup_{X\in F_d}\EE_{\hat{X}\mid X}\left[\left|\hat{X}_i-X_i- \frac{1+\lambda}{2\lambda}(\hat{z}_{d+1}-z_{d+1})X_i+\frac{R}{1-z_{d+1}}(\hat{z}_i-z_i)\right|\cdot \Ind_{\{\hat{z}_{d+1}\le 1-\epsilon\}}\right]=\bigO(d^{1/2}h^2).
	\]
Next, let $\hat{z}-z=(\hat{z}_1-z_1,\dots,\hat{z}_{d+1}-z_{d+1})^T$, we can approximate the ``Term I'' by
	\*[
		\sup_{X\in F_d}\EE_{\hat{X}\mid X}\left[\left|(\textrm{Term I}) -\frac{R}{1-z_{d+1}}\cdot (\textrm{Inner Product Term}) \right|\cdot\Ind_{\{\hat{z}_{d+1}\le 1-\epsilon\}}\right] =o(d^{-1/2}),
	\]
	where the ``Inner Product Term'' represents
	\*[
	 \left((\log f(X_1))',\dots,(\log f(X_d))',\sum_{i=1}^d \left(\frac{1+\lambda}{2\lambda}(\log f(X_i))'z_i+\frac{1}{R}\right)\right) \cdot (\hat{z}-z).
	\]
The key observation is that: we can see the inner product term as the ``projection'' of $\hat{z}-z$ to the particular ``direction'' defined by
\*[
\tilde{v}:= \left((\log f(X_1))',\dots,(\log f(X_d))',\sum_{i=1}^d \left({\frac{1+\lambda}{2\lambda}}(\log f(X_i))'z_i+\frac{1}{R}\right)\right).
\]
This observation allows us to just study the distribution of Term I through projections of the current locations and the proposal for this particular direction. Note that the approximation results of \cref{eq_approx_whole_L1} and its proof hold not only for all coordinates but also for projections to any direction due to the symmetry of the sphere.

The ``projection'' of $z$ onto $\tilde{v}$ is
\*[
\tilde{z}:=\frac{\tilde{v}\cdot z}{\|\tilde{v}\|}&=\frac{\left[\sum_{i=1}^d\left((\log f(X_i))'z_i\right)\right]+z_{d+1}\left[\sum_{i=1}^d \left({\frac{1+\lambda}{2\lambda}}(\log f(X_i))'z_i+\frac{1}{R}\right)\right]}{\sqrt{\sum_{i=1}^d \left((\log f(X_i))'\right)^2+\left[\sum_{i=1}^d \left({\frac{1+\lambda}{2\lambda}}(\log f(X_i))'z_i+\frac{1}{R}\right)\right]^2}}\\
&=\frac{({\frac{2\lambda}{1+\lambda}}+z_{d+1})\left[\sum_{i=1}^d \left({\frac{1+\lambda}{2\lambda}}(\log f(X_i)'z_i+\frac{1}{R}\right)\right]-{\frac{2\lambda}{1+\lambda}}\frac{d}{R}}{\sqrt{\sum_{i=1}^d \left((\log f(X_i))'\right)^2+\left[\sum_{i=1}^d \left({\frac{1+\lambda}{2\lambda}}(\log f(X_i))'z_i+\frac{1}{R}\right)\right]^2}}\in [-1,1]
    \]
{Then we can write
\*[
	\|\tilde{v}\|^2-|\tilde{v}\cdot z|^2&=\sum_{i=1}^d \left((\log f(X_i))'\right)^2+\left[\sum_{i=1}^d \left({\frac{1+\lambda}{2\lambda}}(\log f(X_i))'z_i+\frac{1}{R}\right)\right]^2\\
	&-\left({\frac{2\lambda}{1+\lambda}}+z_{d+1}\right)^2\left[\sum_{i=1}^d \left({\frac{1+\lambda}{2\lambda}}(\log f(X_i))'z_i+\frac{1}{R}\right)\right]^2-{\left(\frac{2\lambda}{1+\lambda}\right)^2}\frac{d^2}{R^2}\\
	&\quad+{\frac{2\lambda}{1+\lambda}}\frac{2d\left({\frac{2\lambda}{1+\lambda}}+z_{d+1}\right)}{R}\left[\sum_{i=1}^d \left({\frac{1+\lambda}{2\lambda}}(\log f(X_i))'z_i+\frac{1}{R}\right)\right].
	    \]
	Note that by the definition of $F_d$, we know
 $\sum_{i=1}^d(\log f(X_i))'X_i+d=o(d)$ and ${\frac{2\lambda}{1+\lambda}}+z_{d+1}=1+\bigO(d^{-1/3})$. Then, we have
	\*[
	\sup_{X\in F_d}\EE_{\hat{X}\mid X}\left[\left|\left(\|\tilde{v}\|^2-|\tilde{v}\cdot z|^2\right)- (\textrm{Term III}) \right|\cdot\Ind_{\{\hat{z}_{d+1}\le 1-\epsilon\}}\right]=o(d^{1/2}),
	\]
	where the ``Term III'' is defined by
	\*[
	\textrm{Term III}:=
	\sum_{i=1}^d \left((\log f(X_i))'\right)^2+{\frac{2\lambda}{1+\lambda}}\frac{2d}{R}\sum_{i=1}^d \left({\frac{1+\lambda}{2\lambda}}(\log f(X_i))'z_i\right)+{\left[\frac{4\lambda}{1+\lambda}-\left(\frac{2\lambda}{1+\lambda}\right)^2\right]}\frac{d^2}{R^2}.
	\]
}

Let $\tilde{z}'$ be the ``projection'' of $\hat{z}$ onto $\tilde{v}$. Then, the inner product term can be written as
$\tilde{v}\cdot(\hat{z}-z)=\|\tilde{v}\|(\tilde{z}'-\tilde{z})$. Next, we can approximate the ``projection'' of $\hat{z}$ onto $\tilde{v}$ using the equivalent result of \cref{eq_approx_whole_L1} for $\tilde{z}$
\*[
\sup_{X\in F_d}\EE_{\hat{X}\mid X}\left[\left|\tilde{z}'- \tilde{a}_d\left(\tilde{z}-\sqrt{1-\tilde{z}^2}h \tilde{U}\right)\right|\right]=\bigO(h^2),
\]
where $\tilde{U}$ is a standard Gaussian and $\tilde{a}_d:=\frac{1}{\sqrt{1+h^2\tilde{U}_{\perp}^2}}=a_d(1+\bigO_{\Pr}(h^2d^{1/2}))=a_d(1+\bigO_{\Pr}(d^{-3/2}))$ where $a_d:= \frac{1}{\sqrt{1+(d-1)h^2}}$. Note that we can also ``upgrade'' the relation between $\tilde{a}_d$ and $a_d$ to $L_1$ statement: $\EE[|\tilde{a}_d-a_d|]=\bigO(d^{-3/2}\log(d))$. Then by triangle inequality
\*[
\sup_{X\in F_d}\EE_{\hat{X}\mid X}\left[\left|\tilde{z}'- a_d\left(\tilde{z}-\sqrt{1-\tilde{z}^2}h \tilde{U}\right)\right|\right]=\bigO(h^2+d^{-3/2}\log(d)).
\]
Therefore, we have
\*[
\sup_{X\in F_d}\EE_{\hat{X}\mid X}\left[\left|\|\tilde{v}\|(\tilde{z}'-\tilde{z})-W_1\right|\cdot\Ind_{\{\hat{z}_{d+1}\le 1-\epsilon\}}\right]=\bigO(d^{-1/2}\log(d)),
\]
where $W_1\sim\mathcal{N}(\mu_1,\sigma_1^2)$ with
$\mu_1:=\left(a_d-1\right)\|\tilde{v}\|\tilde{z}$ and
$\sigma_1^2:=(1-\tilde{z}^2)\|\tilde{v}\|^2a_d^2h^2$. Substituting $1-\tilde{z}^2=\frac{\|\tilde{v}\|^2-|\tilde{v}\cdot z|^2}{\|\tilde{v}\|^2}$, using the definition of $F_d$, we can compute the mean and variance of $W_1$ as
	{	\*[
		\mu_1&= (a_d-1)\frac{1+\lambda}{2\lambda}\left[\sum_i(\log f(X_i))'X_i +\frac{1-\lambda}{1+\lambda}d\right]+\bigO(d^{-1/2}\log(d))\\
		&= (a_d-1)\frac{1+\lambda}{2\lambda}\left[d\EE[(\log f(X_1))'X_1] +\frac{1-\lambda}{1+\lambda}d\right]+\bigO(d^{-1/2}\log(d)), \\
		\sigma_1^2&= \frac{(1+\lambda)^2}{4\lambda}da_d^2h^2\left[\sum_i \left((\log f(X_i))'\right)^2+\frac{4}{1+\lambda}\sum_i (\log f(X_i))'X_i+\frac{4d}{(1+\lambda)^2} \right]\\
		&\qquad+\bigO(d^{-1/2}\log(d))\\
		&= \frac{(1+\lambda)^2}{4\lambda}da_d^2h^2\left[d\EE\left[\left((\log f(X_1))'\right)^2\right]+\frac{4d}{1+\lambda}\EE\left[(\log f(X_1))'X_1\right]+\frac{4d}{(1+\lambda)^2} \right]\\
		&\qquad+\bigO(d^{-1/2}\log(d)).
		    \]
	}

\subsubsection{Term II} Next, we approximate Term II in the Taylor expansion. Clearly, the variance of Term II can be ignored since it is of the order $o(d^{-1})$. We only need to focus on the expectation. We start with computing
\*[
\sup_{X\in F_d}\EE_{\hat{X}\mid X}\left[(\log f(X_i))''(\hat{X}_i-X_i)^2\cdot \Ind_{\{\hat{z}_{d+1}\le 1-\epsilon\}}\right].
\]
Note that
\*[
\sup_{X\in F_d}\EE_{\hat{X}\mid X}\left[\left|(\hat{X}_i-X_i)^2-(\textrm{Term IV})\right|\cdot \Ind_{\{\hat{z}_{d+1}\le 1-\epsilon\}}\right]=o(d^{1/2}h^2),
\]
where the ``Term IV'' is defined as
\*[
{\left(\frac{1+\lambda}{2\lambda}\right)^2}(\hat{z}_{d+1}-z_{d+1})^2X_i^2 +\frac{R^2}{(1-z_{d+1})^2}(\hat{z}_i-z_i)^2+{\frac{1+\lambda}{2\lambda}}\frac{2RX_i}{1-z_{d+1}}(\hat{z}_{d+1}-z_{d+1})(\hat{z}_i-z_i).
\]
Taking expectation over $\hat{X}$ given $X$ and using
\*[
\EE_{\hat{X}\,|\,X}[\hat{z}_i-z_i]&= \left(a_d-1\right)z_i+o(d^{-1})\\
\EE_{\hat{X}\,|\,X}[(\hat{z}_i-z_i)^2]&=a_d^2h^2 + \left[\left(a_d-1\right)^2 - a_d^2h^2\right]z_i^2+o(d^{-1}),
\]
substituting to the ``Term IV'', then uniformly over $X\in F_d$, we have
{
	\*[
	\EE_{\hat{X}\,|\,X}[(\hat{X}_i-X_i)^2]&=\left(\frac{1+\lambda}{2\lambda}\right)^2X_i^2\left[a_d^2h^2(1-z_{d+1}^2)+(1-a_d)^2z_{d+1}^2\right]\\
	&+\frac{R^2}{(1-z_{d+1})^2}[a_d^2h^2(1-z_i^2)+(1-a_d)^2z_i^2]\\
	&+\frac{1+\lambda}{2\lambda}\frac{2RX_i}{1-z_{d+1}}(1-a_d)^2z_iz_{d+1}+o(d^{-1})+o(d^{1/2}h^2).
	    \]
}
This implies that, uniformly over $X\in F_d$, we have
	\*[
	&\sum_i\EE_{\hat{X}\mid X}[(\hat{X}_i-X_i)^2]= \left(\frac{1+\lambda}{2\lambda}\right)^2\frac{d}{\lambda}\left[a_d^2h^2(1-z_{d+1}^2)+(1-a_d)^2z_{d+1}^2\right]\\
	&\qquad +\frac{R^2}{(1-z_{d+1})^2}da_d^2h^2-da_d^2h^2+(1-a_d)^2d +\frac{1+\lambda}{2\lambda}\frac{d}{\lambda}(1-a_d)^2z_{d+1}+\bigO(d^{-1/2})\\
	&\quad=\frac{R^2}{(1-z_{d+1})^2}da_d^2h^2+\bigO(d^{-1/2})= (1-a_d^2)d\frac{(1+\lambda)^2}{4\lambda}+\bigO(d^{-1/2}).
	    \]
Therefore, uniformly on $F_d$, we have
{
	\*[
	&\sum_{i=1}^d\EE_{\hat{X}\,|\,X}\left[(\log f(X_i))''(\hat{X}_i-X_i)^2\right]+\frac{d}{(1-z_{d+1})^2}\EE_{\hat{X}\,|\,X}[(\hat{z}_{d+1}-z_{d+1})^2]\\
	&=\frac{R^2}{(1-z_{d+1})^2}a_d^2h^2\left[\sum_{i=1}^d  (\log f(X_i))''-{\left[1-\left(\frac{1+\lambda}{2\lambda}\right)^2(1-z_{d+1}^2)\right]}\sum_{i=1}^d\left((\log f(X_i))'' z_i^2\right)\right]\\
	&\quad + \frac{R^2}{(1-z_{d+1})^2}\left(1-a_d\right)^2{\left[1+\left(\frac{1+\lambda}{2\lambda}\right)^2z_{d+1}^2+\frac{1+\lambda}{\lambda}z_{d+1}\right]}\sum_{i=1}^d\left((\log f(X_i))'' z_i^2\right)\\
	&\quad+{\frac{1+z_{d+1}}{1-z_{d+1}}}d a_d^2(h)h^2+{\frac{d(1-a_d)^2z_{d+1}^2}{(1-z_{d+1})^2}}+\bigO(d^{-1/2}).
	    \]
}
Finally, using the definition of $F_d$, we can have
{	\*[
	&2\cdot\textrm{Term II}+\bigO(d^{-1/2}\log(d))\\
	&=\frac{(1+\lambda)^2}{4\lambda}da_d^2h^2\sum_i \EE\left[(\log f(X_i))''\right]- a_d^2h^2\left(1-\frac{1}{\lambda}\right)\sum_i \EE\left[(\log f(X_i))''X_i^2\right]\\
	&\quad +\frac{(1+\lambda)^2}{4\lambda}(1-a_d)^2\sum_i \EE \left[(\log f(X_i))''X_i^2\right]+\frac{1}{\lambda}da_d^2h^2+\frac{(1-\lambda)^2}{4\lambda^2}(1-a_d)^2d.
	    \]
}
\subsubsection{Combing Term I and Term II}
Finally, we combine Term I and Term II. Using the assumption $\lim_{x\to\pm\infty}xf'(x)=0$ which implies $\lim_{x\to\pm\infty}f'(x)=0$, one can show
\*[
&\EE_f[(\log f)''+((\log f)')^2]=\int\left[\frac{f(x)f''(x)-(f'(x))^2}{f(x)^2}+\left(\frac{1}{f(x)}f'(x)\right)^2\right]f(x)\dee x=\int f''(x)\dee x=0,\\
&\EE_f\left[X(\log f)'\right]=\int \left[\frac{x}{f(x)}f'(x)\right]f(x)\dee x=\int xf'(x) \dee x= 0-\int f(x)\dee x=-1.
    \]
Using the above results and $(d-1)a_d^2h^2= 1-a_d^2$, we have
	\*[
	&\mu_1= -(1-a_d)\frac{1+\lambda}{2\lambda}d\EE_f[X(\log f)'] - (1-a_d)\frac{1-\lambda}{2\lambda}d+\bigO(d^{-1/2}\log(d)),\\
	&\EE[\textrm{Term II}]= (1-a_d) \frac{(1+\lambda)^2}{4\lambda}d\EE_f[(\log f)'']+\bigO(d^{-1/2}\log(d)),\\
	&\frac{\sigma_1^2}{2}= (1-a_d)\frac{(1+\lambda)^2}{4\lambda}d\EE_f[((\log f)')^2]
	+(1-a_d)\frac{1+\lambda}{\lambda}d\EE_f[X(\log f)']\\
	&\qquad+(1-a_d)\frac{d}{\lambda}+\bigO(d^{-1/2}\log(d)).
	    \]
	Therefore, the mean and variance of the Gaussian distribution satisfy
	{\*[
	\mu&=\mu_1+\EE[\textrm{Term II}]+\bigO(d^{-1/2}\log(d))=d(1-a_d)\frac{(1+\lambda)^2}{4\lambda}\left\{ \frac{4\lambda}{(1+\lambda)^2}+\EE\left[\frac{\partial^2\log\pi}{\partial x^2}\right]\right\},\\
	\frac{\sigma^2}{2}&=\frac{\sigma_1^2}{2}+\bigO(d^{-1/2}\log(d))=d(1-a_d)\frac{(1+\lambda)^2}{4\lambda}\left\{ \EE\left[\left(\frac{\partial\log\pi}{\partial x}\right)^2\right]-\frac{4\lambda}{(1+\lambda)^2}\right\}.
	    \]
	}
	Using $1-a_d=\frac{\ell^2}{2d}{\frac{4\lambda}{(1+\lambda)^2}}$, we have
	\*[
	\mu&= \frac{\ell^2}{2}\left\{\frac{4\lambda}{(1+\lambda)^2}-\EE_f\left[\left((\log f)'\right)^2\right]\right\},\quad
	\sigma^2= \ell^2\left\{\EE_f\left[\left((\log f)'\right)^2\right]-\frac{4\lambda}{(1+\lambda)^2}\right\}.
	\]
	Therefore, we have shown a Gaussian random variable $W_{\hat{X}\mid X}\sim \mathcal{N}(\mu,\sigma^2)$ which satisfies
	\*[
	\sup_{X\in F_d}\EE_{\hat{X}\mid X}\left[\left( (\textrm{Term I} + \textrm{Term II})\cdot \Ind_{\{\hat{z}_{d+1}\le 1-\epsilon\}} - W_{\hat{X}\mid X}\right)^2\right]=\bigO(d^{-1/2}\log(d)).
	\]
	Combining everything together, we have shown
	\*[
	&\sup_{X\in F_d}\EE_{\hat{X}\mid X}\left[\left|\log \frac{\pi(\hat{X})(R^2+\|\hat{X}\|^2)^d}{\pi(X)(R^2+\|X\|^2)^d}\Ind_{\{\hat{z}_{d+1}\le 1-\epsilon\}}- W_{\hat{X}\mid X}\right| \right]\\
	&=\bigO(\sqrt{d^{-1/2}\log(d)})=o(d^{-1/4}\log(d)),
	\]
	which completes the proof.
\end{proof}

\subsection{Proof of \cref{thm_ESJD}}\label{proof-thm_ESJD}

\begin{proof}
By the i.i.d.~assumption, it suffices to consider the first coordinate as
\*[
\ESJD=	d\cdot\EE_{X\sim\pi}\EE_{\hat{X}\,|\,X}\left[(\hat{X}_1-X_1)^2\left(1\wedge 	\frac{\pi(\hat{X})(R^2+\|\hat{X}\|^2)^d}{\pi(X)(R^2+\|X\|^2)^d}\right)\right].
\]
We first follow a similar approach as the proof of \cref{lemma_CLT}. 
We define a sequence of ``typical sets'' $\{F_d\}$ such that $\pi(F_d)\to 1$. Define the sequence of sets $\{F_d\}$ by
\begin{equation}\label{def_Fd}
	    \begin{split}
F_d:=&\left\{x\in \Reals^d: \left|\frac{1}{d-1}\sum_{i=2}^d\left[(\log f(x_i))'\right]^2-\EE_f\left[((\log f)')^2\right]\right|<d^{-1/8} \right\}\\
& \cap \left\{x\in \Reals^d: \left|\frac{1}{d-1}\sum_{i=2}^d\left[(\log f(x_i))''\right]-\EE_f\left[((\log f)'')\right]\right|<d^{-1/8}\right\} \\
&\cap \left\{x\in \Reals^d: \left|\frac{1}{d-1}\sum_{i=2}^dx_i(\log f(x_i))' -\EE_f\left[X(\log f)' \right]\right|<d^{-1/8}\right\}\\
&\cap 
\left\{x\in \Reals^d: \left|\frac{1}{d}\sum_{i=1}^d x_i^2-\EE_f(X^2)\right|<d^{-1/6}\log(d) \right\}\\
&\cap \left\{x\in \Reals^d: |x_1|<d^{1/5}\right\}.
    \end{split}
	\end{equation}
Following the arguments in \cref{proof-lemma_CLT}, there exists a constant $0<\epsilon<1$ such that $\sup_{X\in F_d}\Pr(\hat{z}_{d+1}>1-\epsilon)=o(1)$. We define
\*[
\bar{F}_d:=\{x\in \Reals^d: z_{d+1}\le 1-\epsilon\}.
\]
We will rule out the ``bad'' proposals such that $\hat{X}\in \bar{F}_d^c$. 
By Cauchy--Schwarz inequality
\*[
	&d\cdot\EE_{X\sim\pi}\EE_{\hat{X}\,|\,X}\left[(\hat{X}_1-X_1)^2\Ind_{\{X\notin F_d\}\cup\{\hat{X}\notin\bar{F}_d\}}\right]\\
	&\le d \sqrt{\EE[(\hat{X}_1-X_1)^4]}\sqrt{\Pr(\{X\notin F_d\}\cup \{\hat{X}\notin \bar{F}_d\})}=o(d^2h^2)=o(1).
 \]
Therefore, we have 
\begin{equation}\label{temp222}
    \begin{split}
&\frac{\ESJD}{d}=\EE\left\{(\hat{X}_1-X_1)^2\Ind_{\{X\in F_d,\hat{X}\in\bar{F}_d\}}\left(1\wedge\exp\left[\log\frac{\pi(\hat{X})}{\pi(X)}+\log\frac{(R^2+\|\hat{X}\|^2)^d}{(R^2+\|X\|^2)^d}\right]\right)\right\}+o(d^{-1})\\
&\le \sup_{x\in F_d}\EE_{\hat{X}\mid x}\left\{(\hat{X}_1-x_1)^2\left(1\wedge\exp\left[\log\frac{\pi(\hat{X})}{\pi(x)}+\log\frac{(R^2+\|\hat{X}\|^2)^d}{(R^2+\|x\|^2)^d}\right]\right)\Ind_{\{\hat{X}\in \bar{F}_d\}}\right\}+o(d^{-1})
    \end{split}
\end{equation}

{In order to remove the dependence on $\hat{X}_1$ from $\|\hat{X}\|^2$, we construct a coupling $\tilde{X}\overset{d}{=}\hat{X}$} such that only $\hat{X}_{2:d}=\tilde{X}_{2:d}$ and $\tilde{X}_1$ is identical distributed as $\hat{X}_1$. Furthermore, $\tilde{X}_1$ is independent with $\hat{X}_1$ conditional on $\hat{X}_{2:d}$:
\*[
\tilde{X}:=(\tilde{X}_1,\dots,\tilde{X}_d),\quad \tilde{X}_{2:d}=\hat{X}_{2:d},\quad \left(\tilde{X}_1\perp \!\!\! \perp \hat{X}_1\mid \hat{X}_{2:d}\right),\quad \tilde{X}_1\overset{d}{=}\hat{X}_1.
\]
Using a similar argument, we can replace $\Ind_{\{\hat{X}\in \bar{F}_d\}}$ by $\Ind_{\{\hat{X}\in\bar{F}_d, \tilde{X}\in\bar{F}_d\}}$ in the first term of the r.h.s.~of \cref{temp222}, which gives
\begin{equation}
    \label{temp345}
    \begin{split}
&\sup_{x\in F_d}\EE_{(\hat{X},\tilde{X})\mid x}\left\{ (\hat{X}_1-x_1)^2\Ind_{\{\hat{X}\in\bar{F}_d, \tilde{X}\in\bar{F}_d\}}\right.\\
&\qquad\left.\cdot\left[1\wedge\exp\left(\log\frac{f(\hat{X}_1)}{f(x_1)}+\log\frac{(R^2+\|\hat{X}\|^2)^d}{(R^2+\|\tilde{X}\|^2)^d}{+\sum_{i=2}^d\log\frac{f(\hat{X}_i)}{f(x_i)}+\log\frac{(R^2+\|\tilde{X}\|^2)^d}{(R^2+\|x\|^2)^d}}\right)\right]\right\}.     
    \end{split}
\end{equation}
Using the fact that $1\wedge \exp(\cdot)$ is $1$-Lipschitz, for two random variables $A$ and $B$, if $\sup_{x\in F_d}\EE[(\hat{X_1}-x_1)^2|A-B|\Ind_{\{\hat{X}\in\bar{F}_d, \tilde{X}\in\bar{F}_d\}}]=o(d^{-1})$, then
\*[
&\sup_{x\in F_d}\EE\left\{(\hat{X}_1-x_1)^2\Ind_{\{\hat{X}\in\bar{F}_d, \tilde{X}\in\bar{F}_d\}}\left|[1\wedge \exp(A)]-[1\wedge \exp(B)]\right|\right\}\\
&\le
\sup_{x\in F_d}\EE\left\{(\hat{X}_1-x_1)^2\left|A-B\right|\Ind_{\{\hat{X}\in\bar{F}_d, \tilde{X}\in\bar{F}_d\}}\right\}=o(d^{-1}).
 \]
Alternatively, if $\sup_{x\in F_d}\EE[(A-B)^2\Ind_{\{\hat{X}\in\bar{F}_d, \tilde{X}\in\bar{F}_d\}}]=o(1)$, then by Cauchy--Schwarz inequality
\*[
&\sup_{x\in F_d}\EE\left\{(\hat{X}_1-x_1)^2\left|A-B\right|\Ind_{\{\hat{X}\in\bar{F}_d, \tilde{X}\in\bar{F}_d\}}\right\}\\
&\le\sup_{x\in F_d}\sqrt{\EE\left[(\hat{X}_1-x_1)^4\right]}\sqrt{\EE[(A-B)^2\Ind_{\{\hat{X}\in\bar{F}_d, \tilde{X}\in\bar{F}_d\}}]} =o(dh^2)=o(d^{-1}).
	    \]

We next use the above approximations several times to approximate \cref{temp345}. {We first approximate $\log\frac{(R^2+\|\hat{X}\|^2)^d}{(R^2+\|\tilde{X}\|^2)^d}$ by $\frac{\hat{X}_1^2-\EE[\tilde{X}_1^2]}{2}$}, since by Taylor expansion and mean value theorem, using the fact $R^2=d$, one can show
\*[
\sup_{x\in F_d}\EE\left[(\hat{X}_1-x_1)^2 \left|\log\frac{(R^2+\|\hat{X}\|^2)^d}{(R^2+\|\tilde{X}\|^2)^d}- \frac{\hat{X}_1^2-\tilde{X}_1^2}{2}\right|\Ind_{\{\hat{X}\in\bar{F}_d, \tilde{X}\in\bar{F}_d\}}\right]=\bigO(d^{-1/2}dh^2),
\]
and
\*[
\sup_{x\in F_d}\EE\left[(\hat{X}_1-x_1)^2 \left| \frac{\tilde{X}_1^2-\EE[\tilde{X}_1^2]}{2}\right|\Ind_{\{\hat{X}\in\bar{F}_d, \tilde{X}\in\bar{F}_d\}}\right]=\bigO(d^{-1/2}dh^2).
\]
Next, we consider the conditional expectation over $\tilde{X}$ of the term $\sum_{i=2}^d\log\left(\frac{f(\tilde{X}_i)}{f(x_i)}\right)$. Since $\sup_{x\in F_d}\EE[|\sum_{i=2}^d (\tilde{X}_i-x_i)^3|^2\Ind_{\{\tilde{X}\in\bar{F}_d\}}]=o(d^{-1})$, we can {approximate $\sum_{i=2}^d\log\left(\frac{f(\tilde{X}_i)}{f(x_i)}\right)$ by the first two terms of its Taylor expansion}. Using Taylor expansion and mean value theorem, we have
\begin{equation}\label{temp1}
	    \begin{split}
&\sup_{x\in F_d}\EE_{\tilde{X}\mid x}\left[\left(\sum_{i=2}^d\log\left(\frac{f(\tilde{X}_i)}{f(x_i)}\right){-\textrm{(Term I + Term II)} }\right)^2\Ind_{\{\tilde{X}\in\bar{F}_d\}}\right]\\
&=\bigO\left(\sup_{x\in F_d}\EE_{\tilde{X}\mid x}\left[\left|\sum_{i=2}^d(\tilde{X}_i-x_i)^3\right|^2\Ind_{\{\tilde{X}\in\bar{F}_d\}}\right]\right)=o(d^{-1}),
	    \end{split}
	\end{equation}
where 
\*[
\textrm{Term I + Term II}:=\sum_{i=2}^d\left[(\log f(x_i))'(\tilde{X}_i-x_i)+\frac{1}{2}(\log f(x_i))''(\tilde{X}_i-x_i)^2 \right].
\]
Furthermore,  following the Gaussian approximation arguments in the proof of \cref{lemma_CLT}, we have the following result:
\*[
\sup_{x\in F_d}\EE_{\tilde{X}\mid x}\left[\left(\textrm{(Term I + Term II)}+\log\frac{(R^2+\|\tilde{X}\|^2)^d}{(R^2+\|x\|^2)^d}-W_{\tilde{X}\mid x}\right)^2\Ind_{\{\tilde{X}\in\bar{F}_d\}}\right]=o(1)
\]
where $W_{\tilde{X}\mid x}\sim \mathcal{N}(\mu,\sigma^2)$ with
\*[
\mu&= \frac{\ell^2}{2}\left\{1-\EE_f\left[\left((\log f)'\right)^2\right]\right\},\quad
\sigma^2= \ell^2\left\{\EE_f\left[\left((\log f)'\right)^2\right]-1\right\}.
\]

Therefore, after the above approximations, we have simplified \cref{temp345} to
\[\label{temp333}
&\sup_{x\in F_d}\EE_{(\hat{X}_1,\tilde{X})\mid x}\left\{ (\hat{X}_1-x_1)^2\Ind_{\{\hat{X}\in\bar{F}_d\}}\left[1\wedge\exp\left(\log\frac{f(\hat{X}_1)}{f(x_1)}+\frac{\hat{X}_1^2-\EE[\tilde{X}_1^2]}{2}+W_{\tilde{X}\mid x}\right)\right]\right\}.
\]
Note that the randomness of $\log\frac{f(\hat{X}_1)}{f(x_1)}+\frac{\hat{X}_1^2-\EE[\tilde{X}_1^2]}{2}$ only comes from $\hat{X}_1$, and the randomness of $W_{\tilde{X}\mid x}$ comes from $\tilde{X}$. However, $W_{\tilde{X}\mid x}$ is not independent with $\hat{X}_1$, because of the weak dependence of $\hat{X}_1$ with $\tilde{X}_{2:d}=\hat{X}_{2:d}$. Indeed, it can be argued that $W_{\tilde{X}\mid x}$ is asymptotically independent with $\hat{X}_1$. However, we wish $W_{\tilde{X}\mid x}$ to be independent with $\hat{X}_1$ in order to make further progress.

In \cref{construct_tilde_W}, we will construct another random variable $\tilde{W}$, which is identically distributed with $W_{\tilde{X}\mid x}$, but independent with $\hat{X}_1$. Furthermore, $\tilde{W}$ is ``physically close'' to $W_{\tilde{X}\mid x}$ such that 
\[\label{temp444}
\sup_{x\in F_d}\EE\left[\left(\tilde{W}-W_{\tilde{X}\mid x}\right)^2\right]=o(1).
\]
Note that the existence of $\tilde{W}$ is very intuitive: if $W_{\tilde{X}\mid x}$ is asymptotically independent with $\hat{X}_1$, then there should be an identically distributed random variable $\tilde{W}$, which is independent with $\hat{X}_1$ and dependent with $W_{\tilde{X}\mid x}$. Furthermore, $W_{\tilde{X}\mid x}$ becomes ``physically closer and closer'' to $\tilde{W}$.

Using \cref{temp444} and the Lipschitz property of $1\wedge \exp(\cdot)$, we can simplify \cref{temp333} to
\begin{equation}\label{temp555}
    \begin{split}
&\sup_{x\in F_d}\EE_{(\hat{X}_1,\tilde{X})\mid x}\left\{ (\hat{X}_1-x_1)^2\Ind_{\{\hat{X}\in\bar{F}_d\}}\left[1\wedge\exp\left(\log\frac{f(\hat{X}_1)}{f(x_1)}+\frac{\hat{X}_1^2-\EE[\tilde{X}_1^2]}{2}+W_{\tilde{X}\mid x}\right)\right]\right\}\\
&\to \sup_{x\in F_d}\EE_{(\hat{X},\tilde{X})\mid x}\left\{ (\hat{X}_1-x_1)^2\Ind_{\{\hat{X}\in\bar{F}_d\}}\left[1\wedge\exp\left(\log\frac{f(\hat{X}_1)}{f(x_1)}+\frac{\hat{X}_1^2-\EE[\tilde{X}_1^2]}{2}+\tilde{W}\right)\right]\right\}\\
&\to \sup_{x\in F_d}\EE_{(\tilde{W},\tilde{X}_1)\mid x}\left\{ (\hat{X}_1-x_1)^2\left[1\wedge\exp\left(\log\frac{f(\hat{X}_1)}{f(x_1)}+\frac{\hat{X}_1^2-\EE[\tilde{X}_1^2]}{2}+\tilde{W}\right)\right]\right\}\\
&=\sup_{x\in F_d}\EE_{\hat{X}_1\mid x}\left\{ (\hat{X}_1-x_1)^2\EE_{\tilde{W}\mid\hat{X}_1,x}\left[1\wedge\exp\left(\log\frac{f(\hat{X}_1)}{f(x_1)}+\frac{\hat{X}_1^2-\EE[\tilde{X}_1^2]}{2}+\tilde{W}\right)\right]\right\}.
    \end{split}
\end{equation}

\begin{lemma}{\citeps[Proposition 2.4]{Roberts1997}}\label{lemma_Gaussian}
	If $W\sim \mathcal{N}(\mu,\sigma^2)$ then
	\*[
	\EE_W[1\wedge \exp(W)]=\Phi(\mu/\sigma)+\exp(\mu+\sigma^2/2)\Phi(-\sigma-\mu/\sigma)
	\]
	where $\Phi(\cdot)$ is the standard normal cumulative distribution function.
\end{lemma}

By \cref{lemma_Gaussian}, the inside expectation of \cref{temp555} is
\*[
M(\hat{X}_1):&=\EE_{\tilde{W}\mid\hat{X}_1,x}\left[1\wedge\exp\left(\log\frac{f(\hat{X}_1)}{f(x_1)}+\frac{\hat{X}_1^2-\EE[\tilde{X}_1^2]}{2}+\tilde{W}\right)\right]\\
&=\Phi\left(-\frac{\sigma}{2}+\frac{\log\left(\frac{f(\hat{X}_1)}{f(x_1)}\right)+\frac{\hat{X}_1^2-\EE[\tilde{X}_1^2]}{2}}{\sigma}\right)\\ &\quad+\exp\left(\log\left(\frac{f(\hat{X}_1)}{f(x_1)}\right)+\frac{\hat{X}_1^2-\EE[\tilde{X}_1^2]}{2}\right)\cdot\Phi\left(-\frac{\sigma}{2}-\frac{\log\left(\frac{f(\hat{X}_1)}{f(x_1)}\right)+\frac{\hat{X}_1^2-\EE[\tilde{X}_1^2]}{2}}{\sigma}\right).
    \]
Note that $\EE[\tilde{X}_1^2]=x_1^2+\bigO(d^{-1})$. Then it can be easily verified that
\*[
M(x_1)\to 2\Phi\left(-\frac{\sigma}{2}\right).
\]
Therefore, substituting $M(\hat{X}_1)$ to \cref{temp555} and by Taylor expansion, as $d\to\infty$, we have
	\*[
\ESJD&\to d\cdot \left\{\EE[(\hat{X}_1-x_1)^2]\cdot 2\Phi\left(-\frac{\sigma}{2}\right)\right\}\\
&\to 2\ell^2\cdot\Phi\left(-\frac{\ell}{2}\sqrt{\EE_f\left[\left((\log f)'\right)^2\right]-1}\right),
    \]
which completes the proof. Therefore, it suffices to show the construction of $\tilde{W}$.

\subsubsection{Construction of $\tilde{W}$}\label{construct_tilde_W}

For simplicity of notations, we will write $W_{\tilde{X}\mid x}$ as $W$. 
Recall the following definitions in \cref{proof-lemma-lati}: all the randomness comes from a standard Gaussian $\mathcal{N}(0,I_d)$. We can write it as $d$ independent standard $\mathcal{N}(0,1)$ as $V_1,\dots,V_d$. We have also defined $U_1,\dots,U_{d+1}$ which are marginal Gaussian random variables with $U_i\sim\mathcal{N}(0,1)$. Recall that $U_1,\dots,U_{d+1}$ are usually not independent. However, in the stationary phase when $R=d^{1/2}$, $\{U_i\}$ are ``almost independent'' since their correlations are very small. Particularly, if $z_i\to 0$ then $U_i$ is asymptotically independent with all $\{U_j: j\neq i\}$.

Recall that in \cref{proof-lemma-lati}, we have derived that the proposed $\hat{z}_i$ can be written as
\*[
\hat{z}_i=\frac{\sqrt{1+h^2U_i^2}}{\sqrt{1+h^2\left(\sum_{j=1}^d V_j^2\right)}}\left(\frac{z_i}{\sqrt{1+h^2U_i^2}}+\frac{\sqrt{1-z_i^2}hU_i}{\sqrt{1+h^2U_i^2}}\right),\quad i=1,\dots,d+1.
\]
Also,
$\hat{X}_i=d^{1/2}\frac{\hat{z}_i}{1-\hat{z}_{d+1}}$ as $R=d^{1/2}$.
To emphasize the sources of randomness, we write
\*[
\hat{z}_i=\mathcal{F}_{\hat{z}_i}\left(\sum_{j=1}^dV_j^2,U_i\right),\quad i=1,\dots,d+1,
\]
to denote that the randomness of $\hat{z}_i$ comes from $\sum_j V_j^2$ and $U_i$. Similarly, we can write $\hat{X}_i$ as
\*[ \hat{X}_i=\mathcal{F}_{\hat{X}_i}\left(\sum_{j=1}^dV_j^2,U_i,U_{d+1}\right), \quad i=1,\dots,d
\]
to denote the randomness of $\hat{X}_i$ comes from $\sum_j V_j^2$, $U_i$, and $U_{d+1}$.

Now recall that $W$ is a function of $\tilde{X}$, where $\tilde{X}_1$ is conditional independent with $\hat{X}_1$, and $\tilde{X}_i=\hat{X}_i$ for all $i=2,\dots,d$. We can therefore write $W$ as a random function of $\hat{X}_{2:d}$. Therefore, we can write the following:
\*[
W= \mathcal{F}_{W}\left(\sum_{j=1}^dV_j^2,U_{2:d+1}\right),\quad \hat{X}_1=\mathcal{F}_{\hat{X}_1}\left(\sum_{j=1}^dV_j^2,U_1,U_{d+1}\right).
\]

Now, it is clear that the dependence of $W$ on $\hat{X}_1$ has three sources:
\begin{enumerate}
	\item Shared source of randomness, $\sum_{j=1}^d V_j^2$, which a chi-squared random variable with degree of freedom $d$;
	\item The weak dependence between $U_{2:d}$ and $U_1$. Note that $U_1,U_{2:d}$ are correlated and joint Gaussian distributed. When $z_1\to 0$, then $U_1$ becomes (asymptotically) independent with $U_{2:d}$.
	\item Shared source of randomness $U_{d+1}$, which becomes (asymptotically) independent with $U_{1:d}$ when $z_{d+1}\to 0$.
\end{enumerate}
In order to replace $W$ by an identical distributed $\tilde{W}$, that is also independent with $\hat{X}_1$,
we do three coupling arguments as follows:
\begin{enumerate}
	\item First coupling argument: we first {replace $\sum_{j=1}^d V_j^2$ by a constant $d$} to define
	\*[
	W':=\mathcal{F}_W(d,U_{2:d+1})
	\]
	Following the proof of \cref{lemma_CLT}, it's very easy to check that $W'$ has the same distribution as $W$. The key is to show $W'$ and $W$ are ``physically close'' so that
	\*[
	\sup_{x\in F_d}\EE[(W'-W)^2]=o(1).
	\]
	Actually, we can show a better rate that
	\*[
	\sup_{x\in F_d}\EE[(W'-W)^2]=\bigO(d^{-1}).
	\]
	The proof is delayed to \cref{coupling1}.
	\item Second coupling argument: we remove the dependence of $U_{2:{d+1}}$ on $U_1$. Note that $U_{2:{d+1}}$ are ``almost independent'' with $U_1$. We ``orthogonalize'' these Gaussian random variables $U_{2:d+1}$ by the following decomposition:
	\*[
	U_i=c_{i}^{\parallel} U_1 + c_{i}^{\perp}U_{i}^{\perp}, \quad i=2,\dots,d+1
	\]
	where $c_{i}^{\parallel}$ is the component that $U_i$ is ``parallel'' with $U_1$, and $c_{i}^{\perp}$ is the ``orthogonal'' (that is, independent) component.
	Now we {replace the $U_1$ in the decomposition by an independent copy $\tilde{U}_1$}. That is, we define
	\*[
	U_i':=c_{i}^{\parallel} \tilde{U}_1 + c_{i}^{\perp}U_{i}^{\perp}, \quad i=2,\dots,d+1,
	\]
	where $\tilde{U}_1$ is independent standard Gaussian random variable.
	Note that $U_i'$ is highly correlated with $U_i$ since the component $c_{i}^{\perp}U_{i}^{\perp}$ remains.
	If $z_1=0$, for example, we recover $U_i'=U_i$. Next, we define
	\*[
	W'':=\mathcal{F}_W(d,U'_{2:d+1}).
	\]
	It is clear that $W''$ is identically distributed as $W'$. When $z_1=0$, we recover $|W''-W'|=0$. Intuitively, $W''$ is ``physically close'' to $W'$, when $z_1$ is small. Actually, our definition of $F_d$ guarantees that for all $x\in F_d$ we have $x_1=\bigO(d^{1/5})$ which implies $z_1=\bigO(d^{1/5-1/2})$ is indeed small.
	Then, our goal is to show
	\*[
	\sup_{x\in F_d}\EE[(W''-W')^2]=o(1).
	\]
	Actually, we can prove a slightly better rate, which is
	\*[
	\sup_{x\in F_d}\EE[(W'-W'')^2]=\bigO(z_1^2)=\bigO(d^{-3/5}).
	\]
	The proof is delayed to \cref{coupling2}.
	\item Third coupling argument: finally, we construct $\tilde{W}$ based on $W''$. Recall that $W''=\mathcal{F}_W(d,U_{2:d+1}')$. In order to construct $\tilde{W}$ which is independent with $\hat{X}_1$, it suffices to replace the source of $U_{d+1}'$ by an independent copy $\tilde{U}_{d+1}$. We first ``orthogonalize'' $U_{2:d}'$ using the same way as the previous coupling argument:
	\*[
U_i'=\tilde{c}_{i}^{\parallel} U_{d+1}' + \tilde{c}_{i}^{\perp}{U'_{i}}^{\perp}, \quad i=2,\dots,d.
\]
and then replace the component $U_{d+1}'$ by an independent copy $\tilde{U}_{d+1}$ to define $\tilde{U}_{2:d}$ by
	\*[
\tilde{U}_i:=\tilde{c}_{i}^{\parallel} \tilde{U}_{d+1} + \tilde{c}_{i}^{\perp}{U'_{i}}^{\perp}, \quad i=2,\dots,d.
\]
Finally, we define
	\*[
	\tilde{W}:=\mathcal{F}_W(d,\tilde{U}_{2:d+1}).
	\]
	It is clear that $\tilde{W}$ is identically distributed as $W''$ and it is now independent with $\hat{X}_1$. Then the goal is to show 
	\*[
	\sup_{x\in F_d}\EE[(\tilde{W}-W'')^2]=o(1).
	\]
	It can be seen that the ``orthogonalization'' part of the coupling argument is the same as the second coupling. We then expect this part cause a distance of order $\bigO_{\Pr}(z_{d+1})$, which is indeed the case.
	By the definition of $F_d$, for all $x\in F_d$ we have $\|x_i\|^2-d=o(d^{6/7})$ so $z_{d+1}=o(d^{6/7-1})=o(d^{-1/7})$. The other part of the coupling argument is to replace the direct dependence on $U_{d+1}'$ by $\tilde{U}_{d+1}$, which we can prove a distance of order $\bigO(d^{-1/8})$. Overall, we can show
	\*[
\sup_{x\in F_d}\EE[(\tilde{W}-W'')^2]=\bigO(d^{-1/4})+\bigO(z_{d+1}^2)=\bigO(d^{-1/4})+o(d^{-2/7}).
	\]
	The proof is delayed to \cref{coupling3}.
\end{enumerate}
Overall, we have shown the construction of $\tilde{W}$, which has the same distribution as $W$ and is independent with $\hat{X}_1$. Most importantly, we have shown
\*[
\sup_{x\in F_d}\EE[(\tilde{W}-W)^2]&\le \sup_{x\in F_d}\EE[(|\tilde{W}-W''|+|W''-W'|+|W'-W|)^2]\\
&=\bigO(d^{-1/4})+o(d^{-2/7})+\bigO(d^{-3/5})+\bigO(d^{-1})=o(1).
\]

\subsubsection{The first coupling argument}\label{coupling1}
We first define the replacement of $\hat{z}_i$ by
\*[
z_i':=\frac{\sqrt{1+h^2U_i^2}}{\sqrt{1+h^2d}}\left(\frac{z_i}{\sqrt{1+h^2U_i^2}}+\frac{\sqrt{1-z_i^2}hU_i}{\sqrt{1+h^2U_i^2}}\right),\quad i=1,\dots,d+1,
\]
in which we only replace $\sum_j V_j^2$ in $\hat{z}_i$ by a constant $d$. It is clear that $\hat{z}_i$ and $z_i'$ are highly correlated, we can write
\[\label{temp2}
\hat{z}_i=(1+\bigO_{\Pr}(h^2d^{1/2}))\frac{\sqrt{1+h^2U_i^2}}{\sqrt{1+h^2d}}\left(\frac{z_i}{\sqrt{1+h^2U_i^2}}+\frac{\sqrt{1-z_i^2}hU_i}{\sqrt{1+h^2U_i^2}}\right).
\]
Following the proof of \cref{lemma_CLT}, it is obvious that $W'$ has the same distribution as $W$, which means the mean of $W-W'$ is zero. Therefore, we only need to focus how fast the variance of $W-W'$ goes to zero.

Recall that in the proof of $W$ using \cref{lemma_CLT}, the variance of $W$ is dominated by the variance of the following inner product term
	\*[
&\frac{R}{1-z_{d+1}}\underbrace{\left(0,(\log f(x_2))',\dots,(\log f(x_d))',\sum_{i=2}^d \left((\log f(x_i))'z_i+\frac{1}{R}\right)\right) \cdot (\hat{z}-z)}_{\textrm{inner product}}
\]
When we prove the variance of $W-W'$, by replacing $\hat{z}-z$ to the difference of $\hat{z}-z$ and $z'-z$, which is $\hat{z}-z'$, following the same arguments as the proof of \cref{lemma_CLT}, the variance of $W-W'$ is then determined by
the variance of
	\*[
&\frac{R}{1-z_{d+1}}\left(0,(\log f(x_2))',\dots,(\log f(x_d))',\sum_{i=2}^d \left((\log f(x_i))'z_i+\frac{1}{R}\right)\right) \cdot (\hat{z}-z')
\]
where $\hat{z}-z'=(\hat{z}_1-z_1',\dots,\hat{z}_{d+1}-z_{d+1}')^T$. Substituting the definition of $\hat{z}$ and $z'$, using $\bigO(h^2d^{1/2})=\bigO(d^{-3/2})$, we get the order of variance as
$\bigO\left(\left(\sum_{i=1}^d (\log f(x_i))'x_i\right)^2\right)\bigO(d^{-3})$.
By the definition of $F_d$, we know $\sum_{i=1}^d \left(\log f(x_i))'x_i\right)=\bigO(d)$.
Therefore, we have
\*[
\sup_{x\in F_d}\EE[(W-W')^2]=\bigO(d^2)\bigO(d^{-3})=\bigO(d^{-1}).
\]

	\subsubsection{The second coupling argument}\label{coupling2}
	By the construction of $W''$, it is clear that $W''$ has the same distribution as $W'$.
	Define the replacement of $z_i'$ by
	\*[
	z_i'':=\frac{\sqrt{1+h^2(U_i')^2}}{\sqrt{1+h^2d}}\left(\frac{z_i}{\sqrt{1+h^2(U_i')^2}}+\frac{\sqrt{1-z_i^2}hU_i'}{\sqrt{1+h^2(U_i')^2}}\right),\quad i=1,\dots,d+1,
	\]
	Then we can again follow the proof of \cref{lemma_CLT} and focus on the variance of $W'-W''$ (since the mean of $W'-W''$ is clearly zero). Then the variance of $W'-W''$ is determined by the variance of
		\*[
	&\frac{R}{1-z_{d+1}}\left(0,(\log f(x_2))',\dots,(\log f(x_d))',\sum_{i=2}^d \left((\log f(x_i))'z_i+\frac{1}{R}\right)\right) \cdot (z''-z')
	\]
	where $z''-z'=(z''_1-z_1',\dots,z''_{d+1}-z_{d+1}')^T$.
	Using the fact that
	\*[
		U_i-U_i'=c_{i}^{\parallel} (U_1-\tilde{U}_1), \quad i=2,\dots,d+1,
	\]
	we can get that the variance of $W'-W''$ is bounded by
	\*[
	\left\{d\sum_{i=2}^d((\log f(x_i))')^2+ \left[\sum_{i=2}^d\left((\log f(x_i))'x_i+1\right)\right]^2\right\}\left\{2h^2\sum_{i=2}^d(c_i^{\parallel})^2+o(h^2)\right\}
	\]
	Note that, by the definition of $F_d$, we have
	\*[
	d\sum_{i=2}^d((\log f(x_i))')^2=\bigO(d^2),\quad \left[\sum_{i=2}^d\left((\log f(x_i))'x_i+1\right)\right]^2=o(d^2).
	\]
	Using some basic geometry to analyze the ``angle'' between $U_i$ and $U_1$, we know $c_i^{\parallel}=\bigO(z_iz_1)$. Therefore, using $h^2=\bigO(d^{-2})$, we have
	\*[
	\sup_{x\in F_d}\EE[(W'-W'')^2]=\bigO\left({\sum_{i=2}^d(c_i^{\parallel})^2}\right)=\bigO\left({\sum_{i=2}^dz_i^2}z_1\right)=\bigO(z_1^2).
	\]
	Finally, for all $x\in F_d$, we have $x_1=\bigO(d^{1/5})$ which implies that $z_1=\bigO(d^{1/5-1/2})=\bigO(d^{-3/10})$.

	\subsubsection{The third coupling argument}\label{coupling3}
	By the construction of $\tilde{W}$, it is clear that the distribution of $\tilde{W}$ is the same as the distribution of $W''$. Therefore, the mean of $\tilde{W}-W''$ is zero.
	
	We define the replacement of $z_i''$ by
	\*[
	\tilde{z}_i:=\frac{\sqrt{1+h^2\tilde{U}_i^2}}{\sqrt{1+h^2d}}\left(\frac{z_i}{\sqrt{1+h^2\tilde{U}_i^2}}+\frac{\sqrt{1-z_i^2}h\tilde{U}_i}{\sqrt{1+h^2\tilde{U}_i^2}}\right),\quad i=1,\dots,d+1.
	\]
	Again, following the proof of \cref{lemma_CLT}, the variance of $\tilde{W}-W''$ is determined by the variance of
	\*[
	&\frac{R}{1-z_{d+1}}\left(0,(\log f(x_2))',\dots,(\log f(x_d))',\sum_{i=2}^d \left((\log f(x_i))'z_i+\frac{1}{R}\right)\right) \cdot (\tilde{z}-z'')\\
	&=	\underbrace{\frac{R}{1-z_{d+1}}\left(0,(\log f(x_2))',\dots,(\log f(x_d))',0\right) \cdot (\tilde{z}-z'')}_{\textrm{the first term}}\\
	&\quad +\underbrace{\frac{R}{1-z_{d+1}}\left(0,,\dots,0,\sum_{i=2}^d \left((\log f(x_i))'z_i+\frac{1}{R}\right)\right) \cdot (\tilde{z}-z'')}_{\textrm{the second term}}
	    \]
	where $\tilde{z}-z''=(\tilde{z}_1-z_1'', \dots,\tilde{z}_{d+1}-z_{d+1}'')$. Note that we have decomposed the inner product to the sum of two terms. We compute the variances of the two terms separately.
	
	Similar to the previous coupling argument, the variance of the first term is bounded by
	\*[
	\left[d\sum_{i=2}^d((\log f(x_i))')^2\right]\left[2h^2\sum_{i=2}^d(\tilde{c}_i^{\parallel})^2+o(h^2)\right]=\bigO(z_{d+1}^2).
	\]
	The second term can be written as
	\*[
	\sum_{i=2}^d\left[(\log f(x_i))'x_i+1\right](1+\bigO_{\Pr}(z_{d+1}))(\tilde{z}_{d+1}-z''_{d+1}).
	\]
	By the definition of $F_d$, we know 
	\*[
		\sum_{i=2}^d\left[(\log f(x_i))'x_i+1\right]=d\bigO(d^{-1/8}),\quad z_{d+1}=o(d^{-1/7})
	\]
	Furthermore, by the definition of $\tilde{z}_{d+1}$ and $z_{d+1}''$, both have finite exponential moments. Moreover, we have
	\*[
\EE\left[\left|\tilde{z}_{d+1}-z_{d+1}''-\sqrt{1-z_{d+1}^2}h(\tilde{U}_{d+1}-U_{d+1}')\right|\right]=\bigO(d^{-2}).
	\]
	Therefore, we have the second term
		\*[
	&\sum_{i=2}^d\left[(\log f(x_i))'x_i+1\right](1+\bigO_{\Pr}(z_{d+1}))(\tilde{z}_{d+1}-z''_{d+1})\\
	&=d\bigO(d^{-1/8})(1+o(d^{-1/7}))h\bigO_{\Pr}(1)=\bigO_{\Pr}(d^{-1/8}),
	    \]
	and its variance is of order $\bigO(d^{-1/4})$.
	Overall, adding the variance of the second term to the variance of the first term yields
	\*[
	\sup_{x\in F_d}\EE[(\tilde{W}-W'')^2]=\bigO(z_{d+1}^2)+\bigO(d^{-1/4})=o(d^{-2/7})+\bigO(d^{-1/4}).
	\]
	This completes the proof.
\end{proof}

\subsection{Proof of \cref{thm_diffusion}}\label{proof-thm_diffusion}

\begin{proof}
We follow the framework of \citeps{Roberts1997} using the generator approach \citeps{Ethier1986}. Define the (discrete-time) generator of $x$ by
\*[
(G_d V)(x):=d \EE_{\hat{X}}\left\{ [V(\hat{X})-V(x)]\left(1\wedge\frac{\pi(\hat{X})(R^2+\|\hat{X}\|^2)^d}{\pi(x)(R^2+\|x\|^2)^d}\right)\right\}
\]
for any function $V$ for which this definition makes sense. In the Skorokhod topology, it doesn't cause any problem to treat $G_d$ as a continuous time generator. We shall restrict attention to test functions such that $V(x)=V(x_1)$.
We show uniform convergence of $G_d$ to $G$, the generator of the limiting (one-dimensional) Langevin diffusion, for a suitable large class of real-valued functions $V$, where, for some fixed function $h(\ell)$,
\*[
(GV)(x_1):=h(\ell)\left\{\frac{1}{2}V''(x_1)+\frac{1}{2}\left[(\log f)'(x_1)\right]V'(x_1)\right\}.
\]
Since $f'/f$ is Lipschitz, by \citeps[Thm 2.1 in Ch.8]{Ethier1986}, a core for the generator has domain $C_c^{\infty}$, which is the class of continuous functions with compact support such that all orders of derivatives exist. This enable us to restrict attentions to functions $V\in C_c^{\infty}$ such that $V(x)=V(x_1)$.

Define the sequence of sets $\{F_d\}$ by
\begin{equation}\label{def_Fd}
	    \begin{split}
F_d:=&\left\{x\in \Reals^d: \left|\frac{1}{d-1}\sum_{i=2}^d\left[(\log f(x_i))'\right]^2-\EE_f\left[((\log f)')^2\right]\right|<d^{-1/8} \right\}\\
& \cap \left\{x\in \Reals^d: \left|\frac{1}{d-1}\sum_{i=2}^d\left[(\log f(x_i))''\right]-\EE_f\left[((\log f)'')\right]\right|<d^{-1/8}\right\} \\
&\cap \left\{x\in \Reals^d: \left|\frac{1}{d-1}\sum_{i=2}^dx_i(\log f(x_i))' -\EE_f\left[X(\log f)' \right]\right|<d^{-1/8}\right\}\\
&\cap 
\left\{x\in \Reals^d: \left|\frac{1}{d}\sum_{i=1}^d x_i^2-\EE_f(X^2)\right|<d^{-1/6}\log(d) \right\}\\
&\cap \left\{x\in \Reals^d: |x_1|<d^{1/5}\right\}
    \end{split}
	\end{equation}
Then, for fixed $t$, by the union bound, Markov inequality, and the assumptions \cref{assump1}, we have
\*[
&\Pr\left(X^d(\lfloor ds\rfloor)\notin F_d,\exists 0\le s\le t\right)\le td\Pr_{\pi}(X\notin F_d)\\
&=\bigO\left(
td\left(\frac{d^{1/8}}{(d-1)^{1/2}}\right)^4+td\left(\frac{d^{1/6}}{\log(d)d^{1/2}}\right)^3+td\left(\frac{1}{d^{1/5}}\right)^6\right)=\bigO(t(\log(d))^{-3}).    
\]
Therefore, for any fixed $t$, if $d\to\infty$, the probability of all $\{X^d(\lfloor ds\rfloor), 0\le s\le t\}$ are in $F_d$ goes to $1$. It then suffices to consider only $x\in F_d$.

Next, we decompose the expectation over $\hat{X}=(\hat{X}_1,\dots,\hat{X}_d)$ into expectations over $\hat{X}_1$ and $\hat{X}_{2:d}:=(\hat{X}_2,\dots,\hat{X}_{d})$. Recall that $\hat{X}_1=\hat{X}'_1$, $\hat{X}_{2:d}=\hat{X}''_{2:d}$, and $\hat{X}'$ is independent with $\hat{X}''$. Therefore, $\hat{X}_1$ is independent with $\hat{X}_{2:d}$. Then, we have
\*[
(G_d V)(x)=&d \EE_{\hat{X}_{1}}\left\{ [V(\hat{X}_1)-V(x_1)]\EE_{{\hat{X}_{2:d}}}\left[1\wedge\frac{\pi(\hat{X})(R^2+\|\hat{X}\|^2)^d}{\pi(x)(R^2+\|x\|^2)^d}\right]\right\}\\
=&d \EE_{\hat{X}_{1}}\left\{ [V(\hat{X}_1)-V(x_1)]\EE_{\hat{X}''_1\,|\,\hat{X}_{2:d}}\EE_{{\hat{X}_{2:d}}}\left[1\wedge\frac{\pi(\hat{X})(R^2+\|\hat{X}\|^2)^d}{\pi(x)(R^2+\|x\|^2)^d}\right]\right\}\\
=&d \EE_{\hat{X}_{1}}\left\{ [V(\hat{X}_1)-V(x_1)]\EE_{{\hat{X}''}}\left[1\wedge\frac{\pi(\hat{X})(R^2+\|\hat{X}\|^2)^d}{\pi(x)(R^2+\|x\|^2)^d}\right]\right\}.
    \]
Now we focus on the inner expectation
\*[
&\EE_{{\hat{X}''}}\left[1\wedge\frac{\pi(\hat{X})(R^2+\|\hat{X}\|^2)^d}{\pi(x)(R^2+\|x\|^2)^d}\right]\\
&=\EE_{{\hat{X}''}}\left[1\wedge \exp\left\{d\log\left(\frac{R^2+\|\hat{X}\|^2}{R^2+\|\hat{X}''\|^2}\right)+{\sum_{i=1}^d\log\left(\frac{f(\hat{X}_i)}{f(x_i)}\right)+d\log\left(\frac{R^2+\|\hat{X}''\|^2}{R^2+\|x\|^2}\right)}\right\}\right].
    \]
 Following the proof of \cref{thm_ESJD}, for any $x\in F_d$, we can replace $\sum_{i=2}^d\log\left(\frac{f(\hat{X}_i)}{f(x_i)}\right)+d\log\left(\frac{R^2+\|\hat{X}''\|^2}{R^2+\|x\|^2}\right)$ by 
$W\sim \mathcal{N}(\mu,\sigma^2)$ is a Gaussian random variable with
\*[
\mu&= \frac{\ell^2}{2}\left\{1-\EE_f\left[\left((\log f)'\right)^2\right]\right\},\quad
{\sigma}^2= \ell^2\left\{\EE_f\left[\left((\log f)'\right)^2\right]-1\right\}.
\]
To keep the dependence on $\hat{X}_1$, we denote $r(\hat{X}_1):=d\EE_{\hat{X}''}\left[\log\left(\frac{R^2+\|\hat{X}\|^2}{R^2+\|\hat{X}''\|^2}\right)\right]$. Using Taylor expansion we can verify that $r(x_1)=\bigO(d^{-1})$ and $r'(x_1)=x_1+o(d^{-1/2})$. Therefore, it suggests that we can approximate
\*[
d\log\left(\frac{R^2+\|\hat{X}\|^2}{R^2+\|\hat{X}''\|^2}\right)+{\sum_{i=1}^d\log\left(\frac{f(\hat{X}_i)}{f(x_i)}\right)+d\log\left(\frac{R^2+\|\hat{X}''\|^2}{R^2+\|x\|^2}\right)}
\]
by $\log\left(\frac{f(\hat{X}_1)}{f(x_1)}\right)+r(\hat{X}_1)+W$.
Hence, we define 
\[\label{eq_temp2}
(\tilde{G}_d V)(x):=&d \EE_{\hat{X}_{1}}\left\{[V(\hat{X}_1)-V(x_1)]\cdot\EE_{\hat{X}''}\left[1\wedge \exp\left\{\log\left(\frac{f(\hat{X}_1)}{f(x_1)}\right)+r(\hat{X}_1)+W\right\}\right]\right\}.
\]
Since $V\in C_c^{\infty}$, we have for some $Z_1\in (x_1,\hat{X}_1)$ or $(\hat{X}_1,x_1)$ that
\*[
&V(\hat{X}_1)-V(x_1)=V'(x_1)(\hat{X}_1-x_1)+\frac{1}{2}V''(Z_1)(\hat{X}_1-x_1)^2.
\]
This implies that $\EE_{\hat{X}_{1}}\left[	V(\hat{X}_1)-V(x_1)\right]=\bigO(x_1d^{-1})=o(d^{-1/2})$ and $\EE\left[	\left|V(\hat{X}_1)-V(x_1)\right|\right]=\bigO(d^{-1/2})$, since $x_1=\bigO(d^{1/5})$ for $x\in F_d$. Using the fact that the function $1\wedge \exp(x)$ is Lipschitz \citeps[Proposition 2.2]{Roberts1997}, we can have 
\*[
\sup_{x\in F_d}\left|(G_dV)(x)-(\tilde{G}_dV)(x)\right|&=d\sup_{x\in F_d}\EE_{\hat{X}_{1}}\left[V(\hat{X}_1)-V(x_1)\right]\bigO(d^{-1/2})\\
&\quad +d\sup_{x\in F_d}\EE_{\hat{X}_{1}}\left[\left|V(\hat{X}_1)-V(x_1)\right|\left(o(d^{-1/2})\right)\right]=o(1).
    \]
Therefore, we can now concentrate on $(\tilde{G}_d V)(x)$ defined in \cref{eq_temp2} for $x\in F_d$. 

Note that conditional on $\hat{X}_1$, the term inside the inner expectation of \cref{eq_temp2} is Gaussian distributed, since $\log\left(\frac{f(\hat{X}_1)}{f(x_1)}\right)+r(\hat{X}_1)+W\sim \mathcal{N}(\mu',\sigma^2)$ where $\mu':=\mu+\log\left(\frac{f(\hat{X}_1)}{f(x_1)}\right)+r(\hat{X}_1)$. Therefore, by \cref{lemma_Gaussian}, we have
\*[
M(\hat{X}_1):&=\EE_{\hat{X}''}\left[1\wedge \exp\left\{\log\left(\frac{f(\hat{X}_1)}{f(x_1)}\right)+r(\hat{X}_1)+W\right\}\right]\\
&=\Phi\left(-\frac{\sigma}{2}+\frac{\log\left(\frac{f(\hat{X}_1)}{f(x_1)}\right)+r(\hat{X}_1)}{\sigma}\right)\\ &\quad+\exp\left(\log\left(\frac{f(\hat{X}_1)}{f(x_1)}\right)+r(\hat{X}_1)\right)\cdot\Phi\left(-\frac{\sigma}{2}-\frac{\log\left(\frac{f(\hat{X}_1)}{f(x_1)}\right)+r(\hat{X}_1)}{\sigma}\right).
    \]
Since $r(x_1)=\bigO(d^{-1})$ and $r'(x_1)=x_1+o(d^{-1/2})$, it can be verified that 
\[\label{eq_temp4}
M(x_1)&\to 2\Phi\left(-\frac{\sigma}{2}\right),\quad
M'(x_1)\to \Phi\left(-\frac{\sigma}{2}\right)\left[(\log f(x_1))'{+x_1+o(d^{-1/2})}\right].
\]
Furthermore, one can verify $M''(x_1)$ is bounded since $\Phi(\cdot)$, $\Phi'(\cdot)$, and $\Phi''(\cdot)$ are all bounded functions. 

Therefore, by mean value theorem, there exist constants $K_1$ and $K_2$ such that
\*[
&d[V(\hat{X}_1)-V(x_1)]M(\hat{X}_1)\\
&=d\left[V'(x_1)(\hat{X}_1-x_1)+\frac{1}{2}V''(x_1)(\hat{X}_1-x_1)^2+K_1(\hat{X}_1-x_1)^3\right]\\
&\qquad \cdot\left[M(x_1)+M'(x_1)(\hat{X}_1-x_1)+K_2(\hat{X}_1-x_1)^2\right]\\
&=dV'(x_1)M(x_1)(\hat{X}_1-x_1)\\
&\qquad+d\left[\frac{1}{2}V''(x_1)M(x_1)+V'(x_1)M'(x_1)\right](\hat{X}_1-x_1)^2 + \bigO(|d(\hat{X}_1-x_1)^3|).
    \]
Taking expectations over $\hat{X}_1$, using \cref{eq_temp4}, as $d\to\infty$, we have
\*[
&(\tilde{G}_dV)(x)=\EE_{\hat{X}_{1}}\left[d[V(\hat{X}_1)-V(x_1)]M(\hat{X}_1)\right]\\
&\to {dV'(x_1)M(x_1)\EE[\hat{X}_1-x_1]}+d\left[\frac{1}{2}V''(x_1)M(x_1)+V'(x_1)M'(x_1)\right]\EE\left[(\hat{X}_1-x_1)^2 \right]\\
&\to-{\frac{\ell^2}{2}x_1} V'(x_1)M(x_1)+{\ell^2}\left[\frac{1}{2}V''(x_1)M(x_1)+V'(x_1)M'(x_1)\right]\\
&\to\Phi\left(-\frac{\sigma}{2}\right)\left[-\ell^2{x_1}V'(x_1)+\ell^2V''(x_1)+\ell^2V'(x_1)[(\log f(x_1))'{+x_1+o(d^{-1/2})}]\right]\\
&{\to} 2\ell^2\Phi\left(-\frac{\sigma}{2}\right)\left[\frac{1}{2}V''(x_1)+\frac{1}{2}[(\log f)'(x_1)]V'(x_1)\right]\\
&\to 2\ell^2\Phi\left(-\frac{\ell\sqrt{\EE_f\left[\left((\log f)'\right)^2\right]-1}}{2}\right)\left[\frac{1}{2}V''(x_1)+\frac{1}{2}[(\log f)'(x_1)]V'(x_1)\right].
    \]
This completes the proof.
\end{proof}

\subsection{Comments on \cref{rmk_heavy_tail}}\label{comment_rmk_heavy_tail}
To illustrate the issue that for heavy tail targets, the SPS might stuck if starting from the South Pole, we consider multivariate student's $t$ targets. For these targets, one can easily derive that the ``equator'' is always a stationary point of the density (which is maximum if $\nu<d$ and minimum if $\nu>d$). 

We first consider the case that $k=\nu/d<1$ is a small constant, from \cref{lemma-lati} and \cref{fig-gk}, one can see that in the \emph{transient phase}, the acceptance rate goes to $0$ exponentially fast as $d\to\infty$. For example, when $k=0.01$ if $z=-1$ and $\hat{z}=0$, then the acceptance rate is roughly $\exp(-1.7d)\approx 1/5^d$. 

Next, we consider an approximation for the cases $\nu=\bigO(1)$ and $d\to\infty$. Consider current location $x=0$ which corresponding to the south pole and the ``typical'' proposal $\tilde{x}$ satisfying $\|\tilde{x}\|^d\approx d$. For multivariate student's $t$ target with DoF $\nu=\bigO(1)$, we have
	\*[
	\log\frac{\pi(\tilde{x})}{\pi(x)}\approx -\frac{\nu+d}{2}\log(1+\frac{d}{\nu-2})=\bigO(d\log(d))
	\]
	and $\log\frac{(R^2+\|\tilde{x}\|^2)^d}{(R^2+\|x\|^2)^d}=\bigO(d)$. Therefore, the acceptance probability for a ``typical'' proposal starting from the South Pole is
	\*[
	\min\left\{1,\frac{\pi(\tilde{x})}{\pi(x)}\frac{(R^2+\|\tilde{x}\|^2)^d}{(R^2+\|x\|^2)^d}\right\}\approx 1/(d^d),
	\]
	which implies for multivariate student's $t$ targets with DoF $\nu=\bigO(1)$, the SPS chain initialized at the South Pole can stuck when the proposal variance is large.
	
\subsection{Proof of the Jacobian determinant \cref{eq_jacobian}}\label{proof_jacobian}
The Jacobian determinant of the stereographic projection is well-known. For example, see \citeps[Lemma 3.2.1]{gehring2017introduction}. 
One way to derive it is to calculate the Jacobian by comparing the ratio of volumes.
Let $x=\SP(z)$ and $y=\SP(w)$, one can get
\*[
\|z-w\|^2&=\sum_{i=1}^d\|z_i-w_i\|^2 + \|z_{d+1}-w_{d+1}\|^2\\
&=\sum_{i=1}^d 4\left(\frac{Rx_i}{R^2+\|x\|^2}-\frac{Ry_i}{R^2+\|y\|^2}\right)^2 + 4\left(\frac{R^2}{R^2+\|x\|^2}-\frac{R^2}{R^2+\|y\|^2}\right)^2\\
&=\frac{4R^2}{(R^2+\|x\|^2)(R^2+\|y\|^2)}\|x-y\|^2,
    \]
    where the last equality is obtained using
    \*[
    \sum_i \left[x_i(R^2+\|y\|^2)-y_i(R^2+\|x\|^2)\right]^2=-R^2(\|y\|^2-\|x\|^2)^2+\|x-y\|^2(R^2+\|x\|^2)(R^2+\|y\|^2)
    \]
Then we know the Euclidean distance from $\Reals^d$ to $\mathbb{S}^d$ is scaled by $\frac{2R}{\sqrt{(R^2+\|x\|^2)(R^2+\|y\|^2)}}$. Therefore, comparing ratio of volume elements on $\Reals^d$ and $\mathbb{S}^d$, the Jacobian determinant is proportional to $(R^2+\|x\|^2)^d$.

\section{Additional Simulations}\label{sec_more_simu}

We have included some numerical examples in \cref{section_simu}.
In this section, we give additional simulations for SPS and SBPS. \BLUE{We use ACF to denote the {\em sample autocorrelation function}, which is an estimate of the autocorrelation function $\text{acf}(k):=\frac{\EE_{\pi}[(X_i-\EE[X_i])(X_{i+k}-\EE[X_{i+k}])]}{\sqrt{\var_{\pi}(X_i)}\sqrt{\var_{\pi}(X_{i+k})}}$.}

\subsection{Some simulations using the result from \cref{lemma-lati}}\label{sec-fig-latitute-approximation-simulation}

We plot in \cref{fig-prop} for some simulations using the result from \cref{lemma-lati}.

		\begin{figure}
		\centering
		\includegraphics[width=\textwidth]{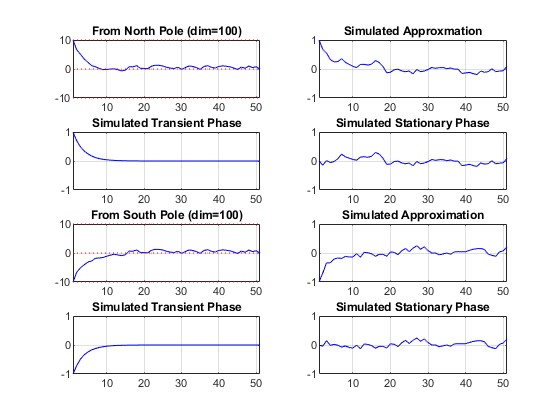}
		\caption{ We consider $h=0.1$ and $d=100$ and two cases for the initial state: from the North Pole (up) and the South Pole (down). For each case, the first subplot is the traceplot of the proposal latitudes. The other three are approximations using \cref{eq_approx_whole}, \cref{eq_approx_transient}, and \cref{eq_approx_stationary} from \cref{lemma-lati}.
		}
		\label{fig-prop}
	\end{figure}

\subsection{The function $g_k(\cdot)$ used in \cref{subsec-iso}}\label{sec-fig-gk}
	We plot the function $g_k$ in \cref{fig-gk} for different values of $k$.
	\begin{figure}[h]
	\centering
	\includegraphics[width=0.9\textwidth]{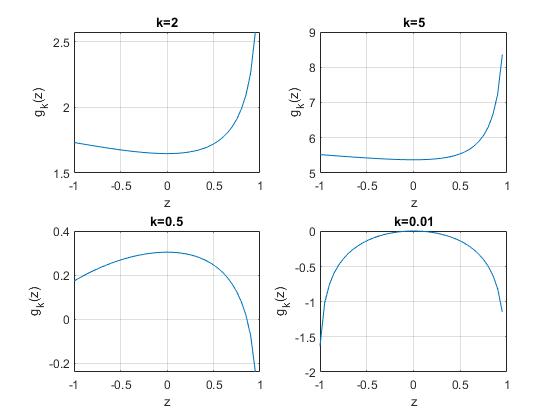}
	\caption{The function $g_{k}(z)$ used in \cref{subsec-iso} for different values of $k$.
	}
	\label{fig-gk}
\end{figure}

\subsection{SPS: traceplot and ACF}

In this example, \cref{fig-SPS-trace} shows traceplots and ACFs of SPS (first column) and RWM (the last three columns). SPS starts from the North Pole and RWM starts from different initial states $(c,c,\dots,c)$ where $c=10,20,50$, respectively (see the last three columns of \cref{fig-SPS-trace}). For SPS starting from the South Pole compared with RWM, we refer to \cref{subsec-SPS-burin-south}. We plot the traceplots and ACFs of the $1$-st coordinate and negative log-target density, as well as traceplots of the first two coordinates. The target is standard Gaussian in $d=100$ dimensions. The proposal variance for RWM is tuned such that the acceptance rate is about $0.234$, which is known to be the optimal acceptance rate \citeps{Roberts1997}. For SPS, the acceptance rate is roughly $0.78$, which is the lowest acceptance rate possible.

According to \cref{fig-SPS-trace}, the SPS mixes almost immediately. Indeed, the transient phase of SPS in $d=100$ dimensions is less than $10$ iterations. Compared with SPS, the mixing time of RWM relies on the initial value. For example, in the last column of \cref{fig-SPS-trace}, the RWM hasn't mixed after $5000$ iterations. In the stationary phase, with an acceptance rate of $0.78$, SPS generates almost uncorrelated samples according to the ACFs. Compared with SPS, the samples from RWM with an optimal acceptance rate of $0.234$ are highly correlated in high dimensions.

    \begin{figure}[h]
		\centering
		\includegraphics[width=\textwidth]{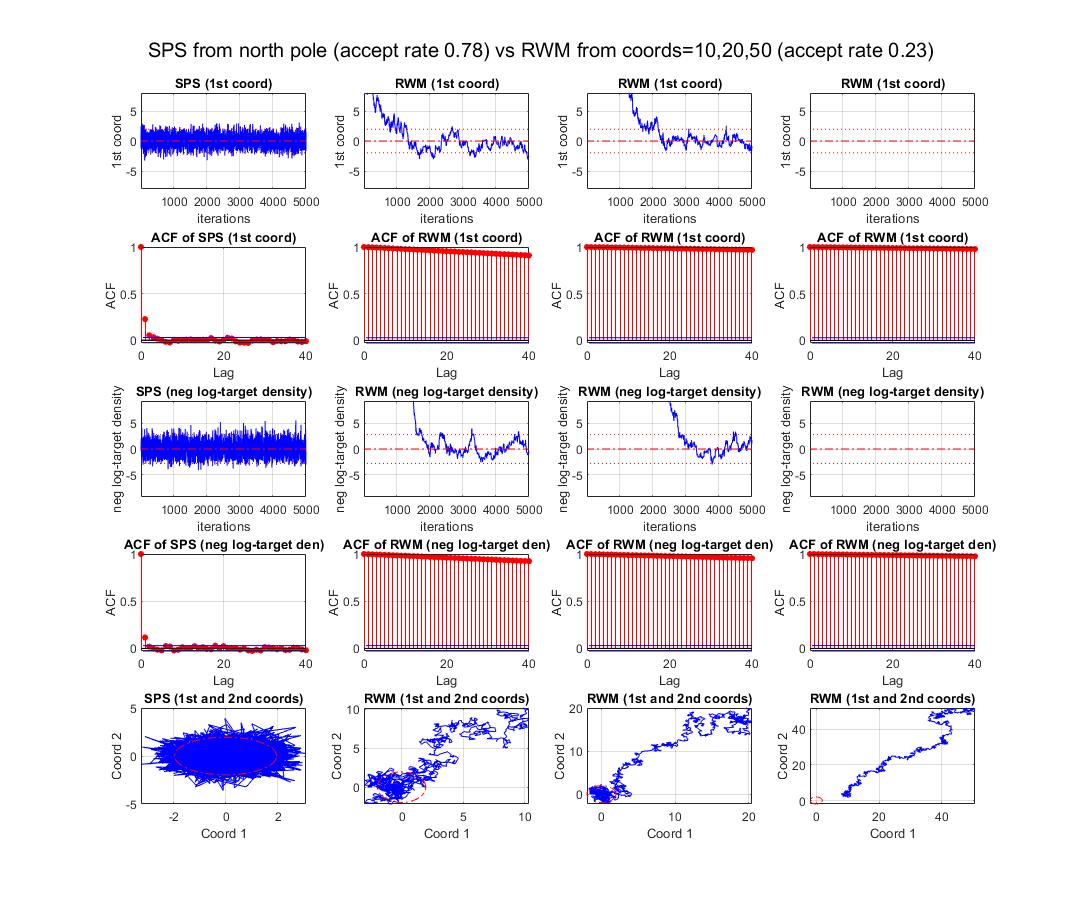}
		\caption{Traceplots and ACFs of the $1$-st coordinate, negative log-target density, and the first two coordinates, for standard Gaussian target in $100$ dimensions: SPS starts from the north pole (first column) vs RWM starts from different initial states $(c,c,\dots,c)$ where $c=10,20,50$ (the last three columns). 
		}
		\label{fig-SPS-trace}
	\end{figure}

\subsection{SPS: burn-in starting from the South Pole} \label{subsec-SPS-burin-south}
We have shown the traceplots and ACFs for starting from the North pole in \cref{fig-SPS-trace}. In this example,  \cref{fig-SPS-trace-2} shows traceplots and ACFs of SPS (the first column) starting from the South Pole and RWM (the second column) starting from the target mode. We plot the traceplots and ACFs of the $1$-st coordinate and negative log-target density, as well as traceplots of the first two coordinates. The target is standard Gaussian in $d=100$ dimensions. The proposal variance for RWM is tuned such that the acceptance rate is about $0.234$, which is known to be the optimal acceptance rate \citeps{Roberts1997}. For SPS, the acceptance rate is roughly $0.78$, which is the lowest acceptance rate possible.

According to \cref{fig-SPS-trace-2}, the SPS mixes almost immediately. Indeed, the transient phase of SPS in $d=100$ dimensions is less than $10$ iterations. Compared with SPS, even starting from the target mode, the first $100$ iterations are the transient phase of RWM. In the stationary phase, with an acceptance rate of $0.78$, SPS generates almost uncorrelated samples according to the ACFs. Compared with SPS, the samples from RWM with an optimal acceptance rate of $0.234$ are highly correlated in high dimensions.
	 \begin{figure}[h]
		\centering
		\includegraphics[width=\textwidth]{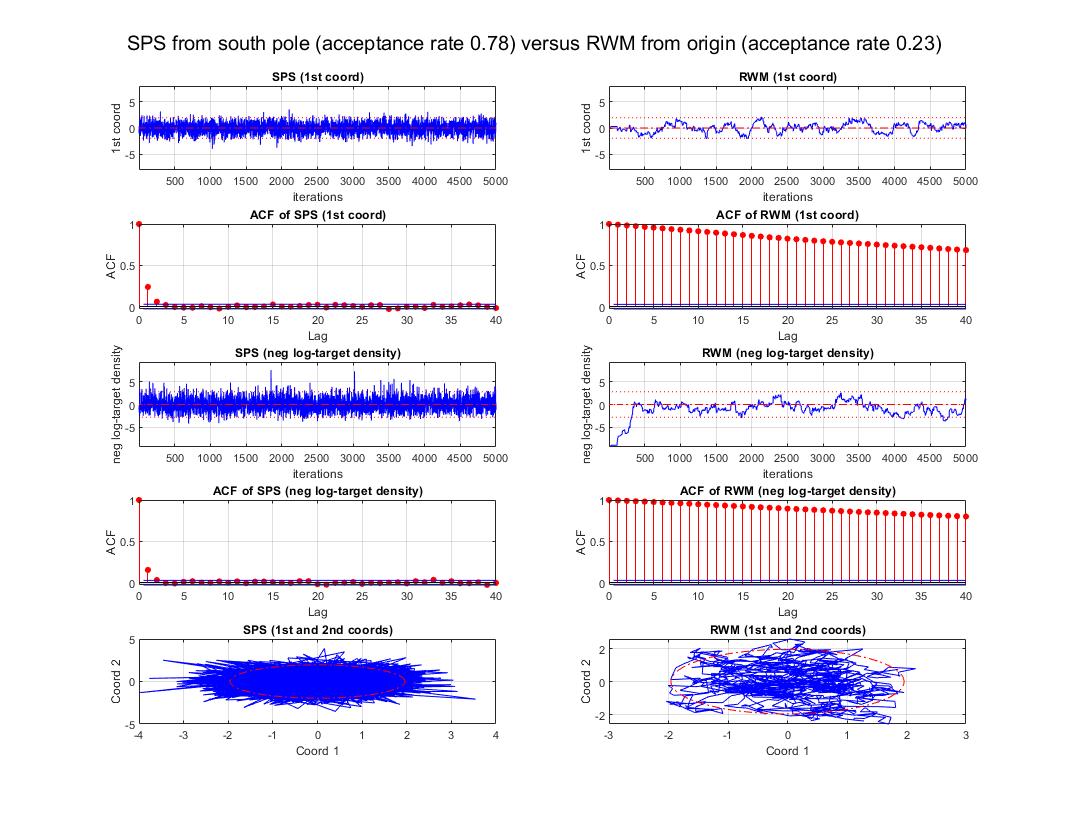}
		\caption{Traceplots and ACFs of the $1$-st coordinate, negative log-target density, and the first two coordinates, for standard Gaussian target in dimension $100$: SPS from south pole (the first column) vs RWM from the origin (the second column).  
		}
		\label{fig-SPS-trace-2}
	\end{figure}

\subsection{SPS versus GSPS}

In this example, we study SPS when the covariance matrix of the target distribution does not equal the identity matrix. We choose $R=\sqrt{d}$ and the target is multivariate student's $t$ with covariance matrix $\Sigma$ that satisfies the sum of eigenvalues of $\Sigma$ is $d$, i.e., $\lambda_1+\dots+\lambda_d=d$. We compare the decay of the performance of SPS with the performance of GSPS which uses the covariance matrix $\Sigma$.
In the simulations, we consider the cases that a sequence of covariance matrices $\Sigma_1,\Sigma_2,\dots,\Sigma_{50}$ are all block-diagonal and we add different numbers of diagonal sub-blocks which equal to a particular $2\times 2$ matrix. For example, for the case of one sub-block, denoted by $\Sigma_1$:
	\*[
	\Sigma_1:=\begin{pmatrix}
	1 & 0.8 & 0 & \cdots & 0\\
	0.8 & 1 & 0 & \cdots & 0\\
	0 & 0 & 1 & \cdots & 0\\
	\vdots  & \vdots  & \cdots & \ddots & \vdots  \\
	0 & 0 & 0 & \cdots & 1
	\end{pmatrix}
	\]
	One can easily verify the eigenvalues are
	\*[
	\Lambda_1=(0.2,1.8,1,1,\dots,1)^T
	\]
	and the first two coordinates have rotation with angle $3\pi/4$, since
	\*[
	\begin{pmatrix}
	1 & 0.8\\0.8 & 1
	\end{pmatrix}
	=\begin{pmatrix}
	\cos\frac{3\pi}{4} &-\sin\frac{3\pi}{4}\\ \sin\frac{3\pi}{4} & \cos\frac{3\pi}{4}
	\end{pmatrix}
	\begin{pmatrix}
	0.2 & 0\\ 0 & 1.8
	\end{pmatrix}
	\begin{pmatrix}
	\cos\frac{3\pi}{4} &-\sin\frac{3\pi}{4}\\ \sin\frac{3\pi}{4} & \cos\frac{3\pi}{4}
	\end{pmatrix}^T
	\]
	Thus, we have the corresponding rotation matrix used by GSPS:
	\*[
	Q_1:=\begin{pmatrix}
	\cos\frac{3\pi}{4} &-\sin\frac{3\pi}{4} & 0 & \cdots & 0\\
	\sin\frac{3\pi}{4} & \cos\frac{3\pi}{4} & 0 & \cdots & 0\\
	0 & 0 & 1 & \cdots & 0\\
	\vdots  & \vdots  & \cdots & \ddots & \vdots  \\
	0 & 0 & 0 & \cdots & 1.
	\end{pmatrix}
	\]
	Similarly, we can replace the diagonal terms by $2$ sub-blocks to define $\Sigma_2$ and the corresponding rotation matrix:
		\*[
			\Sigma_2:=\begin{pmatrix}
	1 & 0.8 & 0 & 0 & 0 & \cdots & 0\\
	0.8 & 1 & 0 & 0 & 0 & \cdots & 0\\
		0 & 0 & 1 &  0.8 &   0 & \cdots & 0\\
	0 & 0 & 0.8 & 1 & 0  & \cdots & 0\\
	0 & 0 & 0 & 0 & 1 & \cdots & 0\\
	\vdots  & \vdots  & \cdots & \ddots & \vdots  \\
	0 & 0 & 0 & 0 & 0 & \cdots & 1.
	\end{pmatrix},\quad
	Q_2:=\begin{pmatrix}
	\cos\frac{3\pi}{4} &-\sin\frac{3\pi}{4} & 0 & 0 & 0 & \cdots & 0\\
	\sin\frac{3\pi}{4} & \cos\frac{3\pi}{4} & 0 & 0 & 0 & \cdots & 0\\
		0 & 0 & \cos\frac{3\pi}{4} &-\sin\frac{3\pi}{4} &   0 & \cdots & 0\\
	0 & 0 & \sin\frac{3\pi}{4} & \cos\frac{3\pi}{4} & 0  & \cdots & 0\\
	0 & 0 & 0 & 0 & 1 & \cdots & 0\\
	\vdots  & \vdots  & \cdots & \ddots & \vdots  \\
	0 & 0 & 0 & 0 & 0 & \cdots & 1
	\end{pmatrix}.
	\]
	We can define $\Sigma_3,\dots,\Sigma_{50}$ and the corresponding rotation matrices used for GSPS in the same way.
	
	We run simulations for SPS for certain targets without using the true covariance matrices and we tune the acceptance rate of SPS to be $0.234$ (or the closest possible). For comparison with the GSPS, we assume GSPS uses the correct covariance matrices of the targets, but the proposal variances $h^2$ are chosen to be \emph{the same} as those used for SPS. We then consider the performance of GSPS as the \emph{benchmark}, since GSPS provides the optimal performance for SPS by using the information of the true target covariance matrix with the same proposal variance. We consider ACFs of two cases: the ``latitude'' and the $1$-st coordinate. In \cref{fig-ellipse-latitude}, the first and third columns show the ACFs of the latitude and first coordinate of SPS for $d=100$ dimensional multivariate student's $t$ targets with different covariance matrices indexed by the number of sub-blocks (the proposal variances are tuned to target the $0.234$ acceptance rate). For the same proposal variances, the second and fourth columns show the performance of the GSPS (the ``benchmark'') using the true target covariance matrices. From \cref{fig-ellipse-latitude}, one can see that the SPS performs well when the number of sub-blocks is small. For the first and second rows, the acceptance rates of SPS cannot be as low as $0.234$ and they equal $0.48$ and $0.30$, respectively. The ACF decreases slower as the lag when the number of sub-blocks increases. From the ``benchmark'' in the second and fourth columns, we know the proposal variance $h^2$ becomes smaller. However, even in the case of $\Sigma_{50}$ (the last row), the ACF still decreases much faster than RWM \cref{fig-SPS-trace} for Gaussian targets. Note that in this example, all the targets are heavy-tailed and we know RWM is not even geometrically ergodic so the performance of RWM would be much worse than in \cref{fig-SPS-trace}. Overall, we can conclude that SPS is much better than RWM even when the covariance matrix is $\Sigma_{50}$. The performance can be significantly improved by GSPS even without changing the proposal variance. This suggests that adaptively tuning the parameters of GSPS (in the adaptive MCMC framework) is very promising.

		\begin{figure}[h]
		\centering
		\includegraphics[width=\textwidth]{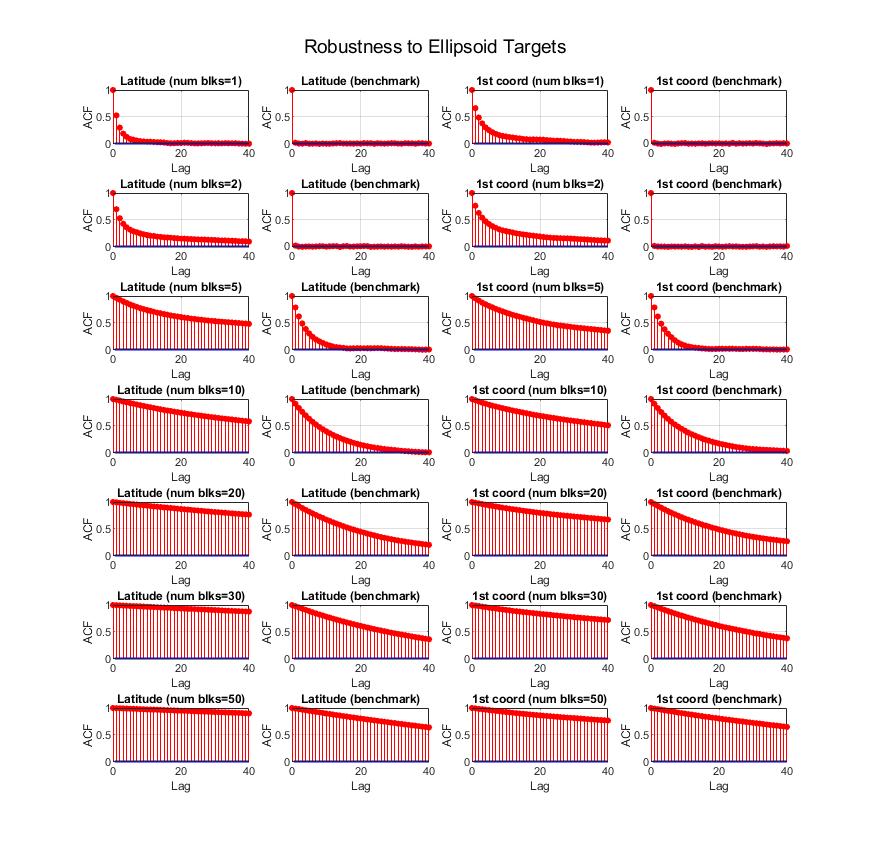}
		\caption{SPS versus GSPS: ACFs of latitude (first two columns) and $1$-st coordinate (last two columns) for SPS (first and third columns) and GSPS (second and fourth columns). The proposal variances of SPS are chosen to target the $0.234$ acceptance rate (for the first and second rows, the acceptance rates are $0.48$ and $0.30$, respectively). The proposal variances of GSPS are chosen to be the same as those for SPS to provide the ``benchmark'' performance. Target is $100$-dim ellipsoid targets using $2\times 2$ blocks; GSPS (benchmark) uses the true covariance matrices. 
		}
		\label{fig-ellipse-latitude}
	\end{figure}


\subsection{The definition of Effective Sample Size (ESS) for SBPS}\label{subsec-def-ESS}

We study the PDMP algorithms using the notion of Effective Sample Size (ESS) per Switch. We first recall the definition of ESS in \cites[supplemental material]{bierkens2019zig}. Consider a function $g:\Reals^d\to \Reals$, for a continuous process $Y(t)$, by CLT
\*[
\frac{1}{\sqrt{t}}\int_0^t \left[g(Y(s))-\pi(h)\right]\dee s \to \sigma^2_g
\]
where $\sigma^2_h$ is called \emph{asymptotic variance}.

The asymptotic variance can be estimated using batches of length $T/B$
\*[
\widehat{\sigma_g^2}=\frac{1}{B-1}\sum_{i=1}^B(X_i-\bar{X})^2,
\]
where $\bar{X}=\frac{1}{B}\sum_i X_i$ and
\*[
X_i:=\sqrt{\frac{B}{T}}\int_{(i-1)T/B}^{iT/B}g(Y(s))\dee s.
\]
We also estimate the mean and variance under $\pi$ by
\*[
\widehat{\pi(g)}:=\frac{1}{T}\int_0^T g(Y(s))\dee s,\quad \widehat{\var_{\pi}(g)}:=\frac{1}{T}\int_0^T g(Y(s))^2\dee s- \left(\widehat{\pi(g)}\right)^2.
\]
Then the ESS is estimated by
\*[
\widehat{\textrm{ESS}}:=\frac{T\, \widehat{\var_{\pi}(g)}}{\widehat{\sigma_g^2}}
\]
The ESS per Switch is estimated by
\*[
\textrm{ESS per Switch}=\frac{\widehat{\textrm{ESS}}}{\textrm{Number of Events Simulated}}
\]

We consider three cases for $g$, one is the first coordinate, the second one is the {\emph{negative log-target density}:
	\*[
	g(t)=\sqrt{d}\left(\frac{\|Y(t)\|^2}{d}-1\right)\sim \mathcal{N}(0,2),
	\]
	where the distribution holds for standard Gaussian targets.
} The third one is the square of the first coordinate.

\begin{remark}
If in a simulation of an event, there's only one evaluation of the full likelihood, then ESS per Switch is the same as ESS per epoch.
\end{remark}

\subsection{SBPS: traceplot and ACF}
In this example, we study SBPS via the traceplots and ACFs for the $1$-st coordinate, the negative log density, and the squared $1$-st coordinate, respectively. In \cref{fig-pdmp-trace-d}, the target is multivariate student's $t$ distribution with $d=100$ degrees of freedom. Since SBPS is a continuous-time process, for the traceplot and ACF, we further discretize every unit period into $5$ samples. $N=1000$ events are simulated with a low refresh rate of $0.2$.
For other settings such as high refresh rate, light-tailed targets, and different covariance matrices, we will study in later subsections for additional simulations for SBPS.

There are several interesting observations from \cref{fig-pdmp-trace-d}. The ACFs for all three cases have certain periodic behaviors and can take negative values for the first two cases. We also include the ESS per Switch for the three cases. 
As a result of negative ACFs, the ESS per Switch is larger than $1$ in the first two cases. This suggests asymptotic variance for estimating the $1$-st coordinate or the negative log density using SBPS is even smaller than the variance using $N$ independent samples from the target.

	\begin{figure}[h]
		\centering
		\includegraphics[width=\textwidth]{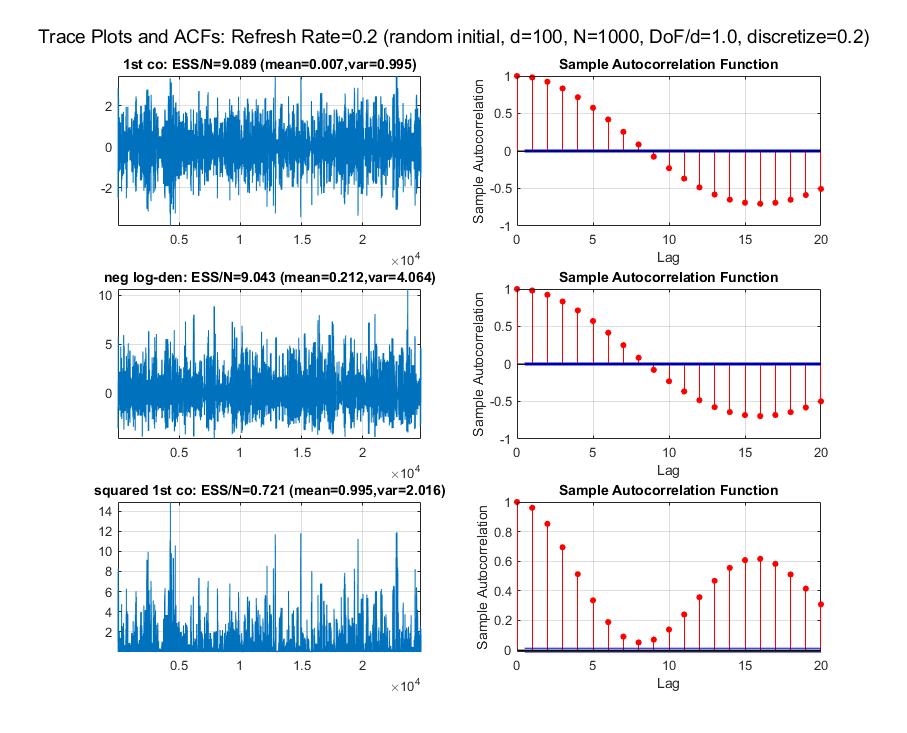}
		\caption{Trace plots and ACFs for the $1$-st coordinate, the negative log-density, and the squared $1$-st coordinate. Target distribution is multivariate student's $t$ distribution with DoF$=d=100$. Every unit period is discretized into $5$ samples. $N=1000$ events are simulated.  Low refresh rate $=0.2$. The ESS per switch is larger than $1$ and the ACF can be negative in the first two cases. 
		}
		\label{fig-pdmp-trace-d}
	\end{figure}

\subsection{SBPS: traceplots and ACFs (cont.)} 

In \cref{fig-pdmp-trace-d}, we have studied SBPS by the traceplots and ACFs for the multivariate student's $t$ target with DoF$=d=100$ for low refresh rate. In this example, we study other settings. \cref{fig-pdmp-trace-d-2} shows the traceplots and ACFs for multivariate student's $t$ target with large refresh rate), \cref{fig-pdmp-trace-inf} and \cref{fig-pdmp-trace-inf-2} are those for standard Gaussian target (with $d=100$) using small and large refresh rates, respectively. We consider three cases: the $1$-st coordinate, the negative log density, and the squared $1$-st coordinate, respectively. Since SBPS is a continuous time process, for the traceplot and ACF, we further discretize every unit of time into $5$ samples. $N=1000$ events are simulated. A low refresh rate corresponds to $0.2$ and a high refresh rate corresponds to $2$.

From \cref{fig-pdmp-trace-inf}, there are several similar interesting observations as in \cref{fig-pdmp-trace-d}. The ACFs for all three cases have certain periodic behaviors and can take negative values for the first two cases. As a result of negative ACFs, the ESS per Switch is larger than $1$ in the first two cases. For high refresh rate cases such as in \cref{fig-pdmp-trace-d-2} and \cref{fig-pdmp-trace-inf-2}, the periodic behavior and negativity of ACFs disappear and the corresponding ESS per Switch is smaller than $1$, (note that the performance for SBPS is always significantly better than the traditional BPS in the simulations). This suggests that a larger proportion of bounce events are helpful since bounce events use the information of the gradients of log-target density. As the refreshment events for SBPS are very efficient in the transient phase, the refresh rate can be chosen to be very low. This is not the case for BPS. The refreshments in the transient phase for BPS are inefficient in high dimensions.

		\begin{figure}[h]
		\centering
		\includegraphics[width=\textwidth]{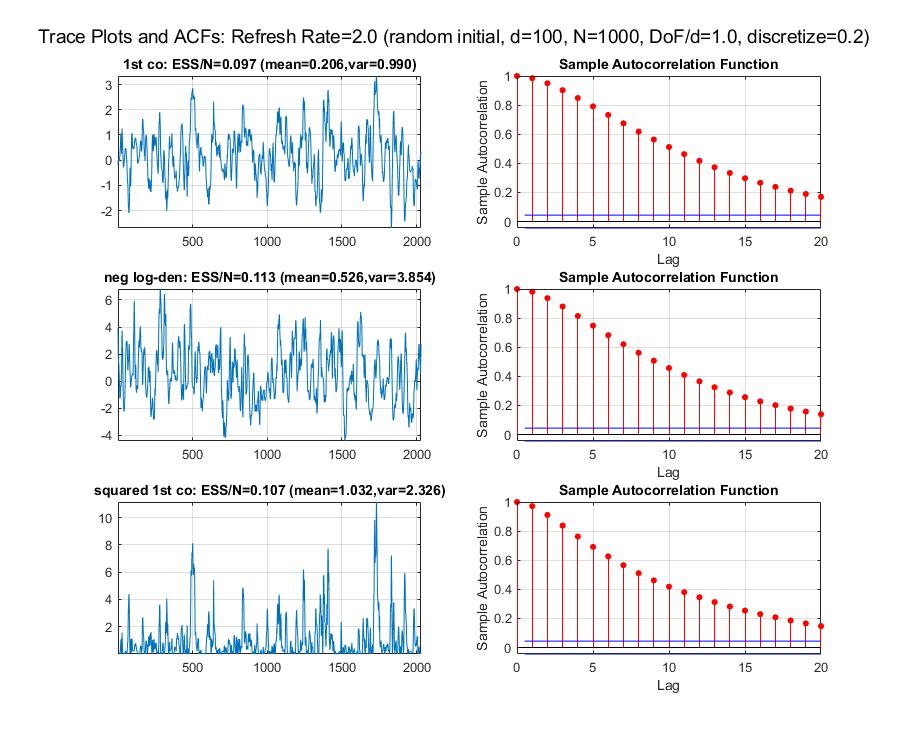} 
		\caption{Trace plots and ACFs for the $1$-st coordinate, the negative log-density, and the squared $1$-st coordinate. Target distribution is multivariate student's $t$ with DoF$=d=100$. Every unit period is discretized into $5$ samples. $N=1000$ events are simulated. High refresh rate $=2$. 
		}
		\label{fig-pdmp-trace-d-2}
	\end{figure}
	
	\begin{figure}[h]
		\centering
		\includegraphics[width=\textwidth]{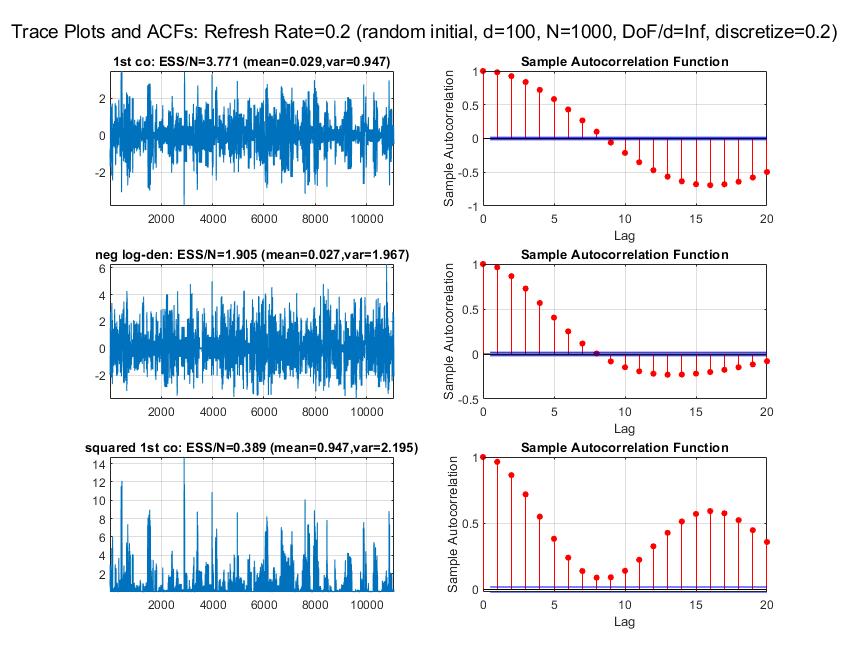}
		\caption{Trace plots and ACFs for the $1$-st coordinate, the negative log-density, and the squared $1$-st coordinate. The target distribution is standard Gaussian with $d=100$. Every unit period is discretized into $5$ samples. $N=1000$ events are simulated. Low refresh rate $=0.2$. 
		}
		\label{fig-pdmp-trace-inf}
	\end{figure}
		\begin{figure}[h]
		\centering
		\includegraphics[width=\textwidth]{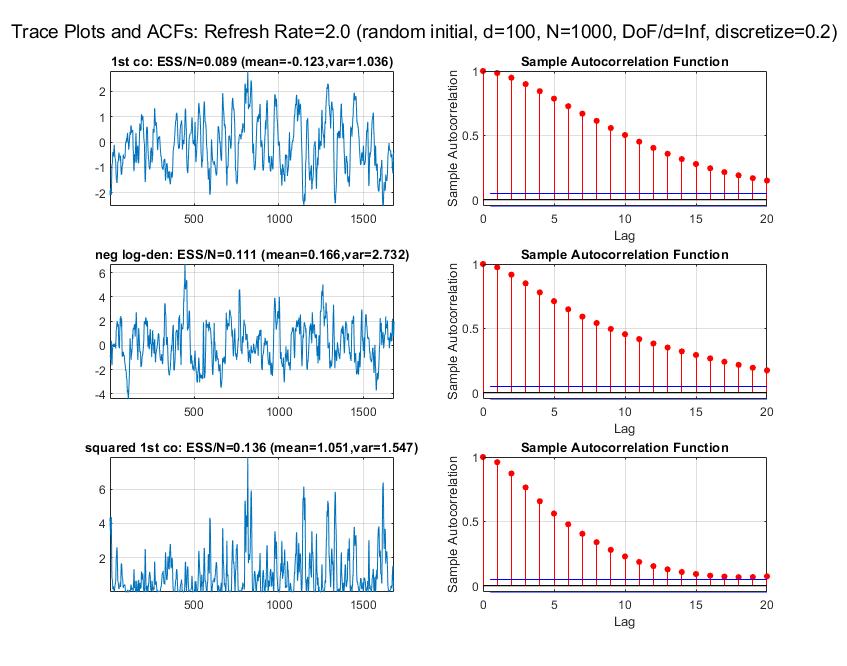} 
		\caption{Trace plots and ACFs for the $1$-st coordinate, the negative log-density, and the squared $1$-st coordinate. The target distribution is standard Gaussian with $d=100$. Every unit period is discretized into $5$ samples. $N=1000$ events are simulated. High refresh rate $=2$. 
		}
		\label{fig-pdmp-trace-inf-2}
	\end{figure}

\subsection{SBPS: ESS per Switch for multivariate student's $t$ target} 

 We have studied the efficiency curves of SBPS and BPS in terms of ESS per Switch versus the refresh rate for the Gaussian target in \cref{fig-pdmp-ESSpS}. In this example, we consider the case of the multivariate student's $t$ target with DoF equals $d$. The first subplot of \cref{fig-pdmp-ESSpS-2} contains the proportion of refreshments in all the $N$ events for varying refresh rates. Note that there are no bounce events for the SBPS so all the events for SBPS are refreshments. In the other three subplots of \cref{fig-pdmp-ESSpS-2}, we plot the logarithm of ESS per Switch as a function of the refresh rate for three cases, the $1$-st coordinate, the negative log-density, and the squared $1$-st coordinate. For each efficiency curve, $N=1000$ events are simulated, random initial value for SBPS and BPS starts from stationarity. As SBPS and BPS are continuous-time processes, each unit time is discretized into $5$ samples. 

According to \cref{fig-pdmp-ESSpS-2}, the ESS per Switch of SBPS is much larger than the ESS per Switch of BPS for all cases (actually the gap is larger than the cases for Gaussian target, and the gap becomes larger in higher dimensions). For all three cases, the ESS per Switch of SBPS can be larger than $1$ if the refresh rate is relatively low, this is also better than the cases for the Gaussian target. For BPS, however, even starting from stationarity, the ESS per Switch is always much smaller than $1$.

		\begin{figure}[h]
		\centering
		\includegraphics[width=0.9\textwidth]{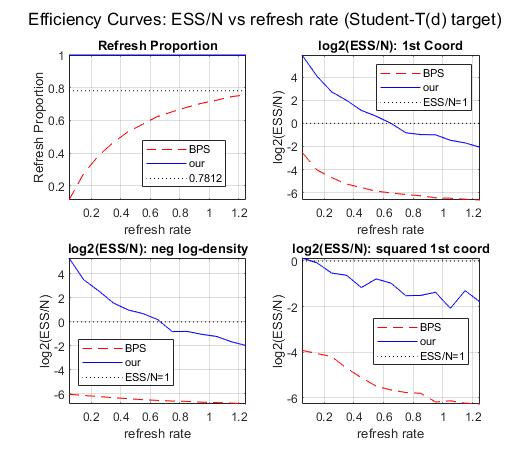} 
		\caption{Efficiency: ESS per Switch for SBPS and BPS for varying refresh rate.  $N=1000$ events are simulated, Random initial values for SBPS and BPS start from stationarity. Each unit time is discretized to $5$ samples. Target distribution: Multivariate student's $t$ with DoF$=d=100$. The first subplot is the proportion of refreshment events in all $N$ events. The other three subplots are ESS for the $1$-st coordinate, the negative log density, and the squared $1$-st coordinate, respectively.
		}
		\label{fig-pdmp-ESSpS-2}
	\end{figure}

\subsection{SBPS: robustness to target variance}

Previously, in \cref{fig-pdmp-trace-inf} and \cref{fig-pdmp-trace-inf-2}, we have studied the traceplots and ACFs of SBPS for standard Gaussian target with $d=100$ using small and large refresh rates, respectively. In this example, we repeat the same numerical experiments except that, instead of considering standard Gaussian targets, we consider 
Gaussian targets with larger or smaller variances.
Since we still choose $R=\sqrt{d}$, this is equivalent to studying the robustness of SBPS to the choice of the radius of the sphere. Other than the target variances, numerical experiment settings are the same as in \cref{fig-pdmp-trace-inf} and \cref{fig-pdmp-trace-inf-2}.

We consider two cases, Gaussian target $\mathcal{N}(0,1.5 I_d)$ (see \cref{fig-pdmp-trace-large-var} for small refresh rate and \cref{fig-pdmp-trace-large-var-2} for large refresh rate) and Gaussian target $\mathcal{N}(0,0.7I_d)$ (see \cref{fig-pdmp-trace-small-var} for small refresh rate and \cref{fig-pdmp-trace-small-var-2} for large refresh rate). Comparing with the cases when the variance is $1$ in  \cref{fig-pdmp-trace-inf} and \cref{fig-pdmp-trace-inf-2}, we can see that the traceplots and ACFs for the first coordinate and squared first coordinate are not affected too much. For the negative log-target density, the traceplots and ACFs behave quite differently. When the target covariance matrix is $1.5I_d$, more bounce events occur when the SBPS is moving to the equator. As a result, the SBPS stays almost all the time in the ``Northern Hemisphere''. When the target covariance matrix is $0.7I_d$, the situation is exactly the opposite: the SBPS stays almost all the time in the ``Southern Hemisphere''. The traceplots and ACFs in all four figures show that the performance of SBPS is quite robust to the target variance.

	\begin{figure}[h]
		\centering
		\includegraphics[width=\textwidth]{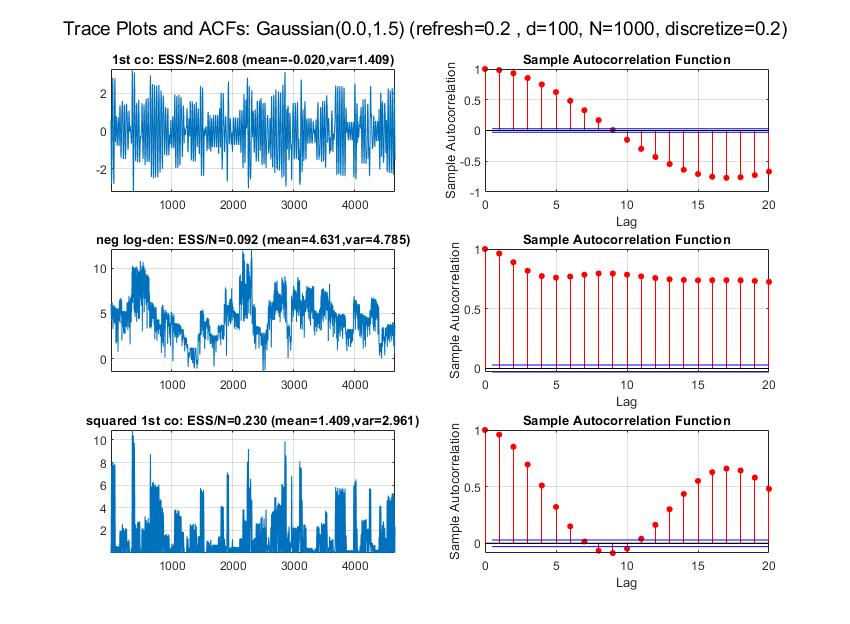}
		\caption{Trace plots and ACFs for the $1$-st coordinate, the negative log-density, and the squared $1$-st coordinate. Target distribution is $\mathcal{N}(0,1.5 I_d)$ with $d=100$ and we choose $R=d^{1/2}$. Every unit period is discretized into $5$ samples. $N=1000$ events are simulated. Low refresh rate $=0.2$.
		}
		\label{fig-pdmp-trace-large-var}
	\end{figure}
		\begin{figure}[h]
		\centering
		\includegraphics[width=\textwidth]{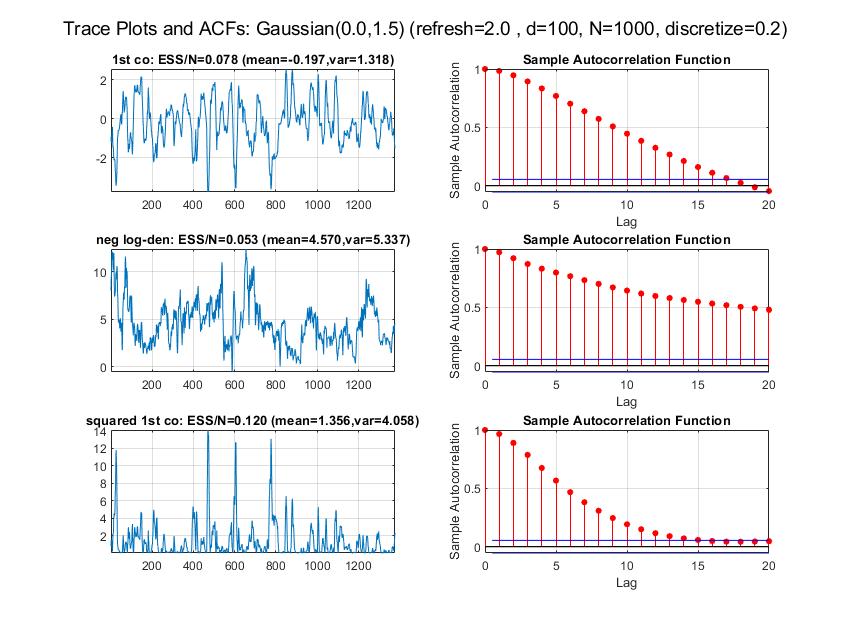} 
		\caption{Trace plots and ACFs for the $1$-st coordinate, the negative log-density, and the squared $1$-st coordinate. Target distribution is $\mathcal{N}(0,1.5 I_d)$ with $d=100$ and we choose $R=d^{1/2}$. Every unit time is discretized into $5$ samples. $N=1000$ events are simulated. High refresh rate $=2$.
		}
		\label{fig-pdmp-trace-large-var-2}
	\end{figure}
		\begin{figure}[h]
		\centering
		\includegraphics[width=\textwidth]{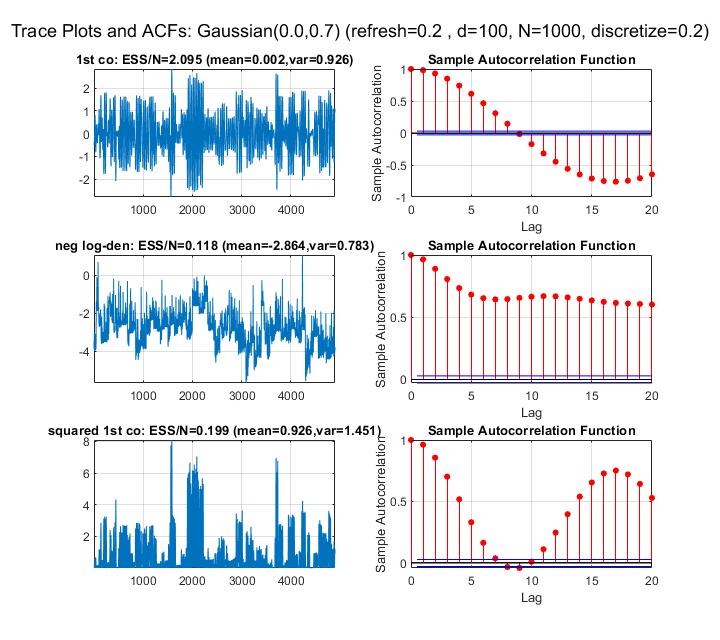}
		\caption{Trace plots and ACFs for the $1$-st coordinate, the negative log-density, and the squared $1$-st coordinate. Target distribution is $\mathcal{N}(0, 0.7 I_d)$ with $d=100$ and we choose $R=d^{1/2}$. Every unit of time is discretized into $5$ samples. $N=1000$ events are simulated. Low refresh rate $=0.2$.  
		}
		\label{fig-pdmp-trace-small-var}
	\end{figure}
			\begin{figure}[h]
		\centering
		\includegraphics[width=\textwidth]{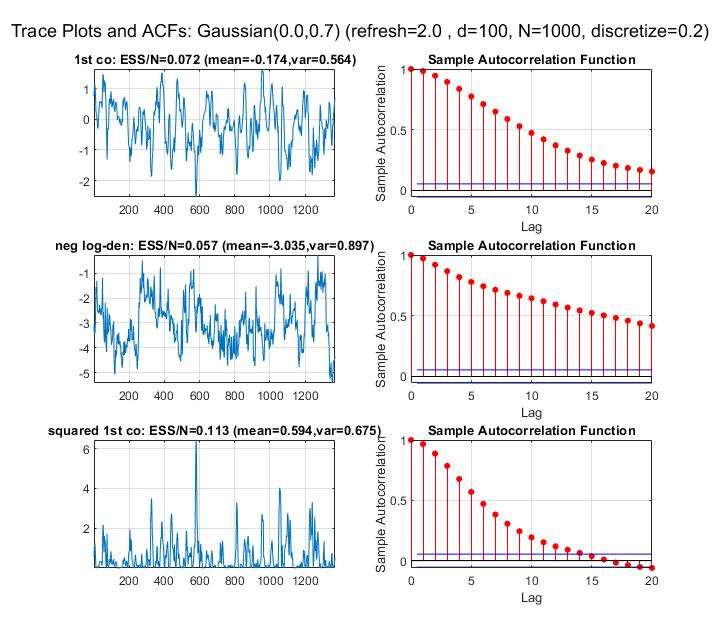} 
		\caption{Trace plots and ACFs for the $1$-st coordinate, the negative log-density, and the squared $1$-st coordinate. Target distribution is $\mathcal{N}(0, 0.7 I_d)$ with $d=100$ and we choose $R=d^{1/2}$. Every unit period is discretized into $5$ samples. $N=1000$ events are simulated. High refresh rate $=2$.  
		}
		\label{fig-pdmp-trace-small-var-2}
	\end{figure}

\clearpage
\bibliographystyle{plainnat}
\bibliography{mcmc2.bib}
\end{document}